\documentclass[conference,final]{IEEEtran}
\IEEEoverridecommandlockouts
\usepackage{amsmath,amssymb,amsfonts}
\usepackage{algorithmic}
\usepackage{graphicx}
\usepackage{textcomp}
\usepackage{amsthm}
\usepackage{mathtools}
\usepackage{semantic} 
\usepackage[export]{adjustbox}
\usepackage{mwe}
\def\BibTeX{{\rm B\kern-.05em{\sc i\kern-.025em b}\kern-.08em
    T\kern-.1667em\lower.7ex\hbox{E}\kern-.125emX}}

\newcommand{\resetCurThmBraces}{%
\gdef\curThmBraceOpen{(}%
\gdef\curThmBraceClose{)}}
\resetCurThmBraces
\newcommand{\removeThmBraces}{%
\gdef\curThmBraceOpen{}%
\gdef\curThmBraceClose{}}

\newenvironment{notheorembrackets}{\removeThmBraces}{\resetCurThmBraces}
\usepackage{etoolbox}
\patchcmd{\thmhead}{(#3)}{\curThmBraceOpen #3\curThmBraceClose }{}{}

\usepackage{hyperref}
\hypersetup{hidelinks,final}
\usepackage[capitalize,nameinlink]{cleveref} 

\usepackage{bm} 

\newcommand{\mybar}[3]{%
  \mathrlap{\hspace{#2}\overline{\scalebox{#1}[1]{\phantom{\ensuremath{#3}}}}}\ensuremath{#3}
}

\newcommand{\barF}{\mybar{0.6}{2.5pt}{F}}
\newcommand{\barB}{\mybar{0.6}{2.5pt}{B}}
\newcommand{\barBs}{\mybar{0.6}{1.8pt}{B}}
\newcommand{\barPow}{\mybar{0.7}{1.5pt}{\Pow}}
\newcommand{\barSigmas}{\mybar{0.9}{0pt}{\Sigmas}}
\newcommand{\barrho}{\mybar{0.9}{1pt}{\rho}}

\providecommand{\catname}{\mathbf} 
\providecommand{\clsname}{\mathcal}
\providecommand{\oname}[1]{{\operatorname{\mathsf{#1}}}}

\def\defcatname#1{\expandafter\def\csname B#1\endcsname{\catname{#1}}}
\def\defcatnames#1{\ifx#1\defcatnames\else\defcatname#1\expandafter\defcatnames\fi}
\defcatnames ABCDEFGHIJKLMNOPQRSTUVWXYZ\defcatnames

\def\defclsname#1{\expandafter\def\csname C#1\endcsname{\clsname{#1}}}
\def\defclsnames#1{\ifx#1\defclsnames\else\defclsname#1\expandafter\defclsnames\fi}
\defclsnames ABCDEFGHIJKLMNOPQRSTUVWXYZ\defclsnames

\def\defbbname#1{\expandafter\def\csname BB#1\endcsname{{\bm{\mathsf{#1}}}}}
\def\defbbnames#1{\ifx#1\defbbnames\else\defbbname#1\expandafter\defbbnames\fi}
\defbbnames ABCDEFGHIJKLMNOPQRSTUVWXYZ\defbbnames

\def\Set{\catname{Set}}

\providecommand{\argument}{\operatorname{-\!-}}

\DeclareOldFontCommand{\bf}{\normalfont\bfseries}{\mathbf}

\providecommand{\Id}{\operatorname{Id}}

\providecommand{\id}{\mathsf{id}}
\providecommand{\op}{\mathsf{op}}
\providecommand{\comp}{\mathbin{\circ}}

\providecommand{\xto}[1]{\,\xrightarrow{#1}\,}

\providecommand{\To}{\mathrel{\Rightarrow}}			           

\providecommand{\dar}{\kern-1.2pt\operatorname{\downarrow}}	
\providecommand{\uar}{\kern-1.2pt\operatorname{\uparrow}}	

\providecommand{\fst}{\oname{fst}}
\providecommand{\snd}{\oname{snd}}

\providecommand{\inl}{\oname{inl}}
\providecommand{\inr}{\oname{inr}}

\DeclareSymbolFont{Symbols}{OMS}{cmsy}{m}{n}
\DeclareMathSymbol{\iobj}{\mathord}{Symbols}{"3B}

\providecommand{\ev}{\oname{ev}}

\usepackage{stmaryrd}

\providecommand{\pacman}[1]{}					                     

\newcommand{\undefine}[1]{\let #1\relax}					                       

\providecommand{\mone}{{\text{\kern.5pt\rmfamily-}\mathsf{\kern-.5pt1}}}

\providecommand{\smin}{\smallsetminus}

\makeatletter
\def\mfix#1{\oname{#1}\@ifnextchar\bgroup\@mfix{}}	       
\def\@mfix#1{#1\@ifnextchar\bgroup\mfix{}}			           
\makeatother

\providecommand{\case}[3]{\mfix{case}{\mathbin{}#1}{of}{#2}{\kern-1pt;}{\mathbin{}#3}}

\usepackage{dsfont}

\DeclareMathSymbol{\mathinvertedexclamationmark}{\mathord}{operators}{'074}
\DeclareMathSymbol{\mathexclamationmark}{\mathord}{operators}{'041}
\makeatletter
\newcommand{\raisedmathinvertedexclamationmark}{%
  \mathord{\mathpalette\raised@mathinvertedexclamationmark\relax}%
}
\newcommand{\raised@mathinvertedexclamationmark}[2]{%
  \raisebox{\depth}{$\m@th#1\mathinvertedexclamationmark$}%
}
\makeatother

\newcommand{\product}{\times}
\newcommand{\wt}{\widetilde}

\newcommand{\Pt}{V}
\newcommand{\R}{\mathcal{R}}
\newcommand{\st}{\mathsf{st}}

\newcommand{\Pow}{\mathcal{P}}
\newcommand{\xTo}{\xRightarrow}

\newcommand{\primrec}{\mathsf{pr}\,}

\newcommand{\lft}{{\mathsf{l}}}
\newcommand{\rgt}{{\mathsf{r}}}

\newcommand{\smc}{\mathbin{;}}
\newcommand{\hatR}{{\widehat{R}}}
\newcommand{\qand}{\quad\text{and}\quad}
\newcommand{\qqand}{\qquad\text{and}\qquad}

\newcommand{\deq}[1]{{\rotatebox{#1}{$=$}}}
\newcommand{\dleq}[1]{{\rotatebox{#1}{$\preceq$}}}
\newcommand{\dgeq}[1]{{\rotatebox{#1}{$\succeq$}}}
\newcommand{\cev}[1]{\reflectbox{\ensuremath{\vec{\reflectbox{\ensuremath{#1}}}}}}

\newcommand{\HO}{\mathcal{HO}}

\newcommand{\Gra}{\mathbf{Gra}}

\newcommand{\happrox}{\mathrel{\widehat\lesssim}}

\newcommand{\under}[1]{\lvert#1\rvert}
\newcommand{\SKI}{\mathrm{SKI}}

\newcommand{\ap}{{\mathrm{ap}}}
\newcommand{\var}{\mathsf{var}}

\newcommand{\Sigmas}{\Sigma^{\star}}

\newcommand{\ar}{\mathsf{ar}}

\newcommand{\epito}{\twoheadrightarrow}

\newcommand{\can}{\mathsf{can}}

\newcommand{\seq}{\subseteq}
\newcommand{\ol}{\overline}
\newcommand{\out}{\mathsf{out}}
\newcommand{\outl}{\mathsf{outl}}
\newcommand{\outr}{\mathsf{outr}}

\providecommand{\C}{}
\providecommand{\D}{}

\renewcommand{\C}{{\mathbb{C}}}
\renewcommand{\D}{{\mathbb{D}}}

\renewcommand{\id}{{\mathsf{id}}}
\newcommand{\Nat}{\mathds{N}}

\newcommand{\Rel}{\mathbf{Rel}}

\newcommand{\goes}[2]{\ensuremath{#1 \rightarrow #2}}

\newcommand{\goesv}[3]{\ensuremath{#1 \xrightarrow{~#3~} #2}}

\newcommand{\f}{\oname{f}}

\newcommand{\takeout}[1]{\empty}

\newcommand{\ini}{\iota}

\newcommand{\wh}{\widehat}

\DeclareMathOperator{\Alg}{\mathbf{Alg}}

\renewcommand{\rho}{\varrho}

\newcommand{\opp}{\mathsf{op}}

\newcommand{\pullbackangle}[2][]{\arrow[phantom,to path={
                     -- ($ (\tikztostart)!1cm!#2:([xshift=8cm]\tikztostart) $)
                        node[anchor=west,pos=0.0,rotate=#2,
                        inner xsep = 0]
                        {\begin{tikzpicture}[minimum
                        height=1mm,baseline=0,#1]
    \draw[-] (0,0) -- (.5em,.5em) -- (0,1em);
                        \end{tikzpicture}}}]{}}

\makeatletter
\newsavebox{\@brx}
\newcommand{\llangle}[1][]{\savebox{\@brx}{\(\m@th{#1\langle}\)}%
  \mathopen{\copy\@brx\kern-0.5\wd\@brx\usebox{\@brx}}}
\newcommand{\rrangle}[1][]{\savebox{\@brx}{\(\m@th{#1\rangle}\)}%
  \mathclose{\copy\@brx\kern-0.5\wd\@brx\usebox{\@brx}}}
\makeatother

\renewcommand{\comp}{\cdot}
\renewcommand{\c}{\colon}

\newcommand{\xra}[1]{\xrightarrow{~#1~}}
\renewcommand{\xto}{\xra}
\let\xmpsto=\xmapsto
\renewcommand{\xmapsto}[1]{\xmpsto{~#1~}}

\newcommand{\V}{\mathcal{V}}
\newcommand{\monoto}{\rightarrowtail}
\newcommand{\subto}{\hookrightarrow}
\newcommand{\commu}{\ensuremath{\circlearrowleft}}
\renewcommand{\Nat}{\mathbb{N}}

\newcommand{\fset}{{\mathbb{F}}}
\newcommand{\vcat}{{\Set}^{\fset}}

\newcommand{\gcat}{\mathbb{C}}

\newcommand{\mS}{{\mu\Sigma}}

\renewcommand{\epsilon}{\varepsilon}

\newcommand{\app}{\,}

\usepackage{commath}

\usepackage{enumitem}
\setlist[enumerate,1]{label=(\arabic*),font=\normalfont,align=left,leftmargin=0pt,labelindent=0pt,listparindent=\parindent,labelwidth=0pt,itemindent=!,topsep=3pt,parsep=0pt,itemsep=3pt,start=1}
\setlist[enumerate,2]{label=(\alph*),,font=\normalfont,align=left,leftmargin=0pt,labelindent=0pt,listparindent=\parindent,labelwidth=0pt,itemindent=!,topsep=3pt,parsep=0pt,itemsep=3pt,start=1}
\setlist[itemize]{labelindent=*,leftmargin=*}
\setlist[description]{labelindent=*,leftmargin=*,itemindent=-1 em}

\usepackage{ifdraft}

\usepackage[inline,final]{showlabels}
\usepackage{seqsplit}
\usepackage{xstring}
\usepackage{xcolor}
\usepackage{bbding}

 

\usepackage{graphicx,csquotes}

\usepackage{tikz}
\usepackage{tikz-cd}

\tikzstyle{shiftarr}=[
        rounded corners,%
        to path={--([#1]\tikztostart.center)
                     -- ([#1]\tikztotarget.center) \tikztonodes
                     -- (\tikztotarget)},
]

\tikzset{
    commutative diagrams/.cd,
    diagrams={>=stealth},
    row sep=large,
    column sep = huge
}

\usetikzlibrary{decorations.pathmorphing}
\usetikzlibrary{cd,calc,positioning,automata,arrows,shapes}
\ifdraft{%
  \makeatletter
  \setcounter{tocdepth}{2}
  \makeatother
}{}

\tikzcdset{scale cd/.style={every label/.append style={scale=#1},
    cells={nodes={scale=#1}}}}

\newcommand{\descto}[3][]{\arrow[phantom]{#2}[#1]{\text{\footnotesize{}#3}}}

\usepackage[footnote,marginclue,nomargin,final]{fixme}
\FXRegisterAuthor{hu}{ahu}{HU} 
\FXRegisterAuthor{sm}{asm}{SM} 
\FXRegisterAuthor{st}{ast}{ST} 
\FXRegisterAuthor{ls}{als}{LS} 
\FXRegisterAuthor{sg}{asg}{SG} 

\usepackage{xspace}

\usepackage{accents}
\newcommand{\dbtilde}[1]{\accentset{\approx}{#1}}

\numberwithin{equation}{section}

\newtheorem{theorem}{Theorem}[section]
\newtheorem{lemma}[theorem]{Lemma}
\newtheorem{proposition}[theorem]{Proposition}
\newtheorem{corollary}[theorem]{Corollary}

\theoremstyle{definition}
\newtheorem{definition}[theorem]{Definition}
\newtheorem{example}[theorem]{Example}
\newtheorem{remark}[theorem]{Remark}
\newtheorem{construction}[theorem]{Construction}
\newtheorem{rem}[theorem]{Remark} 
\newtheorem{notation}[theorem]{Notation} 
\newtheorem{assumptions}[theorem]{Assumptions} 

\begin{document}

\title{Weak Similarity in Higher-Order Mathematical Operational Semantics
\thanks{Henning Urbat and Stefan Milius acknowledge support by the Deutsche Forschungsgemeinschaft (DFG, German
  Research Foundation) -- project number 470467389. Stelios Tsampas and Lutz Schröder acknowledge support by the Deutsche Forschungsgemeinschaft (DFG, German
  Research Foundation) -- project number 419850228.
}
}

\author{\IEEEauthorblockN{Henning Urbat, Stelios Tsampas, Sergey Goncharov, Stefan Milius, Lutz Schröder}
\IEEEauthorblockA{\textit{Friedrich-Alexander-Universität Erlangen-Nürnberg} \\
$\{$henning.urbat, stelios.tsampas, sergey.goncharov, stefan.milius, lutz.schroeder$\}$@fau.de}
}

\IEEEoverridecommandlockouts
\IEEEpubid{\makebox[\columnwidth]{} \hspace{\columnsep}\makebox[\columnwidth]{ }}

\maketitle

\begin{abstract}
Higher-order abstract GSOS is a recent extension of Turi and Plotkin's framework of Mathematical Operational Semantics to higher-order languages. The fundamental well-behavedness property of all specifications within the framework is that coalgebraic strong (bi)similarity on their operational model is a congruence. In the present work, we establish  a corresponding congruence theorem for \emph{weak} similarity, which is shown to instantiate to well-known concepts such as Abramsky's applicative similarity for the $\lambda$-calculus. On the way, we develop several techniques of independent interest at the level of abstract categories, including relation liftings of mixed-variance bifunctors and higher-order GSOS laws, as well as Howe's method.
\end{abstract}


\section{Introduction}
\label{sec:intro}

\noindent Following the emergence of structural approaches to operational semantics
(SOS), e.g.~\cite{DBLP:conf/stacs/Kahn87,DBLP:journals/jlp/Plotkin04a}, operational
reasoning has developed into a widely used methodology in formal
reasoning on higher-order languages. Numerous powerful operational techniques have been developed, tested, and
refined, such as logical relations~\cite{tait1967intensional,
  DBLP:journals/iandc/Statman85,DBLP:journals/iandc/OHearnR95,
  DBLP:journals/corr/abs-1103-0510} and Howe's method~\cite{DBLP:conf/lics/Howe89,
  DBLP:journals/iandc/Howe96,DBLP:conf/lics/LagoGL17}. These methods have been
found to be quite robust, being capable
of providing solutions
to challenging problems such as congruence proofs and reasoning about contextual
equivalence, even in rather involved settings such as effectful,
e.g.\ nondeterministic, higher-order languages.

Unfortunately, such power comes at a price.  Operational methods are
known to be both complex, requiring a daunting amount of machinery in
order to be instantiated, and specialized, in the sense that they need
to be developed on a per-case basis, and any small perturbation in the
problem setting may break earlier machinery. A key ingredient that is needed to alleviate these issues is a sufficiently general
rigorous notion of \emph{SOS specification of programming language semantics}; without
it, reasoning is inevitably bound to specific instances of SOS
specifications, and the only `free' mathematical principle is
induction on the structure of terms.  Capturing the essence of SOS in
a single, precise definition in order to reason at a greater level of
generality has thus been a topic of lasting interest. \emph{Rule
  formats} such as GSOS~\cite{DBLP:journals/jacm/BloomIM95} provide a handle to reason about classes of languages, as opposed to
one language at a time. For instance, the property that bisimilarity
is a congruence holds for any language adhering to the GSOS format. On
a more abstract and conceptual level, Turi and Plotkin's
framework of Mathematical Operational
Semantics~\cite{DBLP:conf/lics/TuriP97}, a.k.a.\ \emph{abstract
  GSOS}, shows that rule formats such as GSOS are instances of a
general principle, namely that operational rules amount to certain natural transformations, so-called \emph{GSOS laws}.  Abstract GSOS has been 
instantiated in quite diverse settings~\cite{56f40c248cb44359beb3c28c3263838e,
  DBLP:conf/fossacs/KlinS08, DBLP:conf/lics/FioreS06,
  DBLP:journals/tcs/MiculanP16, DBLP:conf/fscd/0001MS0U22}. 


In recent work~\cite{gmstu23} we have reconciled Turi and
Plotkin's ideas, originally applicable only in first-order
settings, with higher-order languages. The main insight  is that
\emph{dinatural} transformations are able to express higher-order
operational rules in ways that the original approach based on
naturality could not. Like a classical GSOS
law, a higher-order GSOS law is a form of distributive law of a syntax
functor~$\Sigma$ over a behaviour functor~$B$, but in the context of
higher-order languages,~$B$ in general needs to be a mixed-variance
bifunctor, in the sense that it depends covariantly on the set of
states or terms when these appear as results of functions, and
contravariantly when they are used as arguments of functions. It is
this phenomenon of mixed variance that necessitates the use of
dinatural transformations.
The main result of~\cite{gmstu23} is that the operational semantics of a higher-order GSOS law is \emph{compositional}: for the initial (term)
model $\mS$, coalgebraic bisimilarity for the endofunctor
$B(\mS,\argument)$ is a congruence. For instance, in the case where $B(X,Y)$ is the
behaviour bifunctor for the $\lambda$-calculus, this instantiates to a
\emph{strong} variant of Abramsky's \emph{applicative
  bisimilarity}~\cite{Abramsky:lazylambda}, which unlike applicative
bisimilarity proper makes $\beta$-reductions observable.

The main contribution of the present paper is a generalization of our previous
congruence result~\cite{gmstu23} from strong bisimilarity to \emph{weak (bi)similarity}. It applies to higher-order
GSOS laws whose initial model forms a \emph{higher-order lax
  bialgebra}, extending the corresponding first-order
concept~\cite{DBLP:conf/concur/BonchiPPR15}. When instantiated to the call-by-name
$\lambda$-calculus, weak (bi)similarity amounts to standard
applicative (bi)similarity. Hence we obtain a more useful general
compositionality theorem, an instance of which is the classical result
that applicative bisimilarity (rather than a previously unstudied
notion of strong applicative bisimilarity as in~\cite{gmstu23}) in the
call-by-name $\lambda$-calculus is a
congruence~\cite{Abramsky:lazylambda}. Our approach is parameterized in
such a way that strong similarity is an instance of weak similarity, so
our main result subsumes that of~\cite{gmstu23}.

The passage from strong to weak similarity comes with a number of
technical challenges; most notably, simple and well-established
proof techniques such as coinduction up to congruence now fail. To prove our main theorem, we develop an
abstract categorical version of Howe's method (\cref{prop:howe}). The
abstraction depends centrally on new notions of bifunctorial graph and
relation liftings (applied to liftings of the mixed-variance behaviour
functor), which may in fact turn out to be of independent interest as
generalizations of relation liftings of
functors~\cite{hj98,DBLP:books/cu/J2016} to higher-order
behaviours.

For full proofs and additional details, see Appendix.

\paragraph*{Related Work}

Borthelle et al.~\cite{DBLP:conf/lics/BorthelleHL20} and Hirschowitz and
Lafont~\cite{DBLP:journals/lmcs/HirschowitzL22}
have recently developed a framework for congruence of applicative
bisimilarity based on Howe's method.  Their approach is conceptually quite
different from ours: operational rules are given
as endofunctors on a presheaf category of \emph{transition systems} over models
of a signature endofunctor, and the initial algebra for the rule
endofunctor represents the induced transition system for the given
semantics.

Dal Lago et al.~\cite{DBLP:conf/lics/LagoGL17} propose a
generalization of Howe's method for call-by-value $\lambda$-calculi
with algebraic effects, based on the theory of relators. Their notion
of a \emph{computational} $\lambda$-calculus is parametrized over a
signature $\Sigma$ and a monad $T$ on sets, representing syntax and
effects of the language. The operational semantics is given in
big-step form.

Bonchi et al.~\cite{DBLP:conf/concur/BonchiPPR15} employ lax
bialgebras to establish up-to techniques for weak bisimulations in the
context of (first-order) abstract GSOS. Besides the differences in
scope, two approaches diverge also in the way the are based on
relation liftings: Bonchi et al.\ lift endofunctors from
sets to preorders and further to up-closed relations, while we lift
bifunctors from an abstract category~$\C$ to relations over~$\C$, the
up-closure being replaced with the abstract
\emph{good-for-simulations} condition
(\Cref{def:good-for-simulations-relation}).

\section{Preliminaries}
\subsection{Category Theory}\label{sec:categories}
We assume familiarity with basic category theory. In the following we recall some relevant terminology and notation.

\paragraph*{Products and coproducts}
Given objects
$X_1, X_2$ in a category~$\C$, we write $X_1\times X_2$ for the
product and $\langle f_1, f_2\rangle\c X\to X_1\times X_2$ for the
pairing of morphisms $f_i\c X\to X_i$, $i=1,2$. We let
$X_1+X_2$ denote the coproduct, $\inl\c X_1\to X_1+X_2$ and
$\inr\c X_2\to X_1+X_2$ the injections, $[g_1,g_2]\c X_1+X_2\to X$ the
copairing of morphisms $g_i\colon X_i\to X$, $i=1,2$, and
$\nabla=[\id_X,\id_X]\colon X+X\to X$ the codiagonal.

\paragraph*{Locally distributive categories}
A category $\C$ is \emph{distributive} if it has finite products and
coproducts, and for each $X\in \C$ the endofunctor $X\times(-)$ on
$\C$ preserves finite coproducts. It is \emph{locally distributive} if
for each $X\in \C$ the slice category $\C/X$ is distributive. Recall
that $\C/X$ has as objects all pairs $(Y,p_Y)$ of an object $Y\in \C$
and a morphism $p_Y\c Y\to X$, and a morphism from $(Y,p_Y)$ to
$(Z,p_Z)$ is a morphism $f\c Y\to Z$ of $\C$ such that
$p_Y = p_Z\comp f$. The coslice category $X/\C$ is defined dually.

\begin{example}\label{ex:categories}
Examples of locally distributive categories include the category $\Set$ of sets and functions, the category $\Set^{\C}$ of presheaves on a small category $\C$ and natural transformations, and the categories of posets and monotone maps, nominal sets and equivariant maps, metric spaces and non-expansive maps. In fact, they are all \emph{lextensive}~\cite[Cor~4.9]{cbl93}.
\end{example}

\paragraph*{Algebras}
Given an endofunctor $F$ on a category $\gcat$, an \emph{$F$-algebra}
is a pair $(A,a)$ which consists of an object~$A$ (the \emph{carrier} of the algebra)
and a morphism $a\colon FA\to A$ (its \emph{structure}). A
\emph{morphism} from $(A,a)$ to an $F$-algebra $(B,b)$ is a morphism
$h\colon A\to B$ of~$\gcat$ such that $h\comp a = b\comp Fh$. Algebras
for $F$ and their morphisms form a category $\Alg(F)$, and an
\emph{initial} $F$-algebra is simply an initial object in that
category.  We denote the initial $F$-algebra by $\mu F$ if it exists,
and its structure by $\ini\colon F(\mu F) \to \mu F$. If
$\C$ has binary products, initial algebras entail a useful definition
principle known as \emph{primitive recursion}: for every morphism
$a\c F(\mu F\times A)\to A$ there exists a unique morphism
$\primrec a$ making the square below commute.
  \begin{equation}\label{eq:primitive-recursion}
    \begin{tikzcd}
      F(\mu F)
      \ar{r}{\iota}
      \ar{d}[swap]{F\langle \id,\, \primrec a\rangle}
      &
      \mu F
      \ar{d}{\primrec a}
      \\
      F(\mu F\times A) \ar{r}{a}
      &
      A
    \end{tikzcd}
\end{equation} 
More generally, a \emph{free $F$-algebra} on an object $X$ of $\C$ is an
$F$-algebra $(F^{\star}X,\iota_X)$ together with a morphism
$\eta_X\c X\to F^{\star}X$ of~$\C$ such that for every algebra $(A,a)$
and every morphism $h\colon X\to A$ in $\C$, there exists a unique
$F$-algebra morphism $h^\star\colon (F^{\star}X,\iota_X)\to (A,a)$
such that $h=h^\star\comp \eta_X$. If free algebras
exist on every object, their formation induces a monad
$F^{\star}\colon \C\to \C$, the \emph{free monad} generated by~$F$. (Conversely, in complete and well-powered categories, existence of a free monad implies existence of free algebras~\cite[Thm.~4.2.15]{manes76}.) For every $F$-algebra $(A,a)$, we obtain an
Eilenberg-Moore algebra $\wh{a} \colon F^{\star} A \to A$ as the free
extension of $\id_A\c A\to A$.


The most familiar example of functor algebras are algebras for a
signature.  An \emph{algebraic signature} consists of a set~$\Sigma$
of operation symbols together with a map $\ar\colon \Sigma\to \Nat$
associating to every $\f\in \Sigma$ its \emph{arity}
$\ar(\f)$. Symbols of arity $0$ are called \emph{constants}. Every
signature~$\Sigma$ induces the polynomial set functor
$\coprod_{\f\in\Sigma} (\argument)^{\ar(\f)}$, which we denote by the
same letter $\Sigma$. An algebra for the functor $\Sigma$ is
precisely an algebra for the signature $\Sigma$, viz.~a set $A$
equipped with an operation $\f^A\colon A^n \to A$ for every $n$-ary
operation symbol $\f\in \Sigma$. Morphisms of $\Sigma$-algebras are
maps respecting the algebraic structure.  Given a set $X$ of
variables, the free algebra $\Sigmas X$ is the $\Sigma$-algebra of
$\Sigma$-terms with variables from~$X$. In particular, the free
algebra on the empty set is the initial algebra $\mu \Sigma$; it is
formed by all \emph{closed terms} of the signature. For every
$\Sigma$-algebra $(A,a)$, the induced Eilenberg-Moore algebra
$\wh{a}\colon \Sigmas A \to A$ is given by the map evaluating terms
over $A$ in the algebra.

A relation ${R}\seq A\times A$ on a $\Sigma$-algebra
$A$ is called a \emph{congruence} if for every $n$-ary  $\f\in \Sigma$ and elements $R(a_i,a_i')$, $i=1,\ldots,n$, one has 
$R(\f^A(a_1,\ldots,a_n),\f^A(a_1',\ldots,a_n'))$.
Note that we do not require $R$ to be an equivalence relation.

\paragraph*{Coalgebras}
Dual to the notion of algebra, a \emph{coalgebra} for an
endofunctor $F$ on $\gcat$ is a pair $(C,c)$ of an object $C$ (the
\emph{carrier}) and a morphism $c\colon C\to FC$ (its
\emph{structure}).

\subsection{Higher-Order Abstract GSOS}\label{sec:abstract-gsos}
We review the core principles behind \emph{higher-order abstract GSOS}~\cite{gmstu23}, a categorical framework modelling the operational semantics of higher-order languages. 
It is parametric in
\begin{enumerate}
\item a category $\C$ with finite products and coproducts;
\item an object $\Pt\in \C$ of \emph{variables};
\item two functors $\Sigma\c\C \to \C$ and $B\c \C^\opp\times \C\to
  \C$, where $\Sigma=\Pt+\Sigma'$ for some
  functor $\Sigma'\colon \C \to \C$, and free $\Sigma$-algebras exist on every object (hence $\Sigma$ generates a free monad $\Sigmas$).
\end{enumerate}
Informally, the functors $\Sigma$ and $B$ represent the \emph{syntax}
and the \emph{behaviour} of a higher-order language. The initial
algebra $\mS$ is the object of programs, and the requirement that
$\Sigma=\Pt+\Sigma'$ asserts that variables are programs. An object of
$V/\C$, the coslice category of \emph{$\Pt$-pointed objects}, is
thought of as a set $X$ of programs with an embedding $p_X\c V\to X$
of the variables.

\begin{example}\label{ex:ho}
  A simple instantiation is given by $\Pt=\emptyset$, a polynomial functor
  $\Sigma$ and the bifunctor $B_0(X,Y)=Y+Y^X$ on $\Set$. A map
  $\gamma_0\c \mS\to \mS+\mS^\mS$, that is, a $B_0(\mS,-)$-coalgebra
  with carrier $\mS$, can be thought of as a description of the
  operational behaviour of deterministic higher-order programs: every
  program $p\in \mS$ either performs a silent computation step
  reducing $p$ to $\gamma(p)\in \mS$, or it acts as a function
  $\gamma(p)\in \mS^\mS$ mapping programs to programs.
\end{example}
In order to actually construct coalgebras $\gamma_0$ as in the above example, we use the following concept:

\begin{definition}\label{def:ho-gsos-law}
  A \emph{($\Pt$-pointed) higher-order GSOS law} of $\Sigma$ over $B$
  is a family of morphisms
  \begin{align}\label{eq:ho-gsos-law}
    \rho_{(X,p_X),Y} \c \Sigma (X \times B(X,Y))\to B(X, \Sigma^\star (X+Y))
  \end{align}
  dinatural in $(X,p_X)\in \Pt/\C$ and natural in $Y\in \C$.
\end{definition}

\begin{notation}\label{not:rho}
\begin{enumerate}
\item We usually write $\rho_{X,Y}$ for $\rho_{(X,p_X),Y}$, as the
  point $p_X\c V\to X$ will always be clear from the context.
\item For every $\Sigma $-algebra $(A,a)$, we regard $A$ as {$\Pt$-pointed} by
  \[p_A = \bigl(\Pt\xra{\inl} \Pt+\Sigma' A = \Sigma  A \xra{a} A\bigr).\] 
\end{enumerate}
\end{notation}

\begin{definition}\label{def:operational-model}
The \emph{operational model} of a higher-order GSOS law $\rho$ in \eqref{eq:ho-gsos-law} is the $B(\mS,-)$-coalgebra 
\begin{equation*}
\gamma\c \mS\to B(\mS,\mS)
\end{equation*}
obtained via primitive recursion as the unique morphism making the diagram \eqref{diag:gamma} in \Cref{fig:gamma} commute.
\begin{figure*}
\begin{equation}\label{diag:gamma}
\begin{tikzcd}[column sep=5em]
\Sigma(\mS) \ar{rrr}{\iota} \ar{d}[swap]{\Sigma\langle \id, \gamma\rangle} & & & \mS \ar{d}{\gamma} \\
\Sigma(\mS\times B(\mS,\mS)) \ar{r}{\rho_{\mS,\mS}} & B(\mS,\Sigmas(\mS+\mS)) \ar{r}{B(\id,\Sigmas \nabla)} & B(\mS,\Sigmas(\mS)) \ar{r}{B(\id,\hat\ini)} & B(\mS,\mS) 
\end{tikzcd}
\end{equation}
\caption{Operational model of a higher-order GSOS law}\label{fig:gamma}
\end{figure*}
Here we regard the initial algebra $\mu\Sigma$ as $V$-pointed as in \Cref{not:rho},
and $\wh{\ini}$ is the $\Sigmas$-algebra corresponding
to $\iota\c \Sigma(\mS)\to \mS$.
\end{definition}

\begin{rem}
The commutative diagram \eqref{diag:gamma} states that $(\mu\Sigma,\ini,\gamma)$ forms a \emph{bialgebra} for the higher-order GSOS law~$\rho$; in fact, it is the initial such bialgebra~\cite[Prop.~4.20]{gmstu23}. An important difference to first-order abstract GSOS~\cite{DBLP:conf/lics/TuriP97} is that a final bialgebra usually does not exist even for simple deterministic behaviour functors~\cite[Ex.~4.21]{gmstu23}. This in part explains why higher-order compositionality results are technically involved and first-order proof methods fail.
\end{rem}

Let us illustrate the above concepts in the setting of
\Cref{ex:ho}. A higher-order GSOS law of a polynomial functor $\Sigma$
over $B_0(X,Y)=Y+Y^X$ is a family of maps
\begin{equation*}
\rho^0_{X,Y}\c \Sigma(X\times (Y+Y^X))\to \Sigmas(X+Y)+(\Sigmas(X+Y))^X 
\end{equation*}
dinatural in $X\in \Set$ and natural in $Y\in \Set$. Intuitively, on
input $\f((p_1,b_1),\ldots,(p_n,b_n))$ for $\f\in \Sigma$, the map
$\rho^0_{X,Y}$ specifies the behaviour of the program
$\f(p_1,\ldots,p_n)$ in terms of the behaviours
$b_1,\ldots, b_n\in Y+Y^X$ of its subprograms
$p_1,\ldots,p_n$. (Di)naturality of $\rho^0$ ensures that the maps
$\rho^0_{X,Y}$ are parametrically polymorphic, that is, they do not
look into the structure of their arguments. This can be made formal
via the following syntactic representation of higher-order GSOS
laws. Fix metavariables $x$, $x_i$, $y_i$ and $y_i^z$ for $i\in \Nat$
and $z\in \{x, x_1,x_2,x_3,\ldots\}$. An \emph{$\HO$ rule} is an
expression of the form \eqref{eq:rule-red} or \eqref{eq:rule-non-red},
where $\f\in \Sigma$, $n=\ar(\f)$, $W\seq \{1,\ldots, n\}$,
$\ol{W}=\{1,\ldots,n\}\smin W$, and $t$ is a $\Sigma$-term in the
variables appearing in the premise, and additionally in $x$ for
\eqref{eq:rule-non-red}.
\begin{equation}\label{eq:rule-red}
  \inference{(x_j\to y_j)_{j\in
      W}\quad(\goesv{x_{k}}{y^{z}_{k}}{z})_{k\in\ol{W},\,z \in \{x_1,\ldots,x_n\}}}{\f(x_1,\ldots,x_{n})\to
    t}
\end{equation}
\begin{equation}\label{eq:rule-non-red}
  \inference{(x_j\to y_j)_{j\in W}\quad(\goesv{x_{k}}{y^{z}_{k}}{z})_{k\in\ol{W},\,z \in \{x_1,\ldots,x_n,x\}}}{\goesv{\f(x_1,\ldots,x_{n})}{t}{x}}
\end{equation}
 An \emph{$\HO$ specification} is a complete set $\R$ of $\HO$
  rules, that is, for
  each $n$-ary operation symbol $\f\in \Sigma$ and $W \seq
  \{1,\ldots,n\}$ there is exactly one rule of the form \eqref{eq:rule-red} or \eqref{eq:rule-non-red} in $\R$.

\begin{example}\label{ex:ski-calculus}
  The \emph{extended $\SKI$ calculus}, previously termed \emph{unary $\SKI$ calculus}~\cite{gmstu23}, is a combinatory
  logic expressively equivalent to Curry's \emph{$\SKI$
    calculus}~\cite{10.2307/2370619}, hence to the untyped
  $\lambda$-calculus. Its signature is given by
  $\Sigma=\{S/0,K/0,I/0,S'/1,K'/1,S''/2,\circ/2\}$ with arities as
  indicated. Informally, the operator $\argument \circ \argument$
  corresponds to function application (we write $s\,t$ for
  $s\circ t$), and the constants $S,K,I$ represent the functions
  $(s,t,u)\mapsto (s\app u)\app (t\app u)$, $(s,t)\mapsto s$, and
  $s\mapsto s$. The operators $S',S'',K'$ serve auxiliary
  purposes. The operational semantics is given by an $\HO$
  specification~\cite[Fig.~1]{gmstu23}. For instance, the rules for
  application are
  \begin{equation}\label{eq:ski-rules-app}
    \inference{x_1\to y_1}{x_1 \app x_2\to y_1 \app x_2}
    \qquad
    \inference{x_1\xto{x_2} x_1^{x_2}}{x_1 \app x_2\to x_1^{x_2}}
  \end{equation}
\end{example}
\begin{rem}\label{rem:missing-premises}
  By convention, a rule with incomplete premises represents the set of
  $\HO$ rules obtained by adding missing premises in every feasible
  way. For example, in the first rule of \eqref{eq:ski-rules-app} we
  can add $x_2\to y_2$, or $x_2\xto{x_1} y_2^{x_1}$ and
  $x_2\xto{x_2} y_2^{x_2}$.
\end{rem}
\begin{notheorembrackets}
\begin{proposition}[\cite{gmstu23}]
  Higher-order GSOS laws of $\Sigma$ over~$B_0$ correspond bijectively
  to $\HO$ specifications.
\end{proposition}
\end{notheorembrackets}
\noindent
The bijection is based on the Yoneda lemma, and maps an~$\HO$
specification $\mathcal{R}$ to the higher-order GSOS law $\rho^0$
defined as follows. Given $X,Y\in \Set$ and
\[w=\f((p_1,b_1),\ldots, (p_n,b_n))\in \Sigma(X\times B_0(X,Y)),\]
consider the unique rule in $\mathcal{R}$ matching $\f$ and
$W=\{ j\in \{1,\ldots,n\} : b_j\in Y \}$. If the rule is of the form
\eqref{eq:rule-red}, then
\[\rho^0_{X,Y}(w)\in \Sigmas(X+Y)\seq B_0(X,\Sigmas(X+Y)) \]
is the term obtained by taking the term $t$ in \eqref{eq:rule-red} and applying the following substitutions for $i\in \{1,\dots,n\}$, $j\in W$, $k\in \ol{W}$:
\[
  x_i\mapsto p_i,\qquad y_j\mapsto b_j,\qquad y_k^{x_i} \mapsto
  b_k(p_i).
\]
If the rule is of the form \eqref{eq:rule-non-red}, then 
\[
  \rho^0_{X,Y}(w)\in (\Sigmas(X+Y))^X\seq B_0(X,\Sigmas(X+Y))
\]
is the map $e\mapsto t_e$, where $t_e$ is obtained by taking the term
$t$ in~\eqref{eq:rule-non-red} and applying the above substitutions
along with
\[
  x\mapsto e
  \qquad\text{and}\qquad
  y_k^{x}\mapsto b_k(e)\quad(k\in\ol{W}).
\]
Instantiating \Cref{def:operational-model}, the operational
model of a higher-order GSOS law $\rho^0$ is the
$B_0(\mS,-)$-coalgebra
\begin{equation}\label{eq:operational-model-ho}
\gamma_0\c \mS\to \mS+\mS^\mS
\end{equation}
that runs programs in $\mS$ according to the rules in the corresponding $\HO$ specification.

\section{Compositionality for $\mathcal{HO}$ Specifications}\label{sec:cool}
\noindent Our eventual goal is to reason about weak simulations and their congruence properties on operational
models of higher-order GSOS laws. The required categorical
machinery is developed from \Cref{sec:graphs-relations-preorders} onwards. In the present section we motivate the categorical abstractions by again investigating the special case of $\HO$
specifications, that is, we continue to work in the setting of \Cref{ex:ho}.
\begin{notation}\label{notation:ho-format}
  \begin{enumerate}
  \item In addition to the polynomial functor $\Sigma$
    and~$B_0(X,Y)=Y+Y^X$, we will also consider the bifunctor
\[ B(X,Y)=\Pow B_0(X,Y)=\Pow(Y+Y^X)\c \Set^\opp \times \Set\to \Set, \]
where $\Pow\c \Set\to \Set$ is the powerset functor.
\item Given $X\in \Set$, a coalgebra $c\c C\to C+C^X$ for the functor $B_0(X,-)$ and $p\in C$, we write
\begin{align*}
p\to \ol{p} \qquad&\text{if}\qquad \text{$c(p)\in C$ and $\ol{p}=c(p)$},\\
p \not\to \qquad& \text{if}\qquad \text{$c(p)\not\in C$ (that is, $c(p)\in C^X$)},\\ 
p\xto{x} p_x \;\;\;\;\;&\text{if}\qquad \text{$c(p)\in C^X$, $x\in X$, and $p_x=c(p)(x)$}.  
\end{align*}
In the first case, we say that $p$ \emph{reduces}. Moreover, we put
\begin{align*}
  p\To \ol{p}
  \;\;\;\;&
  \text{if}\;\;\;\; \exists k\geq 0.\,\exists p_0,\ldots,p_k.\, p=
  p_0\to \cdots\to p_k=\ol{p},
  \\
  p \Downarrow \ol{p}
  \;\;\;\;&
  \text{if}\;\;\;\; p\To\ol{p} \text{ and } \ol{p}\not\to.
\end{align*}
The \emph{weak transition system} of $c\c C\to C+C^X$ is the coalgebra
$\wt{c}\c C\to \Pow(C+C^X)$ for the functor $B(X,-)$ where
\[
  \wt{c}(p)= \{\, \ol{p}\in C : p\To \ol{p} \,\} \,\cup\, \{\,
  c(\ol{p}) : p\Downarrow \ol{p}\,\}.
\] 
\end{enumerate}
\end{notation}
\begin{definition}\label{def:weak-sim}
  A \emph{weak simulation} on a $B_0(X,-)$-coalgebra $c\c C\to C+C^X$
  is a relation $R\seq C\times C$ such that for every $R(p,q)$ and
  $\ol{p}\in C$, the following conditions hold:
  \begin{align*}
    p\To \ol{p}
    \quad&\implies\quad
    \exists \ol{q}\in C.\, q\To\ol{q}\,\wedge\, R(\ol{p},\ol{q});
    \\
    p\Downarrow \ol{p}
    \quad&\implies \quad \exists \ol{q}\in C.\, q\Downarrow\ol{q}\,\wedge\, \forall x\in X.\, R(\ol{p}_x,\ol{q}_x).
  \end{align*}
  \emph{Weak similarity} is the greatest weak simulation on $(C,c)$,
  viz. the union of all weak simulations, denoted $\lesssim_{(C,c)}$
  or just $\lesssim$.
\end{definition}
Note that dropping the first condition leads to the same weak similarity relation. We include it to match the abstract view on weak simulations in \Cref{rem:weak-vs-strong}\ref{rem:weak-vs-strong-2} below.   
 

\begin{remark}\label{rem:weak-vs-strong}
  We make some observations that will be key to
  our categorical generalization of weak simulations in
  \Cref{sec:weak-simulations} and \ref{sec:weak-sim}.
\begin{enumerate}
\item\label{rem:weak-vs-strong-1} From a conceptual perspective, weak
  simulations can be understood in terms of \emph{relation liftings}
  of the involved functors. Let $\Rel$ denote the category whose
  objects are pairs $(X,R)$ of a set $X$ and a binary relation
  $R\seq X\times X$, and whose morphisms $h\c (X,R)\to (Y,S)$ are maps
  $h\c X\to Y$ such that $(h\times h)[R]\seq S$. The functors $\Pow$,
  $B_0$ and $B=\Pow\comp B_0$ lift to functors $\barPow$, $\barB_0$
  and $\barB$ on $\Rel$ making the diagram below commute, where
  $\under{-}$ is the forgetful functor $(X,R)\mapsto X$.
  \[ 
    \begin{tikzcd}
      \Rel^\opp\times \Rel
      \ar[shiftarr = {yshift=1.5em}]{rr}{\barBs}
      \ar{d}[swap]{\under{-}^\opp\times \under{-}}
      \ar{r}{\barBs_0}
      &
      \Rel \ar{d}{\under{-}}
      \ar{r}{\barPow}
      &
      \Rel \ar{d}{\under{-}}
      \\
      \Set^\opp \times \Set
      \ar[shiftarr = {yshift=-1.5em}]{rr}{B}
      \ar{r}{B_0}
      &
      \Set \ar{r}{\Pow}
      &
      \Set
    \end{tikzcd}
  \]
  \begin{enumerate}
  \item The lifting $\barPow$ of $\Pow$ is given by
    \[ \barPow(X,R)=(\Pow X, S_R), \qquad \barPow h=\Pow h,   \]
    where $S_R$ is the (one-sided) \emph{Egli-Milner relation} on $\Pow X$:
    \[
      S_R(U,V)
      \quad\iff\quad
      \forall u\in U.\ \exists v\in V.\,
      R(u,v).
    \]
  \item The lifting $\barB_0$ of $B_0$ is given by
    \begin{align*}
      \barB_0((X,R),(Y,S)) &= (B_0(X,Y),E^0_{R,S}),\\
      \barB_0(h,k) & = B_0(h,k),
    \end{align*}
    where $E^0_{R,S}(u,v)$ holds for $u,v\in B_0(X,Y)=Y+Y^X$ whenever
    either of the following conditions is satisfied:
    \begin{itemize}
    \item $u,v\in Y\, \wedge\, S(u,v)$;
    \item $u,v\in Y^X\, \wedge\, \forall x,x'\in X.\, (R(x,x')\implies S(u(x),v(x')))$.
    \end{itemize}
    We note that
    $((X,R),(Y,S))\mapsto (Y^X, E_{R,S}^0\cap Y^X\times Y^X)$ is the
    internal hom-functor of the cartesian closed category $\Rel$.
\item Finally, we put $\barB = \barPow\comp \barB_0$.
  More explicitly,
  \[
    \barB((X,R),(Y,S)) = (B(X,Y), E_{R,S}),\quad \barB(h,k)= B(h,k),
  \]
  where $E_{R,S}$ is the relation on $\Pow(Y+Y^X)$ defined as follows:
  \[ E_{R,S}(U,V)\quad\iff\quad \forall u\in U.\ \exists v\in V.\, E_{R,S}^0(u,v). \]
\end{enumerate}

\item\label{rem:weak-vs-strong-2} A relation $R\seq C\times C$ forms a
  weak simulation on the coalgebra $c\c C\to C+C^X$ iff there exists a
  map $\wt{c}_R$ making the diagram \eqref{eq:weak-sim-diag-1}
  commute, where $\outl$ and $\outr$ are the left and right
  projections and $\Delta_X\seq X\times X$ is the identity relation.
  \begin{equation}\label{eq:weak-sim-diag-1}
    \begin{tikzcd}[column sep=2.8em]
      C \ar{d}[swap]{\wt{c}} & R \ar{l}[swap]{\outl_R} \ar{d}{\wt{c}_R} \ar{r}{\outr_R} & C \ar{d}{\wt{c}} \\
      \Pow(C+C^X)  & \ar{l}[swap]{\outl_{\Delta_X,R}} E_{\Delta_X,R} \ar{r}{\outr_{\Delta_X,R}} & \Pow(C+C^X)
    \end{tikzcd}
  \end{equation}
\item\label{rem:weak-vs-strong-3} For a relation $R\seq C\times C$ to be a weak simulation, it suffices to restrict the premises of the two weak simulation conditions to strong transitions: for every $R(p,q)$ and $\ol{p}\in C$,
\begin{align*}
p\to \ol{p} \quad&\implies\quad \exists \ol{q}\in C.\, q\To\ol{q}\,\wedge\, R(\ol{p},\ol{q}); \\
p\not\to \quad&\implies\quad \exists \ol{q}\in C.\, q\Downarrow\ol{q}\,\wedge\, \forall x\in X.\, R({p}_x,\ol{q}_x).
\end{align*}
This amounts to the existence of a map $\wt{c}_R$ making the
diagram~\eqref{eq:weak-sim-diag-3} commute. Here we regard
$c\c C\to C+ C^X$ as a map $c\c C\to \Pow(C+C^X)$ by postcomposing
with $b\mapsto \{b\}$.
\begin{equation}\label{eq:weak-sim-diag-3}
\begin{tikzcd}[column sep=2.8em]
C \ar{d}[swap]{{c}} & R \ar{l}[swap]{\outl_R} \ar{d}{\wt{c}_R} \ar{r}{\outr_R} & C \ar{d}{\wt{c}} \\
\Pow(C+C^X)  & \ar{l}[swap]{\outl_{\Delta_X,R}} E_{\Delta_X,R} \ar{r}{\outr_{\Delta_X,R}} & \Pow(C+C^X)
\end{tikzcd}
\end{equation}

\end{enumerate}
\end{remark}

The compositionality theorem for $\HO$ specifications~\cite[Prop.~3.2]{gmstu23} asserts that strong similarity on the
operational model~\eqref{eq:operational-model-ho} is a
congruence with respect to the operations from the signature
$\Sigma$. (It is worth noting here that strong similarity coincides
with strong bisimilarity because reductions are
deterministic.) However, for weak similarity that result fails:

\begin{example}\label{ex:weak-bis-not-cong}
Consider the signature $\Sigma=\{c,d,u\}$ where $c,d$ are constants and $u$ is unary, along with the $\HO$ specification given by the following four rules:
\[   \inference{}{\goesv{c}{c}{x}}\; \inference{}{\goes{d}{c}} \; \inference{\goes{x_1}{y_1}}{\goes{u(x_1)}{u(x_1)}} \; \inference{\goesv{x_1}{y_1^{x_1}}{x_1}}{\goes{u(x_1)}{c}} \]
Then $c\lesssim d$ but $u(c)\not\lesssim u(d)$: we have $u(c)\Downarrow c$ while $u(d)\to u(d)\to u(d)\to \cdots$, which means that $u(d)\Downarrow t$ holds for no term $t$. Thus $\lesssim$ is not a congruence on the initial algebra $\mu\Sigma$.
\end{example}

This example illustrates that unrestricted $\HO$ rules are too liberal
for our purposes: They allow operators to behave completely
differently depending on whether a given subterm reduces or not, which
is clearly against the spirit of weak similarity where individual
reduction steps are meant to be unobservable. In the following we
devise a natural condition on $\HO$ specifications that avoids
congruence-breaking behaviour.

\begin{definition} Suppose that $\R$ is an $\HO$ specification. We say that a rule
 \eqref{eq:rule-red}/\eqref{eq:rule-non-red} of $\R$ is \emph{sound for weak transitions} if 
its corresponding \emph{weak rule} \eqref{eq:rule-red-weak}/\eqref{eq:rule-non-red-weak} shown below
\begin{equation}\label{eq:rule-red-weak}
  \inference{(x_j\To y_j)_{j\in
      W}\quad({x_{j}}\xTo{z}{y^{z}_{k}})_{k\in \ol{W},\,z \in \{x_1,\ldots,x_n\}}}{\f(x_1,\ldots,x_{n})\To
    t}
\end{equation}
\begin{equation}\label{eq:rule-non-red-weak}
  \inference{(x_j\To y_j)_{j\in
      W}\quad({x_{j}}\xTo{z}{y^{z}_{k}})_{k\in \ol{W},\,z \in \{x_1,\ldots,x_n,x\}}}{\f(x_1,\ldots,x_{n})\xTo{x}
    t}
\end{equation}
is sound in the operational model
\eqref{eq:operational-model-ho}. This means that for all
$p_1,\ldots,p_n,\ol{p}_1,\ldots,\ol{p}_n\in \mS$ such that $p_j\To
\ol{p}_j$ for $j\in W$ and $p_k\Downarrow \ol{p}_k$ for $k\in \ol{W}$, there exists $\ol{p}\in \mS$ where
\begin{itemize}
\item for \eqref{eq:rule-red-weak}, one has $\f(p_1,\ldots,p_n)\To \ol{p}$ and the term $\ol{p}\in \mS$ emerges from $t$ via the substitutions
\[
  x_i\mapsto p_i,
  \;\;
  \;\; y_j\mapsto \ol{p}_j,
  \;\;
  y_k^{x_i} \mapsto (\ol{p}_k)_{p_i},
\]
for $i\in \{1,\ldots,n\}$, $j\in W$, $k\in \ol{W}$;
\item for \eqref{eq:rule-non-red-weak}, one has $\f(p_1,\ldots,p_n)\Downarrow \ol{p}$ and, for every $e\in \mS$, the term $\ol{p}_e\in \mS$ emerges from $t$ via the substitutions
\[
  x_i\mapsto p_i,
  \;\;
  x\mapsto e,
  \;\; y_j\mapsto \ol{p}_j,
  \;\;
  y_k^{x_i} \mapsto (\ol{p}_k)_{p_i},
  \;\;
  y_k^x \mapsto (\ol{p}_k)_e,
\]
for $i\in \{1,\ldots,n\}$, $j\in W$, $k\in \ol{W}$.
\end{itemize}
\end{definition}

\begin{example}
In the specification of \Cref{ex:weak-bis-not-cong} the last rule is unsound for weak
  transitions: we have $d\Downarrow c$ but $u(d)\not\To c$ since
  $u(d)$ reduces to itself. The remaining rules are sound. (For the first two note that all premise-free rules are sound.)
\end{example}

\begin{remark}\label{rem:sound-vs-lax-bialgebra}
  The soundness condition can be expressed in terms of the
  higher-order GSOS law $\rho^0$ of $\Sigma$ over $B_0(X,Y)=Y+Y^X$
  corresponding to the specification $\R$. First, we turn $\rho^0$
  into a higher-order GSOS law $\rho$ of $\Sigma$ over
  $B(X,Y)=\Pow(Y+Y^X)$, where $\rho_{X,Y}$ is defined to be the
  composite
  \begin{equation}\label{eq:rho0-to-rho}
    \begin{tikzcd}[column sep=50, row sep=25]
      \Sigma(X\times \Pow B_0(X,Y)) \ar{r}{\Sigma \st_{X,B_0(X,Y)}} &
      \Sigma\Pow(X\times B_0(X,Y)) \ar{dl}[description]{\delta_{X\times B_0(X,Y)}} \\
      \Pow\Sigma(X\times B_0(X,Y)) \ar{r}{\Pow\rho^0_{X,Y}} & \Pow B_0(X,\Sigmas(X+Y)) 
    \end{tikzcd}
    \hspace*{-1cm}
  \end{equation}
Here $\st$ is the \emph{canonical strength} of the powerset functor $\Pow$,
\[ \st_{X,Z}\c X\times \Pow Z \to \Pow(X\times Z),\;\; (x,U)\mapsto \{\, (x,u) : u\in U \,\},\] and $\delta\c \Sigma\Pow\to \Pow\Sigma$ is the natural transformation whose component $\delta_Z\c \Sigma \Pow Z\to \Pow\Sigma Z$ is given by
\[ \f(Z_1,\ldots,Z_n) \mapsto \{\, \f(z_1,\ldots,z_n) : z_i\in Z_i\text{ for $i=1,\ldots,n$} \,\}.   \] 
The operational model of $\rho$ is easily seen to be the composite
\[ \mS \xto{\gamma_0} B_0(\mS,\mS)\xto{\eta} \Pow B_0(\mS,\mS)=B(\mS,\mS), \]
where $\gamma_0$ is the operational model of $\rho^0$ and $\eta(b)=\{b\}$.

Then the rules of $\R$ are sound for weak transitions iff the diagram \eqref{eq:lax-bialgebra-ho} commutes laxly. Here $\wt{\gamma}$ is the weak transition system of $\gamma$, and the partial order $\preceq$ on a hom-set $\Set(X,\Pow Y)$ is given by $f\preceq g$ iff $f(x)\seq g(x)$ for all $x\in X$. In the terminology introduced later (\Cref{D:lax-bialg}),
$\ini$ and $\wt{\gamma}$ thus form a \emph{lax bialgebra} for the higher-order GSOS law
$\rho$.
\begin{figure*}
\begin{equation}\label{eq:lax-bialgebra-ho}
\begin{tikzcd}
\Sigma(\mS) \ar{r}{\ini} \ar{d}[swap]{\Sigma\langle \id, \wt{\gamma}\rangle} & \mS \ar{r}{\wt{\gamma}} & \Pow(\mS+\mS^\mS) \\
\Sigma(\mS\times \Pow(\mS+\mS^\mS)) \ar{r}{\rho_{\mS,\mS}} & \Pow(\Sigmas(\mS+\mS) + (\Sigmas(\mS+\mS))^\mS) \ar[phantom]{u}[description]{\dgeq{-90}} \ar{r}[yshift=.6em]{\Pow(\Sigmas\nabla+(\Sigmas\nabla)^\mS)} & \Pow(\Sigmas(\mS) + (\Sigmas(\mS))^\mS) \ar{u}[swap]{\Pow(\hat\ini+\hat\ini^\mS)}  
\end{tikzcd}
\end{equation}
\end{figure*}
\end{remark}

\begin{theorem}\label{thm:cool}
For every $\HO$ specification whose rules are sound for weak transitions, the weak similarity relation $\lesssim$ on the canonical model $\gamma\c \mS\to \mS+\mS^\mS$ is a congruence.
\end{theorem}
The proof uses Howe's method~\cite{DBLP:journals/iandc/Howe96}, a standard technique for establishing higher-order congruence results.

\begin{notation}\label{not:howe}
The \emph{Howe closure} of a relation $R\,\seq \mS\times \mS$ is the relation
\[ \hatR = \bigcup_{m\in \Nat} \hatR_m \]
on $\mS$ where $\hatR_0\seq \hatR_1\seq \hatR_2\seq \cdots$
are defined inductively: $\hatR_0 = R$ and for every $m\in \Nat$ and $p,r\in \mS$, one has $\hatR_{m+1}(p,r)$ whenever $\hatR_m(p,r)$ or
\[
\exists \f\in\Sigma, \vec{p},\vec{q}\in (\mS)^{\ar(\f)}.\;  p=\f(\vec{p})\,\wedge\,\hatR_m(\vec{p},\vec{q}) \, \wedge\,  R(\f(\vec{q}),r).
\]
Here $\hatR_m(\vec{p},\vec{q})$ means $\hatR_m(p_i,q_i)$ for $i=1,\ldots,\ar(\f)$.
\end{notation}

\begin{rem}\label{rem:howe-props}
\begin{enumerate}
\item  If $R$ is reflexive, then the Howe closure $\hatR$ is
  a congruence: put $r=f(\vec{q})$ in the definition of
  $\hatR_{m+1}$.
\item  If $R$ is transitive, then $\hatR$ satisfies a weak
  transitivity property: $\hatR(p,r)$ and $R(r,r')$ implies $\hatR(p, r')$
  for all $p,r,r'\in \mS$. This follows by induction on the
  least $m$ such that $\hatR_m(p,r)$.
\item Thus, if $R$ is both reflexive and transitive (in particular, if it is some
  weak similarity relation), then $\hatR$ is the least weakly transitive congruence containing $R$. 
\end{enumerate}
\end{rem}
\begin{proof}[Proof of \Cref{thm:cool}]
Form the Howe closure $\happrox$ of $\lesssim$.
Since~$\happrox$ is a congruence, it suffices to prove ${\happrox} =
{\lesssim}$. The inclusion ${\lesssim}\seq{\happrox}$ is clear. For
the inclusion ${\happrox} \seq {\lesssim}$ %
we show that $\happrox$ is a weak simulation; then the inclusion holds because $\lesssim$ is the greatest weak simulation. By \Cref{rem:weak-vs-strong}\ref{rem:weak-vs-strong-3}, we need to establish the following for every $p\happrox r$ and $\ol{p}\in \mS$:
\begin{align}
p\to \ol{p} &\implies \exists \ol{r}\in \mS.\, r\To\ol{r}\,\wedge\,
\ol{p}\happrox \ol{r}; \label{eq:sim-cond-1} \\
p \not\to &\implies
  \exists \ol{r}\in \mS.\,  r\Downarrow\ol{r} \wedge 
  \forall e\in \mS.\, {p}_e\happrox \ol{r}_e.
\label{eq:sim-cond-2}
\end{align}
In lieu of \eqref{eq:sim-cond-2} we will actually prove a stronger statement:
\begin{align}
p \not\to &\implies\exists \ol{r}\in \mS.\, r\Downarrow\ol{r}\,\wedge\, \forall d\happrox e.\, {p}_d\happrox \ol{r}_e.\label{eq:strengthened-induction-hypothesis}
\end{align}
The proof is by induction on the least $m$ such that $p\happrox_m r$.

\medskip\noindent \textbf{Induction base ($m=0$).} Suppose that $p\happrox_0 r$, that is, $p\lesssim r$.

\medskip\noindent \emph{Proof of \eqref{eq:sim-cond-1}.} If $p\to \ol{p}$,  since $\lesssim$ is a weak simulation, there exists $\ol{r}\in \mS$ such that $r\To \ol{r}$ and $\ol{p}\lesssim \ol{r}$, hence also $\ol{p}\happrox \ol{r}$.

\medskip\noindent \emph{Proof of \eqref{eq:strengthened-induction-hypothesis}.} If $p\not\to$, since $\lesssim$ is a weak simulation, there exists $\ol{r}\in \mS$ such that $r\Downarrow \ol{r}$ and $p_e\lesssim \ol{r}_e$ for $e\in \mS$. By definition of the $\HO$ format, there exists a term $t_p(x)$ in a single variable $x$ such that ${p}_e = t_p(e)$ for $e\in \mS$.
Since $\happrox$ is a congruence, it follows that, for $d\happrox e$, 
\[ {p}_d = t_{{p}}(d)\happrox t_{{p}}(e) = {p}_e \lesssim \ol{r}_e  \]
Thus ${p}_d\happrox\ol{r}_e$ by weak transitivity of the relation $\happrox$.

\medskip\noindent \textbf{Induction step ($m\to m+1$).} Suppose that $p\happrox_{m+1} r$. We only verify condition \eqref{eq:strengthened-induction-hypothesis}, the argument for \eqref{eq:sim-cond-1} is analogous. Thus suppose that $p\not\to$. If $p\happrox_{m} r$, we are done by induction. Otherwise, by definition of $\happrox_{m+1}$, there exists an $n$-ary operation symbol $\f\in\Sigma$ and $\vec{p},\vec{q}\in(\mS)^n$ such that
\[  p=\f(\vec{p}),\qquad \vec{p}\, \mathbin{\happrox_{m}}\vec{q},\qquad q:=\f(\vec{q}) \lesssim r. \]
To avoid bulky notation, we consider the representative case of a binary operator $\f$ where $p_1$ reduces (say $p_1\to \ol{p}_1$) and $p_2\not\to$. Then we know by induction that
\begin{itemize}
\item $\exists \ol{q}_1\in \mS.\, q_1\To \ol{q}_1 \,\wedge\,\ol{p}_1\happrox \ol{q}_1$;

  \smallskip
\item $\exists \ol{q}_2\in \mS.\, q_2\Downarrow \ol{q}_2\,\wedge\,\forall d\happrox e.\, (p_2)_d\happrox (\ol{q}_2)_e$.  
\end{itemize}
 Since $p=\f(p_1,p_2)\not\to$, the rule applying to $p$ has the form
\[
  \inference{x_1\to y_1\quad x_2\xto{{x_1}}y_2^{x_1} \quad x_2\xto{x_2}y_2^{x_2}\quad x_2\xto{x} y_2^x}{\goesv{\f(x_1,x_2)}{t(x_1,x_2,x,y_1,y_2^{x_1},y_2^{x_2},y_2^x)}{x}}.
\]
Thus, for every $d\in \mS$,
\[ p\xto{d} p_d = t(p_1,p_2,d,\ol{p}_1,(p_2)_{p_1}, (p_2)_{p_2}, (p_2)_{d}). \]
The above rule is sound for weak transitions, and we have $q_1\To \ol{q}_1$ and $q_2\Downarrow \ol{q}_2$, so there exists $\ol{q}\in \mS$ such that
\[
  q\Downarrow \ol{q} \qand \ol{q}\xto{e} \ol{q}_e
  =
  t(q_1,q_2,e,\ol{q}_1,(\ol{q}_2)_{q_1}, (\ol{q}_2)_{q_2}, (\ol{q}_2)_{e})
\]
for all $e\in \mS$. Thus for $d\happrox e$ we
have $p_d\happrox \ol{q}_e$ because $\happrox$ is a congruence and the
terms substituted in $t$ for the variables are related by
$\happrox$. Moreover, since $\lesssim$ is a weak simulation and
$q\lesssim r$, there exists $\ol{r}\in \mS$ such that
$r\Downarrow \ol{r}$ and $\ol{q}_e\lesssim \ol{r}_e$ for all
$e\in \mS$. It follows that $p_d\happrox \ol{r}_e$ for $d\happrox e$
because ${p}_d\happrox \ol{q}_e\lesssim \ol{r}_e$ and the relation
$\happrox$ is weakly transitive.
\end{proof}

%

\begin{rem}\label{rem:bifunctor-lifting}
\begin{enumerate}
\item The strengthening \eqref{eq:strengthened-induction-hypothesis} of the induction hypothesis is required, for otherwise the proof gets stuck: the argument in the induction step showing ${p}_d\happrox \ol{q}_e$ for $d\happrox e$ (or even ${p}_e\happrox \ol{q}_e$ for $e\in \mS$) relies on relations such as $({p_{2}})_{p_1}\happrox (\ol{q_{2}})_{q_1}$, which only hold by \eqref{eq:strengthened-induction-hypothesis}, not by \eqref{eq:sim-cond-2}.
\item The strengthened induction hypothesis
  \eqref{eq:sim-cond-1} + \eqref{eq:strengthened-induction-hypothesis}
  can be expressed via the relation lifting of the bifunctor $\barB$,
  see \Cref{rem:weak-vs-strong}\ref{rem:weak-vs-strong-1}: It amounts
  to the existence of a map $\delta$ making the diagram below commute.
\[
\begin{tikzcd}[column sep=3.5em]
\mS \ar{d}[swap]{{\gamma}} & {\happrox} \ar{l}[swap]{\outl_{\happrox}} \ar{d}{\delta} \ar{r}{\outr_{\happrox}} & \mS \ar{d}{\wt{\gamma}} \\
\Pow(\mS+\mS^\mS)  & \ar{l}[swap]{\outl_{{\happrox},{\happrox}}} E_{{\happrox},{\happrox}} \ar{r}{\outr_{{\happrox},{\happrox}}} & \Pow(\mS+\mS^\mS)
\end{tikzcd}
\]
\item The induction does not go through when the Howe closure~$\happrox$ is replaced with more obvious candidates. If $\happrox$ is taken to be the least congruence containing $\lesssim$, then already the induction base fails, as the argument requires weak transitivity of $\happrox$. If $\happrox$ is taken to be the least transitive congruence containing $\lesssim$, it is no longer clear how to construct $\happrox$ as a union of inductively defined relations $\happrox_m$ in a way that makes the induction step work. It thus appears that Howe's method is the simplest and most natural approach to the present result.
\end{enumerate}
\end{rem}
We conclude this section by identifying a natural class of $\HO$
specifications, the \emph{cool $\HO$ specifications}, whose rules are
sound for weak transitions. It resembles first-order formats such as
\emph{cool
  GSOS}~\cite{DBLP:journals/tcs/Bloom95,DBLP:journals/tcs/Glabbeek11}
for labelled transition systems, and \emph{cool stateful
  SOS}~\cite{DBLP:conf/fscd/0001MS0U22} for stateful computations.

\begin{definition}\label{def:cool}
  \begin{enumerate}
  \item An $n$-ary operator $\f\in \Sigma$ is \emph{passive} if it is
    specified by a premise-free rule (cf.\
    \Cref{rem:missing-premises})
    \begin{equation}\label{eq:passive-1}
      \inference{}{\f(x_1,\ldots,x_{n})\to t}
      \;\;\text{or}\;\;
      \inference{}{\f(x_1,\ldots,x_{n})\xto{x} t}
    \end{equation}
    where $t$ is a term in the variables $x_1,\ldots,x_n$ or
    $x_1,\ldots,x_n,x$, resp. Thus the behaviour of $\f$ does not
    depend on the behaviour of its subterms. An \emph{active} operator
    is one which is not passive.
    
  \item An $\HO$ specification is \emph{cool} if for every active
    $n$-ary operator $\f$ there exists $j\in \{1,\ldots, n\}$ (called
    the \emph{receiving position of $\f$}) such that all rules for
    $\f$ are of the form
    \begin{equation}\label{eq:active-1}
      \raisebox{15pt}{$
      \inference{x_j\to y_j}{\f(x_1,\ldots,x_j,\ldots x_{n})\to
        \f(x_1,\ldots,y_j, \ldots x_n)}$}
    \end{equation}
    
    \vspace*{-25pt}
    \[
      \inference{(\goesv{x_{j}}{y^{z}_{j}}{z})_{z \in \{x_1,\ldots,x_n\}}}{\goes{\f(x_1,\ldots,x_{n})}{t}}
      \,\text{or}\,
      \inference{(\goesv{x_{j}}{y^{z}_{j}}{z})_{z \in
          \{x_1,\ldots,x_n,x\}}}{\goesv{\f(x_1,\ldots,x_{n})}{t}{x}}
  \]
    where $t$ is a term in the variables $x_i$ and $y_j^{x_i}$ ($i\in \{1,\ldots,n\}\smin \{j\}$), and moreover in $x$ and $y_j^x$ for the third rule in~\eqref{eq:active-1}.
  \end{enumerate}
\end{definition}

\noindent
Coolness thus asserts that for active $\f$, a program
$p=\f(p_1,\ldots,p_n)$ must run its $j$-th subprogram $p_j$ (for some
fixed $j$ depending only on $\f$) until it does not further reduce, correctly
propagate all reduction steps of $p_j$ to $p$, and continue the computation as a program $t$ that no longer refers to $p_j$. 

\begin{proposition}\label{prop:cool-vs-sound}
For cool $\HO$ specifications, all rules are sound for weak transitions.
\end{proposition}
Thus, we obtain as an instance of \Cref{thm:cool}:
\begin{corollary}\label{cor:cool}
For cool $\HO$ specifications, the weak similarity relation on the operational model is a congruence.
\end{corollary}
This generalizes corresponding congruence results for cool first-order specifications~\cite{DBLP:journals/tcs/Bloom95,DBLP:journals/tcs/Glabbeek11,DBLP:conf/fscd/0001MS0U22}.

\begin{example}
The extended $\SKI$ calculus (\Cref{ex:ski-calculus}) has application $-\circ-$ as its only active operator, whose rules \eqref{eq:ski-rules-app} are cool. Therefore weak similarity on the operational model is a congruence. This means that, for instance, $p\lesssim q$ implies $p\app r \lesssim q\app r$ and $r\app p\lesssim r\app q$ for all $r\in \mS$. 
\end{example}

The aim of the following sections is to generalize the congruence result of \Cref{thm:cool} to the level of abstract higher-order GSOS laws. The technical key lies in 
the construction of relation liftings of bifunctors (\Cref{S:lifting}), along with
a suitable categorification of Howe's method (\Cref{sec:howe}).

\section{Graphs, Relations, and Preorders}\label{sec:graphs-relations-preorders}
\noindent For our categorical account of weak similarity we will need to restrict to base categories where operations on relations, such as union or composition, are well-behaved and interact with each other in a way familiar from the category of sets. Therefore, we work under the following global assumptions:

\begin{assumptions}\label{asm:cat}
  From now on, fix a category $\C$ such that
  \begin{enumerate}
  \item\label{asm:cat-1} $\C$ is complete, cocomplete, and well-powered;
   \item\label{asm:cat-3} strong epimorphisms in $\C$ are pullback-stable: for every pullback as shown below, if $e$ is strongly epic then so is $\ol{e}$;
   \[
\begin{tikzcd}
\bullet \pullbackangle{-45} \ar{r} \ar{d}[swap]{\ol{e}} & \bullet \ar{d}{e} \\
\bullet \ar{r} & \bullet 
\end{tikzcd}
   \]
  \item\label{asm:cat-2} $\C$ is locally distributive.
  \end{enumerate}
\end{assumptions}
All categories of \Cref{ex:categories} satisfy these
assumptions. Since $\C$ is complete and well-powered, the subobjects of a fixed object form a complete lattice, and every morphism has a (strong epi, mono)-factorization~\cite[Prop.~4.4.3]{borceux94}. All our results easily generalize to arbitrary proper factorization systems.

\subsection{Graphs and Relations}\label{sec:graphs-and-relations}
We review some terminology for the
categorical version of graphs (more precisely, directed multigraphs) and relations.

\subsubsection{Graphs in a category}\label{sec:graphs-in-a-category}
 A \emph{graph} in $\C$ is a quadruple $(X,R,\outl_R,\outr_R)$ given by two objects $X,R\in \C$ and a parallel pair of morphisms $\outl_R, \outr_R\c R\to X$. A graph is usually denoted by its pair $(X,R)$ of objects. A \emph{morphism} from $(X,R)$ to a graph $(Y,S)$ is a pair $h=(h_0,h_1)$ of $\C$-morphisms making the diagram below commute:
\begin{equation}\label{diag:hom}
\begin{tikzcd}
X \ar{d}[swap]{h_0} & R \ar{l}[swap]{\outl_R} \ar{d}{h_1} \ar{r}{\outr_R} & X \ar{d}{h_0} \\
Y  & \ar{l}[swap]{\outl_S} S \ar{r}{\outr_S} & Y
\end{tikzcd}
\end{equation}
We let $\Gra(\C)$ denote the category of graphs in $\C$ and their
morphisms. For every $X\in \C$ we write
$\Gra_X(\C) \subto \Gra(\C)$
for the non-full subcategory consisting of all graphs of the form $(X,R)$ and graph morphisms $h$ such that $h_0=\id_{X}$. 

\subsubsection{Relations in a category}\label{sec:relations-in-a-category}
A graph $(X,R)\in \Gra(\C)$ is a \emph{relation} if $\outl_R$
  and $\outr_R$ are jointly monic, or equivalently if the morphism
  $\langle \outl_R,\, \outr_R \rangle\c R\to X\times X$ is monic. We let
$\Rel(\C)\subto \Gra(\C)$ and $\Rel_X(\C)\subto \Gra_X(\C)$
denote the full subcategories given by relations; note that $\Rel_X(\C)$ is thin, i.e.\ an ordered set, and a complete lattice when  isomorphic relations are identified. Both subcategories are reflective: The reflection of a graph $(X,R)$ is given by
$(\id_X,e_R)\c (X,R)\epito (X,R^\dag)$
where $e_R$ and $R^\dag$ are obtained via the (strong epi, mono)-factorization of $\langle \outl_R,\,\outr_R\rangle$:
\begin{equation}\label{diag:Rdag}
  \begin{tikzcd}[column sep=7em]
    R
    \ar[two heads]{r}{e_R}
    \ar[shiftarr={yshift=1.75em}]{rr}{\langle \outl_R,\,\outr_R\rangle } &
    R^\dag
    \ar[rightarrowtail]{r}{\langle
      \outl_{R^\dag},\,\outr_{R^\dag}\rangle }
    & 
    X \times X
  \end{tikzcd}
\end{equation}
The various categories are connected by the functors
\[
\begin{tikzcd}
  \Gra_X(\C) \ar[shiftarr ={yshift=1.7em}]{rr}{\under{-}}
  \ar[yshift=.2em]{r}{(-)^\dag} \ar[hook]{d}
  &
  \Rel_X(\C) \ar[hook, yshift=-.2em]{l} \ar[hook]{d} \ar{r}{\under{-}}
  & \C \ar[equals]{d}
  \\
  \Gra(\C) \ar[shiftarr ={yshift=-1.6em}]{rr}{\under{-}}
  \ar[yshift=.2em]{r}{(-)^\dag}
  &
  \Rel(\C) \ar[hook, yshift=-.2em]{l} \ar[yshift=.2em]{r}{\under{-}} & \C \ar[hook, yshift=-.2em]{l}
\end{tikzcd}
\]
where $(-)^\dag$ denotes the reflector and $\under{-}$ is the projection functor given by
$(X,R)\mapsto X$ and $h\mapsto h_0$.
We regard $\C$ as a full subcategory of $\Rel(\C)$ by identifying $X\in \C$ with the \emph{identity relation} $(X,X,\id_X,\id_X)\in \Rel(\C)$, which we simply denote by $(X,X)$.

\subsubsection{Limits and colimits}\label{sec:graph-cocomplete} The categories $\Gra(\C)$, $\Gra_X(\C)$, $\Rel(\C)$, $\Rel_X(\C)$ are complete and cocomplete. Coproducts in $\Gra(\C)$ and $\Gra_X(\C)$, denoted by $(X,R)+(Y,S)$ and $(X,R)+_X (X,S)$, are formed using $\C$-coproducts. Coproducts in $\Rel(\C)$ are given by $(X,R)\vee (Y,S)=((X,R)+(Y,S))^\dag$ and in $\Rel_X(\C)$ by $(X,R)\vee_X (X,S) = ((X,R)+_X (X,S))^\dag$.

Products $(X,R)\times (Y,S)$ in both $\Gra(\C)$ and $\Rel(\C)$ are formed in $\C$. The product $(X,R)\times_X (X,S)$ in $\Gra_X(\C)$ and $\Rel_X(\C)$ is the pullback of $\langle \outl_R,\outr_R\rangle$ and $\langle \outl_S,\outr_S\rangle$.

\subsubsection{Composition of graphs and relations}\label{sec:graph-rel-comp} The \emph{composite} $(X,R)\smc (X,R')$ of two graphs $(X,R)$ and $(X,R')$ is the graph
$(X,R\smc R')$
defined via the following pullback:
\begin{equation}\label{eq:composite-graph}
\begin{tikzcd}[column sep=0.65em]
  & &
  R\smc R' 
  \pullbackangle{-90}
  \ar{dl}[swap]{\pi_{R\smc R',R}}
  \ar{dr}{\pi_{R\smc R',R'}}
  \ar[shiftarr = {xshift=-50}]{ddll}[description]{\outl_{R\smc R'}}
  \ar[shiftarr = {xshift=50}]{ddrr}[description]{\outr_{R\smc R'}}
  &  & \\
& R \ar{dl}[swap]{\outl_R} \ar{dr}[description]{\outr_R} & & R' \ar{dl}[description]{\outl_{R'}} \ar{dr}{\outr_{R'}}  & \\
X & & X & & X 
\end{tikzcd}
\end{equation}
The \emph{composite} of two relations $(X,R),(X,R')$, given by
\[ (X,R)\bullet (X,R') \;=\; ((X,R)\smc (X,R'))^\dag, \] defines a
bifunctor $(-)\bullet (-)$ on $\Rel_X(\C)$ (that is, composition is
a monotone map on the ordered set of relations). Using \Cref{asm:cat}\ref{asm:cat-3},\ref{asm:cat-2}, relation composition can be shown to distribute over coproducts. This is the key property of relations needed for our account of Howe's method in \Cref{sec:howe}.

\subsubsection{Reflexive and transitive relations}\label{sec:morphism-preorder}
Given graphs $(X,R)$ and $(X,R')$ in $\Gra_X(\C)$, we put
$(X,R)\leq (X,R')$ if there exists a $\Gra_X(\C)$-morphism from
$(X,R)$ to $(X,R')$. For relations, $(X,R)\leq (X,R')\leq (X,R)$
implies $(X,R)\cong (X,R')$.  A relation $(X,R)$ is \emph{reflexive}
if $(X,X)\leq (X,R)$, and \emph{transitive} if
$(X,R)\bullet (X,R)\leq (X,R)$.

\subsubsection{Reindexing}\label{sec:transfer}
Every morphism $f\c X\to Y$ in $\C$ induces a functor
$f_\star\c  \Gra_X(\C)\to \Gra_Y(\C)$ given by
\[
\qquad (X,R,\outl_R,\outr_R)\;\mapsto\; (Y,R,f\comp \outl_R, f\comp \outr_R).\]
Readers familiar with the language of fibrations may note that $\under{-}\c \Gra(\C)\to \C$ is a bifibration with fibres $\Gra_X(\C)$, and~$f_\star$ is the reindexing functor induced by opcartesian lifts.

\subsection{Preorders}
We extend some of the above terminology to graphs over preordered
objects. Recall that a \emph{preorder} on a set $X$ is a reflexive and
transitive relation ${\preceq}\seq X\times X$. Replacing elements
$1\to X$ with ``generalized elements'' $Y\to X$, one obtains a
categorical notion of preorder.

\subsubsection{Preorders in a category}\label{sec:preordered-object}
A \emph{preordered object} in $\C$ is a pair $(X,\preceq)$ of an
object $X\in \C$ and a family ${\preceq} = (\preceq_{Y})_{Y\in \C}$
where $\preceq_{Y}$ is a preorder on the hom-set $\C(Y,X)$ satisfying
\[ f\preceq_Y g \quad\implies\quad f\comp h\preceq_Z g\comp h \quad\text{for all $h\c Z\to Y$}.\]
We usually drop subscripts and write $\preceq$ for $\preceq_Y$, $\preceq_Z$, etc.
\begin{example}\label{ex:preordered-object}
  \begin{enumerate}
  \item\label{ex:preordered-object-1} Every preordered set
    $(X,\preceq)$ in the usual order-theoretic sense can be regarded
    as a preordered object in $\Set$ by taking the pointwise preorder
    on $\Set(Y,X)$:
    \[
      f\preceq g \quad\iff\quad \forall y\in Y.\, f(y)\preceq g(y).
    \]
  \item\label{ex:preordered-object-discrete} On every $X\in \C$, one
    has the \emph{discrete} preordered object $(X,=)$, where $=$ is
    the equality preorder.
  \end{enumerate}
\end{example}

\subsubsection{Preordered functors}
 A \emph{preordered functor} is a functor $F\c \D\to \C$ equipped with a preorder $(FD,\preceq)$ for all $D\in \D$.

\begin{example}
The powerset functor $\Pow\c \Set\to \Set$ is preordered by taking the inclusion preorder $\seq$ on $\Pow X$.
\end{example}

\subsubsection{Right-lax morphisms}
Given a preordered object $(Y,\preceq)$, a \emph{right-lax morphism} from a
graph $(X,R)$ to a graph $(Y,S)$ is a pair $h=(h_0, h_1)$ of
$\C$-morphisms such that
\[
\begin{tikzcd}
X \ar[phantom]{dr}[description, pos=.4]{\deq{-45}} \ar{d}[swap]{h_0} & R \ar[phantom]{dr}[description, pos=.4]{\dleq{45}} \ar{l}[swap]{\outl_R} \ar{d}{h_1} \ar{r}{\outr_R} & X \ar{d}{h_0} \\
Y & \ar{l}[swap]{\outl_S} S \ar{r}{\outr_S} & Y
\end{tikzcd}
\]
For preordered $(X,\preceq)$ we put $(X,R)\preceq (X,S)$ if there exists a right-lax morphism $h\c (X,R)\to (X,S)$ where $h_0=\id_X$.
\begin{example}
Given relations $(X,R)$, $(X,S)$ on a preordered set $(X,\preceq)$, regarded as a preordered object as in \Cref{ex:preordered-object}\ref{ex:preordered-object-1}, we have $(X,R)\preceq (X,S)$ iff for $x,y\in X$,
\[ R(x,y) \quad\implies\quad \exists z\in X.\,  S(x,z) \wedge z\preceq y.  \]
\end{example}
Graphs over a fixed preordered object $(X,\preceq)$ and right-lax morphisms $h$ satisfying $h_0=\id_X$ form a category, but in contrast to the unordered case, the full subcategory of relations is usually not reflective. This turns out to be the main technical challenge for our preorder-based approach to simulations. The key concept to overcome this issue is as follows:  

\begin{definition}\label{def:good-for-simulations-relation} Let $(X,\preceq)$ be a preordered object. A relation $(X,S)$ is \emph{good for simulations} if, for all $(X,R)\in \Gra_X(\C)$,
\[ (X,R)\preceq (X,S) \quad\implies\quad (X,R)\leq (X,S). \]
\end{definition}
Note that $(X,R)$ ranges over graphs, not just relations, and that the implication ``$\Longleftarrow$''  also holds trivially. The good-for-simulations condition thus ensures that right-lax graph morphisms into $(X,S)$ can be turned into strict ones.

\begin{example}\label{ex:good-for-simulations}
For every relation $(X,R)$, the relation $(\Pow X, S_R)$ is good for
simulations, where $\Pow X$ is equipped with the inclusion preorder
and $S_R$ is the Egli-Milner relation (\Cref{rem:weak-vs-strong}\ref{rem:weak-vs-strong-1}). This follows from the observation that $S_R$ is up-closed: $S_R(A,B)$ and $B\seq B'$ implies $S_R(A,B')$.
\end{example}

\section{Lifting (Bi-)Functors and Higher-Order GSOS Laws}
\label{S:lifting}
\noindent As pointed out in \Cref{rem:bifunctor-lifting}, the compositionality proof for $\HO$ specifications implicitly relies on the fact that the behaviour bifunctor $B(X,Y)=\Pow(Y+Y^X)$ admits a lifting to the category of relations. We next study liftings of endofunctors, mixed-variance bifunctors, and higher-order GSOS laws on $\C$ to the categories $\Gra(\C)$ and $\Rel(\C)$ of graphs and relations. We start with the case of endofunctors, which is straightforward and well-known:

\begin{definition}\label{def:endofunctor-lift} Let $\Sigma\c \C\to\C$ be an endofunctor.
\begin{enumerate}
\item A \emph{graph lifting} of $\Sigma$ is a functor
  $\ol{\Sigma}\c \Gra(\C)\to \Gra(\C)$ making the diagram on the left below
  commute.
  
\item A \emph{relation lifting} of $\Sigma$ is a functor
  $\ol{\Sigma}\c \Rel(\C)\to \Rel(\C)$ making the diagram on the right
  below commute.
\end{enumerate}
\[
\begin{tikzcd}
\Gra(\C) \ar{d}[swap]{\under{-}}  \ar{r}{\ol{\Sigma}} & \Gra(\C) \ar{d}{\under{-}}  \\
\C \ar{r}{\Sigma}  & \C
\end{tikzcd}
\qquad
\begin{tikzcd}
\Rel(\C) \ar{d}[swap]{\under{-}}  \ar{r}{\ol{\Sigma}} & \Rel(\C) \ar{d}{\under{-}}  \\
\C \ar{r}{\Sigma}  & \C
\end{tikzcd}
\]
\end{definition}
\begin{construction}\label{cons:lifting-endofunctor}
  Every functor $\Sigma\c \C\to \C$ admits a \emph{canonical graph
    lifting} $\ol{\Sigma}_\Gra\c \Gra(\C)\to \Gra(\C)$ and a
  \emph{canonical relation lifting}
  $\ol{\Sigma}_\Rel\c \Rel(\C)\to\Rel(\C)$ defined as follows:
  \begin{enumerate}
  \item The functor $\ol{\Sigma}_\Gra$ 
    is given on objects and morphisms by
    \[
      (X,R)\mapsto (\Sigma X,\Sigma
      R,\Sigma\outl_R,\Sigma\outr_R),\quad h\mapsto (\Sigma h_0,\Sigma
      h_1).
    \]
\item The functor $\ol{\Sigma}_\Rel$ is the composite
  \[
    \begin{tikzcd}[column sep = 25, cramped]
      \Rel(\C)\ar[hook]{r}
      &
      \Gra(\C) \ar{r}{\ol{\Sigma}_\Gra}
      &
      \Gra(\C)\ar{r}{(-)^\dag}
      &
      \Rel(\C).
    \end{tikzcd}
  \]
\end{enumerate}
\end{construction}
\noindent (This is similar to the usual Barr extension~\cite{Barr70b}, except that
relations are treated as objects rather than as morphisms.)

\begin{example}\label{ex:lifting-endofunctor}
For a polynomial functor $\Sigma$ on $\Set$, the canonical relation lifting is the restriction of the canonical graph lifting to $\Rel$. Thus $\ol{\Sigma}_\Rel(X,R)=(\Sigma X,\Sigma R)$ where $\Sigma R(\f(x_1,\ldots,x_n),\f(x_1',\ldots,x_n'))$ iff
$R(x_i,x_i')$ for all $i$.
\end{example}

\begin{proposition}\label{prop:free-monad-lift}
Suppose that $\Sigma\c \C\to \C$ preserves strong epimorphisms and generates a free monad $\Sigma^\star$. Then $\ol{\Sigma}_\Gra$ and $\ol{\Sigma}_\Rel$ generate free monads satisfying
\[ (\ol{\Sigma}_\Gra)^\star = (\barSigmas)_\Gra \qqand (\ol{\Sigma}_\Rel)^\star = (\barSigmas)_\Rel.  \]
\end{proposition}

Next we turn to liftings of mixed-variance bifunctors.
\begin{definition}\label{def:bifunctor-lift} A \emph{relation lifting} of a functor $B\c \C^\opp\times \C\to \C$ is a functor $\barB$ such that the diagram below commutes.
\[
\begin{tikzcd}
 \Rel(\C)^\opp\times \Rel(\C)  
 \ar{d}[swap]{\under{-}^\opp\times \under{-}} \ar{r}{\barBs} & \Rel(\C) \ar{d}{\under{-}}  \\
\C^\opp \times \C \ar{r}{B}  & \C 
\end{tikzcd}
\]
\end{definition}
\noindent Every bifunctor admits a canonical relation lifting,
generalizing the lifting $\barB_0$ of
\Cref{rem:weak-vs-strong}\ref{rem:weak-vs-strong-1}. Since the
construction is more involved than for endofunctors, and
our compositionality result works with any lifting, we refer to the Appendix (\Cref{sec:bifunctor-lift}). Finally, we lift higher-order GSOS
laws:
\begin{definition}\label{def:rho-lift}
  Let $\Sigma\c \C\to \C$ and $B\c \C^\opp\times \C\to \C$ be functors
  with relation liftings $\ol{\Sigma}$ and $\barB$, respectively,
  where $\Sigma$ preserves strong epimorphisms and
  $\ol{\Sigma}=\ol{\Sigma}_\Rel$ is the canonical lifting. Given a
  $V$-pointed higher-order GSOS law
  \[
    \rho_{X,Y}\c \Sigma(X\times B(X,Y))\to B(X,\Sigmas(X+Y))
  \]
  of $\Sigma$ over $B$, a \emph{relation lifting} of $\rho$ is a
  $(V,V)$-pointed higher-order GSOS law
  \[
    \begin{tikzcd}
      \ol{\Sigma}((X,R)\times \barB((X,R),(Y,S)))
      \ar{d}{\barrho_{(X,R),(Y,S)}}
      \\
      \barB((X,R),\ol{\Sigma}^\star((X,R)\vee(Y,S))) 
    \end{tikzcd}
  \]
  of $\ol{\Sigma}$ over $\barB$ such that
  \[
    (\barrho_{(X,R),(Y,S)})_0=\rho_{X,Y}
  \]
  for $((X,R),p_{(X,R)})\in (\Pt,\Pt)/\Rel(\C)$ and
  $(Y,S)\in \Rel(\C)$. Here we regard $X$ as $V$-pointed by
  $p_X=(p_{(X,R)})_0\c V\to X$.
\end{definition}

\begin{rem}\label{rem:rho-rel-lifting}
\begin{enumerate}
\item Recall that for a $V$-pointed higher-order GSOS law $\rho$ we assume the functor $\Sigma$ to be of the form $V+\Sigma'$. This implies $\ol{\Sigma}=(V,V)\vee \ol{\Sigma'}$, as required.
\item The product $\times$ and coproduct $\vee$ in $\Rel(\C)$ are formed as explained in \Cref{sec:graph-cocomplete}, and we have $\ol{\Sigma}^\star= \barSigmas$ by \Cref{prop:free-monad-lift}. It follows that $(\barrho_{(X,R),(Y,S)})_0$ is a $\C$-morphism of type $\Sigma(X\times B(X,Y))\to B(X,\Sigmas(X+Y))$.
\item Since $\Rel(\C)$-morphisms are uniquely determined by their $(-)_0$-component, a higher-order GSOS law $\rho$ admits at most one lifting $\barrho$. The requirement that the morphisms $\barrho_{(X,R),(Y,S)}$ form a higher-order GSOS law of $\ol{\Sigma}$ over $\barB$ is thus vacuous: the (di-)naturality of $\barrho$ is implied by that of $\rho$.
\end{enumerate}
\end{rem}
\noindent For the canonical relation lifting $\barB$ of $B$, every higher-order GSOS law admits a relation lifting (see Appendix, \Cref{sec:can-lift-gsos-laws}).

\section{Weak Simulations}\label{sec:weak-simulations}
\noindent We next introduce the notion of weak simulation featuring in our abstract congruence result. 

\begin{notation}\label{not:f-lifting}
  Fix a functor $F\c \C\to \C$ and a relation lifting~$\barF$. We
  denote the relation $\barF(X,R)$ by $(FX,E_R)$.
\end{notation}
We recall the notion of \emph{(lifting)
  bisimulation}~\cite{DBLP:books/cu/J2016} for coalgebras. We use the
term \emph{simulation} instead, as this is what the concept amounts to
in our applications, due to the use of asymmetric liftings such as the
one-sided Egli-Milner lifting. An alternative approach to simulations uses lax liftings~\cite{hj04}.

\begin{definition}\label{def:simulation}
  Let $(C,c)$ be an $F$-coalgebra. A relation $(C,R)$ is a
  \emph{simulation} on $(C,c)$ if $c_\star(C,R)\leq \barF(C,c)$, that
  is, there exists a morphism $c_R$ making \eqref{eq:simulation}
  commute.
  \begin{equation}\label{eq:simulation}
    \begin{tikzcd}[column sep = 35]
      C \ar{d}[swap]{c} & R \ar{l}[swap]{\outl_R} \ar{d}{c_R} \ar{r}{\outr_R} & C \ar{d}{c} \\
      FC  & \ar{l}[swap]{\outl_{E_R}} E_R \ar{r}{\outr_{E_R}} & FC
    \end{tikzcd}
  \end{equation}
  If it exists, the greatest simulation with respect to the partial
  order $\leq$ on $\Rel_C(\C)$ is called the \emph{similarity relation} on
  $(C,c)$.
\end{definition}

\begin{lemma}
\label{lem:sim-props} Suppose that the functor $\barF$ satisfies the following conditions for all $X\in \C$ and $(X,R),(X,S)\in \Rel_X(\C)$:
\begin{enumerate}
\item\label{eq:good-for-simulations-1} the relation $\barF(X,X)$ is reflexive;
\item\label{eq:good-for-simulations-2} $\barF(X,R)\bullet \barF(X,S)\leq \barF((X,R)\bullet (X,S))$. 
\end{enumerate}
Then for every $F$-coalgebra $(C,c)$ the similarity relation exists, and it is reflexive and transitive. 
\end{lemma}
\noindent The conditions in the above lemma are similar to ones
occurring in work on \emph{lax extensions}, e.g.~by Marti and
Venema~\cite{MartiVenema15}.

In the setting of $\HO$ specifications,
where $F=\Pow B_0(X,-)$, a weak simulation on a $B_0(X,-)$-coalgebra
$(C,c)$ as per \Cref{def:weak-sim} is precisely a simulation on the
weak transition system $(C,\wt{c}\hspace{1pt})$. As observed in
\Cref{rem:weak-vs-strong}\ref{rem:weak-vs-strong-3}, in order to check
the weak simulation conditions for $R(p,q)$, it suffices to show that
strong transitions from $p$ are simulated by weak transitions from
$q$. This turns out to be the only property of weak simulations needed
for our categorical congruence proof, and so we take it as our
abstract definition:

\begin{definition}\label{def:weakaning-and-weak-sim}
A \emph{weakening} of a coalgebra $c\c C\to FC$ is a coalgebra $\wt{c}\c C\to FC$ such that for every relation $(C,R)$ the following two statements are equivalent:
\begin{enumerate}
\item $(C,R)$ is a simulation on $(C,\wt{c}\hspace{1pt})$;
\item there exists a morphism $\wt{c}_R$ making \eqref{eq:weak-sim-diag} commute.
\end{enumerate}
\
\begin{equation}\label{eq:weak-sim-diag}
\begin{tikzcd}[column sep = 35]
C \ar{d}[swap]{c} & R \ar{l}[swap]{\outl_R} \ar{d}{\wt{c}_R} \ar{r}{\outr_R} & C \ar{d}{\wt{c}} \\
FC  & \ar{l}[swap]{\outl_{E_R}} E_R \ar{r}{\outr_{E_R}} & FC
\end{tikzcd}
\end{equation}
A \emph{weak simulation} on $(C,c)$, with respect to a given weakening $(C,\wt{c})$, is a relation $(C,R)$ satisfying the two equivalent properties above. If it exists, the greatest weak simulation is called \emph{weak similarity}, denoted $\lesssim_{(C,c)}$ or just $\lesssim$.  
\end{definition}

\begin{rem}\label{rem:weakening}
\begin{enumerate}
\item\label{rem:weakening-2} For the trivial weakening $\wt{c}=c$, weak simulations are just (strong) simulations.
\item The above definition is agnostic about how the
  weakening~$\wt{c}$ is actually constructed from~$c$. The
  construction of weak coalgebras has
  been studied in specific order-enriched settings~\cite{DBLP:journals/corr/Brengos13,DBLP:journals/jlp/BrengosMP15,DBLP:conf/icalp/GoncharovP14}.
  Our present abstract approach is flexible in the choice
  of~$\wt{c}$. For example, the weak transition system $\wt{c}$ of
  \Cref{notation:ho-format} is an instance of the framework of~\cite{DBLP:journals/corr/Brengos13}, but the choice
  $\wt{c}=c$ as in part \ref{rem:weakening-2} above is not.
\end{enumerate}
\end{rem}

\section{Howe's Method, Categorically}\label{sec:howe}
\noindent Next, we set up our version of Howe's method, which regards Howe closures abstractly as initial algebras. In a restricted setting of presheaf categories, this idea already appears in the work of Borthelle et al.~\cite{DBLP:conf/lics/BorthelleHL20} and Hirschowitz and Lafont~\cite{DBLP:journals/lmcs/HirschowitzL22}.   
\begin{notation}\label{not:howe-cat}
  Let $\Sigma\c \C\to \C$ be an endofunctor with its canonical
  relation lifting $\ol{\Sigma}=\ol{\Sigma}_\Rel$
  (\Cref{cons:lifting-endofunctor}). For every $(X,R)\in \Rel_X(\C)$
  and every $\Sigma$-algebra $(X,\xi)$ with monic structure $\xi\c
  \Sigma X\monoto X$, let
  \[ 
    \ol{\Sigma}_{R,\xi}\c \Rel_X(\C)\to \Rel_X(\C) 
  \]
 be the endofunctor (= monotone map) given by
\begin{align*}
(X,S) \quad&\mapsto\quad (X,R) \vee_X \big((\xi_\star\ol{\Sigma}(X,S))\bullet (X,R)\big),
\end{align*}
see \Cref{sec:graphs-and-relations} for the notation. (The assumption that $\xi$ is monic ensures that $\xi_\star$ maps relations to relations.)
Since $\Rel_X(\C)$ is equivalent to a complete lattice, the initial algebra of $\ol{\Sigma}_{R,\xi}$  exists, and we denote it by
  \begin{equation}\label{eq:howe-alg-struct}
    (X,R)\vee_X\big((\xi_\star \ol{\Sigma}(X,\hatR))\bullet (X,R)\big) \xto{\alpha_{R,\xi}} (X,\hatR).
  \end{equation}
  The relation $(X,\hatR)$ is called the \emph{Howe closure} of $(X,R)$ with respect to the algebra $(X,\xi)$.
\end{notation}

\begin{rem}
  We will instantiate the above to the initial algebra
  $(X,\xi)=(\mS,\ini)$; note that the structure $\ini$ is an isomorphism. For
  $\C=\Set$ and $\Sigma$ a polynomial functor, the above definition of
  $\hatR$ is equivalent to the one of \Cref{not:howe}.
\end{rem}
\Cref{lem:howe-props} below establishes some basic properties of Howe closures, generalizing \Cref{rem:howe-props}. For that purpose let us recall the notion of \emph{congruence} for functor
algebras:
\begin{definition}\label{def:cong}
  A \emph{congruence} on a $\Sigma$-algebra $(A,a)$ is a relation
  $(A,R)$ such that $a_\star\ol{\Sigma}_\Rel(A,R)\leq (A,R)$.
\end{definition}
\noindent For a polynomial functor $\Sigma$ on $\Set$, this matches
the definition of congruence from universal algebra (cf.~\Cref{sec:categories}).

\begin{lemma}\label{lem:howe-props}
Let $\Sigma\c \C \to \C$ be an endofunctor. Then for each $(X,R)\in \Rel(\C)$ and each monic algebra $\xi\c \Sigma X\monoto X$,
\begin{enumerate}
\item if $(X,R)$ is reflexive, then $(X,\hatR)$ is reflexive and a congruence on $(X,\xi)$;
\item if $(X,R)$ is transitive, then $(X,\hatR)$ is \emph{weakly transitive}, that is, $(X,\hatR)\bullet (X,R)\leq (X,\hatR)$.
\end{enumerate}
\end{lemma}

\section{Compositionality}\label{sec:weak-sim}
\noindent We proceed to establish our main theorem, which asserts that under natural conditions, weak similarity is a congruence on the operational model of a higher-order GSOS law. 

\begin{assumptions}\label{asm-sim}
In this section we fix the following data:
\begin{enumerate}
\item\label{asm-sim-1} a functor $\Sigma=V+\Sigma'\c \C \to \C$ that preserves strong epimorphisms and generates a free monad $\Sigmas$;
\item\label{asm-sim-2} a preordered bifunctor $B\c \C^\op\times \C\to
  \C$ with a relation lifting $\barB$ that is \emph{good for
    simulations} (\Cref{def:good-for-simulations-bifunctor});
\item\label{asm-sim-3} a $V$-pointed higher-order GSOS law $\rho$ of $\Sigma$ over $B$ that admits a (necessarily unique) relation lifting $\barrho$.
\end{enumerate}
\end{assumptions}
\noindent It remains to explain Assumption \ref{asm-sim}\ref{asm-sim-2}:
\begin{definition}\label{def:good-for-simulations-bifunctor}
A relation lifting $\barB$ of $B$ is \emph{good for si\-mu\-la\-tions} if, for $X,Y\in \C$ and $(X,R),(Y,S),(Y,S')\in \Rel(\C)$,
\begin{enumerate}[label=(G\arabic*)]
\item\label{eq:good-for-simulations-bifunctor-1} the relation $\barB((X,R),(Y,S))$ is good for simulations;
\item\label{eq:good-for-simulations-bifunctor-2} the relation $\barB((X,X),(Y,Y))$ is reflexive;
\item\label{eq:good-for-simulations-bifunctor-3}
  $
  \begin{array}[t]{@{}l}
    \barB((X,R),(Y,S))\bullet \barB((X,X),(Y,S'))\leq
    \\
    \barB((X,R), (Y,S)\bullet (Y,S')).
  \end{array}
  $
\end{enumerate}
\end{definition}
\begin{rem}To motivate Assumption \ref{asm-sim}\ref{asm-sim-3}, let us revisit the setting of $\HO$ specifications, where $B(X,Y)=\Pow(Y+Y^X)$ and $\rho$ is given by \eqref{eq:rho0-to-rho}. Existence of a relation lifting of $\rho$ means that for $(X,R)$, $(Y,S)\in \Rel$ the map $\rho_{X,Y}$ is a $\Rel$-morphism w.r.t.\ the lifting $\barB$ of \Cref{rem:weak-vs-strong}\ref{rem:weak-vs-strong-1}. In the proof of \Cref{thm:cool} (induction base for \eqref{eq:strengthened-induction-hypothesis}) a syntactic argument shows that $p_d \happrox p_e$ for $d \happrox e$, which amounts to the above property for $(X,R)=(Y,S)=(\mu\Sigma,\happrox)$. Hence, the purpose of Assumption \ref{asm-sim}\ref{asm-sim-3} is to replace the syntactic part of that proof by an abstract condition on the law $\rho$.
\end{rem}

In the following we study weak simulations on the operational model $(\mS,\gamma)$
of the higher-order GSOS law $\rho$, understood w.r.t.~to the relation lifting $\barB((\mS,\mS),-)\c \Rel(\C)\to\Rel(\C)$ of the endofunctor $B(\mS,-)\c \C\to \C$ and a given weakening $(\mS,\wt{\gamma})$ of $(\mS,\gamma)$. By \ref{eq:good-for-simulations-bifunctor-2} and \ref{eq:good-for-simulations-bifunctor-3} the lifted endofunctor satisfies the conditions of \Cref{lem:sim-props}, hence the weak similarity relation on $(\mS,\gamma)$ exists. 

The core ingredient for our congruence theorem is a higher-order
variation of \emph{lax models} for monotone GSOS laws~\cite{DBLP:conf/concur/BonchiPPR15}:

\begin{definition}\label{D:lax-bialg}
  A \emph{lax $\rho$-bialgebra} $(X,a,c)$ is given by an object
  $X\in \C$ and morphisms $a\c \Sigma X\to X$ and $c\c X\to B(X,X)$
  such that the diagram below commutes laxly. Note that $X$ is
  $V$-pointed; the point $p_X\colon V \to X$ is induced by the algebra
  $a\colon \Sigma X \to X$ (\Cref{not:rho}).
  \[ 
    \begin{tikzcd}[column sep=1.1em]
      \Sigma X \ar{r}{a} \ar{d}[swap]{\langle \id, c\rangle}
      &
      X \ar{r}{c}
      &
      B(X,X)
      \\
      \Sigma(X\times B(X,X))
      \ar{r}[rotate=30,pos=1.5,yshift=.5em]{\rho_{X,X}}
      &
      B(X,\Sigmas(X+X))
      \ar[phantom]{u}[description]{\dgeq{-90}}
      \ar{r}[rotate=30,pos=2.3,yshift=.8em]{B(\id,\Sigmas\nabla)}
      &
      B(X,\Sigmas X) \ar{u}[swap]{B(\id,\wh{a})}  
    \end{tikzcd}
  \] 
\end{definition}
\noindent
This generalizes the notion of \emph{$\rho$-bialgebra}~\cite{gmstu23}
which requires strict commutativity of the above diagram.

 Our congruence theorem rests on the assumption that $(\mS,\ini,\wt{\gamma})$ is a lax $\rho$-bialgebra. As indicated in \Cref{rem:sound-vs-lax-bialgebra}, this expresses in abstract terms that the operational rules encoded by the higher-order GSOS law $\rho$ are sound for weak transitions in the operational model. In the setting of $\HO$ specifications, we proved that this entails the congruence property for weak similarity (\Cref{thm:cool}). The next proposition is key to our categorical generalization of that result.

\begin{proposition}\label{prop:howe}
Suppose that $\wt{\gamma}$ is a weakening of the operational model $(\mS,\gamma)$ such that $(\mS,\ini,\wt{\gamma})$ is a lax $\rho$-bialgebra. Then for every reflexive and transitive
  weak simulation $(\mS,R)$ on $(\mS,\gamma)$, the Howe closure
  $(\mS,\hatR)$ w.r.t.~$(\mS,\ini)$ is a weak simulation.
\end{proposition}
\noindent In the proof below we denote the relation $\barB((X,S),(Y,T))$ by $(B(X,Y),E_{S,T})$ and its projections by $\outl_{S,T}, \outr_{S,T}$.

\begin{proof}[Proof sketch]
  Form the
  relation $(\mS,P)$ via the pullback 
\begin{equation*}
\begin{tikzcd}
P \pullbackangle{-45} \ar{r}{p} \ar{d}[swap]{\langle \outl_P, \outr_P \rangle} & E_{\hatR,\hatR} \ar{d}{\langle \outl_{\hatR,\hatR},\outr_{\hatR,\hatR}\rangle} \\
\mS\times \mS \ar{r}{\langle \gamma, \wt{\gamma}\rangle} & B(\mS,\mS)\times B(\mS,\mS)
\end{tikzcd}
\end{equation*}
The crucial step is to show existence of $\Rel_{\mS}(\C)$-morphisms
  \begin{align*}
    \beta^\lft \c & (\mS,R)\to (\mS,P), \\
    \beta^\rgt\c & \iota_\star\ol{\Sigma}((\mS,\hatR)\times_{\mS} (\mS,P))\bullet (\mS,R) \to (\mS,P),
  \end{align*}
  where $\times_{\mS}$ is the product in $\Rel_\mS(\C)$
  (\Cref{sec:graph-cocomplete}). Their construction imitates the
  arguments for the induction base and induction step, respectively, in the
  proof of \Cref{thm:cool}.

  Once this is achieved, we can conclude the proof as follows. By
  copairing $\beta^\lft$ and $\beta^\rgt$ we obtain the
  $\Rel_{\mS}(\C)$-morphism
  \[
    \beta = [\beta^\lft,\beta^\rgt]\c 
    \ol{\Sigma}_{R,\iota} \big((\mS,\hatR) \times_{\mS}
    (\mS,P)\big)
    \to (\mS,P)
  \]
  (cf.~\cref{not:howe-cat}) and thus primitive
  recursion~\eqref{eq:primitive-recursion} yields the
  $\Rel_\mS(\C)$-morphism
  $\primrec\beta\c \mu\ol{\Sigma}_{R,\ini}=(\mS,\hatR)\to (\mS,P)$.
  Choose a morphism $r\c (\mS,\mS)\to (\mS,\hatR)$ witnessing that
  $(\mS,\hatR)$ is reflexive (\Cref{lem:howe-props}). Then the
  commutative diagram below proves $(\mS,\hatR)$ to be a weak
  simulation.
  \[
    \begin{tikzcd}[row sep=20, column sep=30,baseline = (E.base)]
      &
      \hatR \ar[bend right=2em]{dl}[swap]{\outl_\hatR} \ar[bend
      left=2em]{dr}{\outr_\hatR} \ar{d}{(\primrec\beta)_1}
      \\
      \mS \ar{d}[swap]{\gamma}
      &
      P  \ar{l}[swap]{\outl_{P}} \ar{d}{p} \ar{r}{\outr_{P}}
      &
      \mS \ar{d}{\wt{\gamma}}
      \\
      B(\mS,\mS)
      &
      \ar{l}[swap]{\outl_{\hatR,\hatR}} E_{\hatR,\hatR}
      \ar{d}[description,inner sep=0]{\barBs(r,(\mS,\hatR))_1}
      \ar{r}{\outr_{\hatR,\hatR}}
      &
      B(\mS,\mS)
      \\
      &
      |[alias = E]|
      E_{\mS,\hatR}  
      \ar[bend left=2em]{ul}{\outl_{\mS,\hatR}}
      \ar[bend right=2em]{ur}[swap]{\outr_{\mS,\hatR}} 
    \end{tikzcd}\tag*{\qedhere}
  \]
\end{proof}
We are ready to state our main result. Recall that we work under the
\Cref{asm:cat} and \ref{asm-sim}.
\begin{theorem}[Compositionality]\label{thm:compositionality-cat}
  Suppose that $\wt{\gamma}$ is a weakening of the operational model
  $(\mS,\gamma)$ such that $(\mS,\ini,\wt{\gamma})$ is a lax
  $\rho$-bialgebra. Then the weak similarity relation on
  $(\mS,\gamma)$ is a congruence.
\end{theorem}
\begin{proof}
  Let $(\mS,\lesssim)$ be the weak similarity relation on
  $(\mS,\gamma)$. Its Howe closure $(\mS,\happrox)$ satisfies
  \[
    (\mS,\lesssim)\leq (\mS,\happrox) \leq (\mS,\lesssim).
  \]
  The first inequality is witnessed by the morphism
  $\alpha_{\lesssim,\ini}\comp \inl$, for $\alpha_{\lesssim,\ini}$
  from~\eqref{eq:howe-alg-struct}. For the second one we use that the
  relation $(\mS,\happrox)$ is a weak simulation by \Cref{prop:howe}
  (note that it is reflexive and transitive by \Cref{lem:sim-props},
  so the proposition applies) and that $\lesssim$ is the greatest weak
  simulation. Thus $(\mS,\lesssim)\cong (\mS,\happrox)$, and since
  $(\mS,\happrox)$ is a congruence by \Cref{lem:howe-props}, we
  conclude that so is $(\mS,\lesssim)$.
\end{proof}

By choosing the trivial weakening $\wt{\gamma}=\gamma$ and
equipping~$B$ with the equality preorder, we obtain similarity as an
instance of weak similarity (\Cref{rem:weakening}\ref{rem:weakening-2}), and the laxness condition
on the bialgebra $(\mS,\ini,\gamma)$ holds trivially by
\eqref{diag:gamma}. We obtain

\begin{corollary}
Similarity on $(\mS,\gamma)$ is a congruence.
\end{corollary}
\noindent This is a variant of the main result of~\cite{gmstu23}. In fact, the present version is more general since
its notion of similarity is parametric in a lifting of $B$, while the
result in \emph{op.\ cit.}\ is about coalgebraic behavioural
equivalence, which corresponds to the \emph{canonical} lifting of $B$ (see Appendix, \Cref{sec:bifunctor-lift}).

\section{Applications}
\noindent We conclude with two applications of \Cref{thm:compositionality-cat}.
\subsection{$\HO$ Specifications}\label{S:HO-specs}
To recover the results of \Cref{sec:cool}, fix an $\HO$
specification~$\R$ over the signature $\Sigma$, corresponding to a
higher-order GSOS law $\rho^0$ of $\Sigma$ over $B_0(X,Y)=Y+Y^X$.  We
take $\C=\Set$ and instantiate the data of \Cref{asm-sim} to
\begin{enumerate}
\item the given polynomial functor $\Sigma$;
  
\item the behaviour functor $B(X,Y)=\Pow(Y+Y^X)$, preordered by
  inclusion, with its relation lifting $\barB$ as in
  \Cref{rem:weak-vs-strong}\ref{rem:weak-vs-strong-1};
  
\item the higher-order GSOS law $\rho$ of $\Sigma$ over $B$ given by
  \eqref{eq:rho0-to-rho}.
\end{enumerate}
\noindent It is not difficult to verify that the above data satisfies the
\Cref{asm-sim}. Then by choosing the weakening $\wt{\gamma}$ to be the
weak transition system associated to $\gamma_0$, see
\Cref{notation:ho-format}, we recover \Cref{thm:cool} as a special
case of \Cref{thm:compositionality-cat}.%
\hunote{Maybe mention nondeterministic systems $B(X,Y)=\Pow(Y+Y^X)$ or
  probabilistic systems $B(X,Y)=\mathcal{S}(Y+Y^X)$ for the
  subdistribution functor $\mathcal{S}$ as further applications? The
  latter also has an Elgi-Milner-type relation lifting, and the notion
  of similarity should be natural. LS: Cite Gavazzo if this is
  mentioned as future work}

\subsection{The $\lambda$-Calculus}\label{sec:lambda-sketch}
\noindent We briefly sketch how our framework applies to
the $\lambda$-calculus, building on ideas
from the 
work of Fiore et al.~\cite{DBLP:conf/lics/FiorePT99} and our previous work~\cite{gmstu23}. The (untyped call-by-name) $\lambda$-calculus is given by the
small-step operational rules shown below, where $s,s',t$ range over
possibly open $\lambda$-terms and $[t/x]$ denotes capture-avoiding
substitution.
\begin{equation}\label{eq:lambda-cbn}
  \begin{array}{l@{\qquad}l}
    \inference[\texttt{app1}]{\goes{s}{s'}} {\goes{s \app t}{s' \app t}}
    &
    \inference[\texttt{app2}]{}{\goes{(\lambda x.s) \app t}{s[t/x]}}
  \end{array}
\end{equation}
Take the presheaf
category $\C=\vcat$, where $\fset$ is the category of finite cardinals
and functions, and the functors
\begin{align*}
&\Sigma\c \C\to \C, && \Sigma X = V + \delta X + X\times X,\\
& B_0\c \C^\opp \times \C\to \C, &&  B_0(X,Y) =\llangle X,Y \rrangle \times (Y+Y^X+1).
\end{align*}
Here $Y^X$ denotes the exponential object in the topos $\vcat$, and
the presheaves $V$, $\delta X$ and $\llangle X,Y\rrangle$ are meant to represent variables, $\lambda$-abstraction, and
substitution, respectively:
\[
  V(n)=n,\; \delta X(n)= X(n+1),\; \llangle X,Y\rrangle(n) =
  \vcat(X^n,Y).
\]
The initial algebra for $\Sigma$ is the presheaf $\Lambda\in \vcat$ given by
\[
  \Lambda(n) =  \text{$\lambda$-terms in free variables from $n=\{0,\ldots,n-1\}$},  
\]
where $\lambda$-terms are formed modulo
$\alpha$-equivalence~\cite{DBLP:conf/lics/FiorePT99}.  In~\cite[Def.~5.8]{gmstu23} we devised a $V$-pointed higher-order GSOS
law $\rho^0$ of $\Sigma$ over $B_0$ whose operational model
\begin{equation}\label{eq:operational-model-lambda}
  \gamma^0
  =
  \langle \gamma^0_1,\gamma^0_2\rangle \c \Lambda\to \llangle \Lambda,\Lambda\rrangle \times (\Lambda+\Lambda^\Lambda +1) 
\end{equation}
captures the operational semantics of \eqref{eq:lambda-cbn}, in the
sense that for every $m,n\in \fset$, $t\in \Lambda(n)$ and
$\vec{u}\in \Lambda(m)^n$:
\begin{itemize}
\item $\gamma^0_1(t)(\vec{u}) = t[\vec{u}] = t[(u_0,\ldots,u_{n-1})/(0,\ldots,n-1)]$;
\item $t\to t'\implies\gamma^0_2(t)=t'\in \Lambda(n)$;
\item $t=\lambda x.t'\Longrightarrow \gamma^0_2(t)\in \Lambda^\Lambda(n)\wedge\forall e\in \Lambda(n).\gamma^0_2(t)(e)=t'[e]$;
\item $\gamma^0_2(t)=\ast$ otherwise (that is, if $t$ is stuck).
\end{itemize}
Here $\gamma^0_2(t)(e)=\ev(\gamma^0_2(t),e)$ for the evaluation map
$\ev\c \Lambda^\Lambda\times \Lambda\to \Lambda$. We instantiate the
data of \Cref{asm-sim} to
\begin{enumerate}
\item the above functor $\Sigma X = V + \delta X + X\times X$;
\item the functor $B(X,Y)=\llangle X,Y\rrangle \times
  \Pow_\star(Y+Y^X)$, where $\Pow_\star\c \vcat\to\vcat$ is the
  pointwise powerset functor given by $X\mapsto \Pow\comp X$ for $X \in\vcat$. (Note
  that we dropped the ``$+1$'' summand from the behaviour, whose role is
  taken over by the empty set.)  The relation lifting $\barB$ is
  constructed  similarly to the one of $B(X,Y)=\Pow(Y+Y^X)$ in
  \Cref{rem:weak-vs-strong}\ref{rem:weak-vs-strong-1};
\item a higher-order GSOS law $\rho$ of $\Sigma$ over $B$ which is derived
  from~$\rho^0$ in a way similar to the construction of
  \Cref{rem:sound-vs-lax-bialgebra}.
\end{enumerate}
The \emph{weak operational model} is the $B(\Lambda,-)$-coalgebra
\[
  \wt{\gamma}=\langle \wt{\gamma}_1, \wt{\gamma}_2\rangle \c
  \Lambda\to \llangle \Lambda,\Lambda\rrangle \times
  \Pow_\star(\Lambda+\Lambda^\Lambda)
\]
given for $t\in \Lambda(n)$ by $\wt{\gamma}_1(t)=\gamma_1^0(t)$ and
\begin{align*}
  \wt{\gamma}_2(t) =
  &\;\{\, \ol{t}\in \Lambda(n) : t\To \ol{t} \,\}\;\cup
  \\
  &\; \{\, f\in \Lambda^\Lambda(n) : \exists \ol{t}.\,t\To \ol{t}\, \wedge\, \gamma^0_2(\ol{t})=f  \,\}.
\end{align*}
Here $\To$ is the reflexive transitive hull of the reduction
relation~$\to$. Weak similarity on \eqref{eq:operational-model-lambda}
then can be shown to coincide with the following concept due to
Abramsky~\cite{Abramsky:lazylambda}:
\begin{definition}\label{def:app-sim}
  \emph{Applicative similarity} is the greatest relation $\lesssim^\ap_0 \,\subseteq\, \Lambda(0) \times \Lambda(0)$ on the set of closed $\lambda$-terms such that if
  $t_{1} \lesssim^\ap_0 t_{2}$ and  $t_{1} \To \lambda x.t_1'$, then there is a term $t_2'$ such that
\[
t_{2} \To \lambda x.t_2'\, \wedge\, \forall e \in
    \Lambda(0).~t_1'[e/x]\lesssim^\ap_0t_2'[e/x]. \]
Its \emph{open extension} is the relation $\lesssim^{\ap}\,\seq \Lambda\times \Lambda$ whose component $\lesssim^{\ap}_n\,\seq \Lambda(n)\times \Lambda(n)$ for $n>0$ is given by
\[ t_1 \lesssim^{\ap}_n t_2 \qquad \text{iff}\qquad t_1[\vec{u}] \lesssim^{\ap}_0 t_2[\vec{u}]\quad \text{for every $\vec{u}\in \Lambda(0)^n$}.\]
\end{definition}
One can verify that $(\Lambda,\ini,\wt{\gamma})$ forms a lax $\rho$-bialgebra, which amounts to observing that the rules \eqref{eq:lambda-cbn} are sound for weak transitions. Consequently, \Cref{thm:compositionality-cat} instantiates to a well-known and fundamental result about the $\lambda$-calculus~\cite{Abramsky:lazylambda}:

\begin{theorem}\label{thm:app-sim-cong}
The open extension of applicative similarity is a congruence: for all $\lambda$-terms $s,t,t'$, one has
\[ t\lesssim^{\ap} t'\implies s\app t\lesssim^{\ap} s\app t'\;\wedge\; t\app s\lesssim^{\ap} t'\app s \;\wedge\; \lambda x.t \lesssim^{\ap} \lambda x.t'. \]
\end{theorem}
It follows that the open extension of applicative \emph{bi}similarity, viz.~the relation ${\approx^\ap} = {\lesssim^\ap} \cap {\gtrsim^\ap}$, is also a congruence.  

\section{Conclusions and Future Work}
We have developed relation liftings of bifunctors and an abstract
analogue of Howe's method to prove congruence of
coalgebraic weak similarity for higher-order GSOS laws. We have thus taken the first
steps towards operational reasoning in the higher-order
abstract GSOS framework.

Logical relations~\cite{tait1967intensional,
  DBLP:journals/iandc/Statman85,DBLP:journals/iandc/OHearnR95,
  DBLP:journals/corr/abs-1103-0510} are another important operational
reasoning technique that we would like to cover in the future. Logical
relations are typically type-indexed, while higher-order abstract GSOS
has so far been applied to untyped languages. We aim to investigate
typed languages in the context of higher-order abstract GSOS and
develop abstract analogues of logical relations. It is worth noting
that, even in the untyped setting, relation liftings of bifunctors
already share a key characteristic with logical relations, namely that
functions send related inputs to related outputs
(\Cref{rem:weak-vs-strong}\ref{rem:weak-vs-strong-1}).

Another goal is to apply our methods to call-by-value languages. As
already noted in our previous work~\cite[Sec.~5.4]{gmstu23}, this
appears to be more subtle than the call-by-name case. We envision a
{multi-sorted} setting as a possible approach.

Finally, we aim to explore effectful languages. For instance, by
taking the behaviours $\Pow(Y+Y^X)$ or $\mathcal{S}(Y+Y^X)$,
where~$\mathcal{S}$ is the subdistribution functor, our results
already yield a form of compositionality for nondeterministic and
probabilistic combinatory logic. For the latter, exploring behavioural
distances instead of (bi)similarity is also a natural direction; we expect that existing work on probabilistic $\lambda$-calculi~\cite{cd17,gavazzo18} can provide some guidance.
\bibliographystyle{IEEEtranS}
\bibliography{mainBiblio-short}

\appendices
\onecolumn

\section*{Appendix}
\noindent This appendix is structured as follows:
\begin{itemize}
\item \Cref{sec:more-on-graphs-and-relations} establishes a number of auxiliary results on graphs and relations.
\item \Cref{sec:omitted-proofs} provides complete proofs of all results from the extended abstract.
\item \Cref{sec:bifunctor-lift} explains how to construct canonical graph and relation liftings of mixed-variance bifunctors.
\item \Cref{sec:can-lift-gsos-laws} addresses canonical graph and relation liftings of higher-order GSOS laws.
\item \Cref{sec:lambda} is a more detailed version of \Cref{sec:lambda-sketch} on the $\lambda$-calculus.
\end{itemize}
Recall that throughout our paper, including this appendix, we work under the \Cref{asm:cat} on the base category $\C$.

\section{More on Graphs and Relations}\label{sec:more-on-graphs-and-relations}
In \Cref{sec:graph-cocomplete} we stated that the category $\Gra(\C)$ is complete and cocomplete, with limits and colimits formed at the level of $\C$. In particular, this means that the product $(X,R)\times (Y,S)$ and coproduct $(X,R)+(Y,S)$ are given by the graphs
  \[
    \begin{tikzcd}[row sep=2em]
      R\times S
      \ar[bend right=2em]{d}[swap]{\outl_R\times\outl_{S}}
      \ar[bend left=2em]{d}{\outr_S\times \outr_{S}}
      \\
      X\times Y
    \end{tikzcd}
\qqand 
    \begin{tikzcd}[row sep=2em]
      R+S
      \ar[bend right=2em]{d}[swap]{\outl_R+\outl_{S}}
      \ar[bend left=2em]{d}{\outr_R+\outr_{S}}
      \\
      X+Y
    \end{tikzcd}
  \] 
Colimits in $\Gra_X(\C)$ are also formed
 at the level of $\C$. In particular, the coproduct of $(X,R), (X,R') \in \Gra_X(\C)$, denoted by $(X,R)+_X (X,R')$, is the graph
  \[
    \begin{tikzcd}[row sep=2em]
      R+R'
      \ar[bend right=2em]{d}[swap]{[\outl_R,\outl_{R'}]}
      \ar[bend left=2em]{d}{[\outr_R,\outr_{R'}]}
      \\
      X
    \end{tikzcd}
  \] 
(Co-)limits in $\Rel(\C)$ and $\Rel_X(\C)$ are formed by taking the (co-)limit in $\Gra(\C)$ and $\Gra_X(\C)$, respectively, and applying the reflector $(-)^\dag$. Note that $\Rel(\C)$ is closed under products in $\Gra(\C)$, since products of monomorphisms are monomorphisms. 

%
The remaining results of this section are about composition of graphs and relations. We first mention a useful property of pullbacks that follows from our assumptions:

\begin{lemma}\label{lem:pullback-stable}
For every commutative diagram
    \eqref{eq:pullback-stable-strong-epis}, if the outside and the inner square are pullbacks and $e_0,e_1$ are strong epimorphisms, then $e$ is a strong epimorphism.
    \begin{equation}\label{eq:pullback-stable-strong-epis}
      \begin{tikzcd}[column sep=1em, row sep=1em]
        A \ar{rrr}{f} \ar{ddd}[swap]{g} \ar{dr}[inner sep=1]{e} & & &
        B \ar{ddd}{h} \ar{dl}[swap,inner sep=1]{e_0}  \\
        & A' \ar{r}{f'} \ar{d}[swap]{g'}  & B' \ar{d}{h'}  &  \\
        & C' \ar{r}{k'}  & D  & \\
        C \ar{ur}[inner sep=1]{e_1} \ar{rrr}{k} & & & D
        \ar{ul}[swap,inner sep=1]{\id}
      \end{tikzcd}
    \end{equation}
\end{lemma}

\begin{proof}
Recall that pullbacks correspond to products in a slice category. Therefore, the lemma states that for any two morphisms $e_0\colon h\to h'$ and $e_1\colon g\to g'$ in $\C/D$ with $e_0,e_1$ strongly epic in $\C$, their product $e=e_0\times_D e_1$ in $\C/D$ is also a strong epimorphism in $\C$. Since $e_0\times_D e_1 = (\id\times_D e_1) \circ (e_0\times_D \id)$, and strong epimorphisms are closed under composition, it suffices to assume that one of the $e_i$ is an identity morphism, w.l.o.g.\ $e_1=\id$. Then we have the following commutative diagram:
    \begin{equation}\label{eq:pullback-stable-strong-epis}
	\begin{tikzcd}
         A \ar[shiftarr={xshift=-3em}]{dd}[swap]{g} \ar{r}{f} \ar{d}[swap]{e} & B \ar{d}{e_0} \ar[shiftarr ={xshift=3em}]{dd}{h} \\
         A' \ar{r}{f'} \ar{d}[swap]{g'}  & B' \ar{d}{h'}    \\
         C=C' \ar{r}{k=k'}  & D   \\
      \end{tikzcd}
    \end{equation}
The lower rectangle and the outside are pullbacks, so the upper rectangle is also a pullback~\cite[Prop.~2.5.9]{borceux94}. Since $e_0$ is strongly epic and strong epimorphisms are pullback-stable by Assumption \ref{asm:cat}\ref{asm:cat-3}, we conclude that $e$ is strongly epic. 
\end{proof}

\begin{notation}
 Recall that the graph $(X,R)\smc (X,R')$ is defined by via the pullback \eqref{eq:composite-graph}. We often write $RR'$ for $R\smc R'$.
\end{notation}
Graph composition extends to a bifunctor
\[ (-)\smc(-)\c \Gra_X(\C)\times \Gra_X(\C)\to \Gra_X(\C) \]
as follows. Given $f\c (X,R)\to (X,S)$ and $f'\c (X,R')\to (X,S')$ in $\Gra_X(\C)$, the morphism
\[f\smc f'\c (X,R)\smc (X,R')\to (X,S)\smc (X,S') \]
is defined by $(f\c f')_0=\id_{X}$, and $(f\smc f')_1$  is the unique $\C$-morphism making the two squares below commute, using the universal property of the pullback $SS'$:
\begin{equation}\label{diag:fcomp}
  \begin{tikzcd}
    R R'  \ar{d}[swap]{\pi_{R R',R}} \ar[dashed]{r}{(f\smc f')_1} & S S' \ar{d}{\pi_{S S',S}} \\
    R \ar{r}{f_1} & S
  \end{tikzcd}
  \quad\quad
  \begin{tikzcd}
    R R'  \ar{d}[swap]{\pi_{R R',R'}} \ar[dashed]{r}{(f\smc f')_1} & S S' \ar{d}{\pi_{S S',S'}} \\
    R' \ar{r}{f_1'} & S'
  \end{tikzcd}
\end{equation}

\begin{lemma}\label{lem:strong-epi-graph-comp}
If $f_1$ and $f_1'$ are strong epimorphisms, then $(f\smc f')_1$ is a strong epimorphism.
\end{lemma}

\begin{proof}
This follows from \Cref{lem:pullback-stable} applied to the commutative diagram
\[
\begin{tikzcd}
RR' \ar{rrr}{\pi_{RR',R'}} \ar{ddd}[swap]{\pi_{RR',R}} \ar{dr}{(f\smc f')_1} & & & R' \ar{ddd}{\outl_{R'}} \ar{dl}[swap]{f_1'}  \\
& SS' \ar{r}{\pi_{SS',S'}} \ar{d}[swap]{\pi_{SS',S}} & S' \ar{d}{\outl_{S'}}  &  \\
& S \ar{r}{\outr_S}  & X  & \\
R \ar{ur}{f_1} \ar{rrr}{\outr_R} & & & X \ar{ul}[swap]{\id}
\end{tikzcd}
\]
\end{proof}
Analogously, composition of relations extends to the bifunctor
\[ (-)\bullet (-) \;=\; (\,\Rel_X(\C)\times \Rel_X(\C) \subto \Gra_X(\C)\times \Gra_X(\C) \xto{(-)\smc (-)} \Gra_X(\C) \xto{(-)^\dag} \Rel_X(\C)\,). \]
\begin{lemma}\label{lem:reflection-comp}
For every $(X,R),(X,R')\in \Gra_X(\C)$, we have 
\[(X,R)^\dag\bullet (X,R')^\dag \;\cong\; ((X,R)\smc (X,R'))^\dag.\]
\end{lemma}

\begin{proof}
Since $(\id,e_R)\c (X,R)\to (X,R^\dag)$ and $(\id,e_{R'})\c (X,{R'})\to (X,{R'}^\dag)$ are graph morphisms and by \eqref{diag:Rdag}, the following diagram commutes:
\[
\begin{tikzcd}[column sep=15em, row sep=4em]
R\smc {R'} \ar[two heads]{d}[swap]{e_{R\smc {R'}}} \ar[two heads]{r}{((\id,e_R)\smc (\id,e_{R'}))_1} \ar{drr}[swap]{\langle \outl_{R\smc {R'}},\outr_{R\smc {R'}} \rangle}  & R^\dag\smc {R'}^\dag \ar[two heads]{r}{e_{R^\dag\smc {R'}^\dag}} \ar{dr}[description]{\langle \outl_{R^\dag\smc {R'}^\dag},\outr_{R^\dag\smc {R'}^\dag} \rangle} & (R^\dag\smc {R'}^\dag)^\dag \ar[>->]{d}[description]{\langle \outl_{(R^\dag\smc {R'}^\dag)^\dag},\outr_{(R^\dag\smc {R'}^\dag)^\dag} \rangle} \\
(R\smc {R'})^\dag \ar[>->]{rr}{\langle \outl_{(R\smc {R'})^\dag},\outr_{(R\smc {R'})^\dag} \rangle} & & X\times X
\end{tikzcd}
\]
Note that $((\id,e_R)\smc (\id,e_{R'}))_1$ is a strong epimorphism by \Cref{lem:strong-epi-graph-comp}. Thus the desired isomorphism follows from the uniqueness of (strong epi, mono)-factorizations.
\end{proof}

The following lemma shows that graph and relation composition
distribute over coproducts. This is where our assumption that $\C$ is locally distributive is used. 
\begin{lemma}\label{lem:comp-dist-over-coproducts}
For $X\in \C$ we have the following natural isomorphisms in $\Gra_X(\C)$ and $\Rel_X(\C)$, respectively:
\begin{align}
(X,R)\smc ((X,S)+(X,T))& \cong (X,R)\smc (X,S) + (X,R)\smc (X,T),\label{eq:dist-1}\\
((X,S)+(X,T))\smc (X,R)& \cong (X,S)\smc (X,R) + (X,T)\smc (X,R),\label{eq:dist-2} \\
(X,R)\bullet ((X,S)\vee (X,T))& \cong (X,R)\bullet (X,S) \vee (X,R)\bullet (X,T),\label{eq:dist-3}\\
((X,S)\vee (X,T))\bullet (X,R)& \cong (X,S)\bullet (X,R) \vee (X,T)\bullet (X,R). \label{eq:dist-4}
\end{align}
Here we write $+$ and $\vee$ for $+_X$ and $\vee_X$. 
\end{lemma}

\begin{proof}
We construct the isomorphisms \eqref{eq:dist-1} and \eqref{eq:dist-3}; the other two are analogous.

\medskip\noindent\emph{Proof of \eqref{eq:dist-1}.}
By the universal property of the pullback $R(S+T)$, there exists a
unique $\C$-morphism $i\c RS\to R(S+T)$ making the following diagrams commute:
\begin{equation}\label{diag:i}
\begin{tikzcd}
RS \ar[dashed]{rr}{i} \ar{dr}[swap]{\pi_{RS,R}} & & R(S+T) \ar{dl}{\pi_{R(S+T),R}} \\
  & R &
\end{tikzcd}
\quad\quad
\begin{tikzcd}
RS \ar[dashed]{rr}{i} \ar{d}[swap]{\pi_{RS,S}} & & R(S+T) \ar{dl}{\pi_{R(S+T),S+T}} \\
S  \ar{r}{\inl} & S+T &
\end{tikzcd}
\end{equation}
The morphism $j\c RT\to R(S+T)$ is defined analogously. Since the slice category $\C/X$ is distributive, and products in $\C/X$ correspond precisely to pullbacks in $\C$, it follows that 
\[[i,j]\c RS+RT\xto{\cong} R(S+T)\]
is an isomorphism in $\C$. Thus
\[ (\id_X,[i,j])\c (X,R)\smc (X,S) + (X,R)\smc (X,T) \xto{\cong}  (X,R)\smc ((X,S)+(X,T)) \]
is an isomorphism in $\Gra_X(\C)$; the two diagrams below show that is
indeed a $\Gra_X(\C)$-morphism. (Below and henceforth we indicate the
reason why the parts of a commutative diagram commutes, and we write
$\commu$ if a part obviously commutes.) 

\[
  \begin{tikzcd}
    RS + RT
    \ar{rrr}{[i,j]}
    \ar{dr}[description]{\pi_{RS,R}+\pi_{RT,R}}
    \ar{dd}[description]{[\outl_{RS},\outl_{RT}]}
    & & { }&
    R(S+T) \ar{dd}[description]{\outl_{R(S+T)}}
    \ar{dl}[description]{\pi_{R(S+T),R}}
    \\
    { }& R+R
    \descto{ru}{\eqref{diag:i}}
    \descto[pos=.2]{l}{\eqref{eq:composite-graph}}
    \descto{rd}{\commu}
    \ar{dl}[description]{[\outl_R,\outl_R]} \ar{r}{\nabla}
    &
    R \ar{dr}[swap]{\outl_R}
    \descto[pos=.4]{r}{\eqref{eq:composite-graph}}
    &
    { }
    \\
    X \ar{rrr}{\id} & & { } & X
\end{tikzcd}
\quad\quad
\begin{tikzcd}[column sep=4em]
  RS + RT \ar{rr}{[i,j]}
  \ar{dr}[description]{\pi_{RS,S}+\pi_{RT,T}}
  \ar{dd}[description]{[\outr_{RS},\outr_{RT}]}
  &{ } &
  R(S+T) \ar{dd}[description]{\outr_{R(S+T)}}
  \ar{dl}[description]{\pi_{R(S+T),S+T}}
  \\
  { }
  &
  S+T
  \descto{u}{\eqref{diag:i}}
  \descto{d}{\commu}
  \descto[pos=.2]{l}{\eqref{eq:composite-graph}}
  \descto[pos=.4]{r}{\eqref{eq:composite-graph}}
  \ar{dl}[description]{[\outr_S,\outr_T]}
  \ar{dr}[description]{[\outr_S, \outr_T]}
  &
  { }
  \\
  X \ar{rr}{\id} & { }  & X
\end{tikzcd}
\]
Naturality of the isomorphism $(\id_X,[i,j])$ amounts to showing that
the square below commutes for all $\Gra_X(\C)$-morphisms $f\c (X,R)\to
(X,R')$, $g\c (X,S)\to (X,S')$ and $h\c (X,T)\to (X,T')$, where
$[i',j']$ is defined analogously to $[i,j]$ above.
\[
\begin{tikzcd}
R S+R T \ar{r}{[i,j]} \ar{d}[swap]{(fg)_1+(fh)_1} & R (S+T) \ar{d}{(f (g+h))_1} \\
R' S'+R' T' \ar{r}{[i',j']} & R' (S'+T') 
\end{tikzcd}
\] 
It suffices to show that the square commutes when postcomposed with the pullback projections $\pi_{R' (S'+T'),R'}$ and $\pi_{R' (S'+T'),S'+T'}$; the latter follows from the two commutative diagrams below. 
\[
\begin{tikzcd}[outer xsep=0, column sep=8.5,baseline = (B.base)]
  R S+R T \ar{rrr}{[i,j]} \ar{dr}[description]{\pi_{R S,R}+\pi_{R
      T,R}} \ar{ddd}[description]{(fg)_1+(fh)_1}
  & & { } &
  R (S+T) \ar{ddd}[description]{(f(g+h))_1}
  \ar{dl}[description]{\pi_{R (S+T),R}}
  \\
  { }
  &
  R+R
  \descto{ru}{\eqref{diag:i}}
  \descto{rd}{\commu}
  \descto{l}{\eqref{diag:fcomp}}
  \ar{d}[swap]{f_1+f_1} \ar{r}{\nabla}
  &
  R
  \descto{r}{\eqref{diag:fcomp}}
  \ar{d}{f_1}
  &
  { }
  \\
  { }
  &
  R'+R'
  \descto{rd}{\eqref{diag:i}}
  \ar{r}{\nabla} & R' & { }
  \\
  |[alias = B]|
  R' S'+R' T' \ar{rrr}{[i',j']} \ar{ur}[description]{\pi_{R'
      S',R'}+\pi_{R' T,R'}}
  & &{ }
  & R' (S'+T') \ar{ul}[description]{\pi_{R' (S'+T'),R'}} 
\end{tikzcd}
\begin{tikzcd}[column sep = 20,baseline = (B.base)]
  R S+R T \ar{rr}{[i,j]} \ar{dr}[description]{\pi_{R S,S}+\pi_{R T,T}}
  \ar{ddd}[description]{(fg)_1+(fh)_1}
  &{ } &
  R (S+T) \ar{ddd}[description]{(f (g+h))_1}
  \ar{dl}[description]{\pi_{R (S+T),S+T}}
  \\
  { }
  &
  S+T \ar{d}{g_1+h_1}
  \descto{u}{\eqref{diag:i}}
  \descto{l}{\eqref{diag:fcomp}}
  \descto{r}{\eqref{diag:fcomp}}
  &
  { }
  \\
  &
  S'+T'
  \descto[pos=.1]{d}{\eqref{diag:i}}
  & \\
  |[alias = B]|
  R' S'+R' T'
  \ar{rr}{[i',j']}
  \ar[pos=.4]{ur}[description]{\pi_{R'    S',R'}+\pi_{R' T,R'}}
  &{ } &
  R' (S'+T')
  \ar[pos=.4]{ul}[description]{\pi_{R' (S'+T'),S'+T'}} 
\end{tikzcd}
\] 
\medskip\noindent \emph{Proof of \eqref{eq:dist-3}.} For every $(X,R), (X,S), (X,T)\in \Rel_X(\C)$ we compute
\begin{align*}
(X,R)\bullet ((X,S)\vee (X,T))&\cong (X,R)^\dag\bullet ((X,S)+(X,T))^\dag & \\
&\cong ((X,R)\smc ((X,S)+(X,T)))^\dag \\
&\cong ((X,R)\smc (X,S) +(X,R)\smc (X,T))^\dag \\
&\cong ((X,R)\smc (X,S))^\dag \vee ((X,R)\smc (X,T))^\dag \\
&\cong (X,R)\bullet (X,S) \vee (X,R)\bullet (X,T).
\end{align*}
In the first step we use the definition of $\vee$ and that $(X,R)^\dag \cong (X,R)$ because $(X,R)$ is a relation; the second step follows from \Cref{lem:reflection-comp}; the third step uses the isomorphism \eqref{eq:dist-1} established above; the fourth step uses that the reflector $(-)^\dag$ preserves coproducts, being a left adjoint; the last step follows by definition of $\bullet$. 
\end{proof}
Recall from \Cref{sec:transfer} the reindexing functor $f_\star\c \Gra_X(\C)\to \Gra_Y(\C)$ induced by a morphism $f\c X\to Y$. 
\begin{lemma}\label{lem:f-star-respects-reflection}
For every $f\c X\to Y$ in $\C$ and $(X,R),(X,R')\in \Gra_X(\C)$ we have
\begin{align}
 f_\star((X,R)\smc (X,R')) &\leq f_\star(X,R)\smc f_\star(X,R'), \label{eq:f-star-respects-comp}\\
 f_\star((X,R)^\dag) &\leq (f_\star(X,R))^\dag. \label{eq:f-star-respects-dag}
\end{align}
\end{lemma}

\begin{proof}
\emph{Proof of \eqref{eq:f-star-respects-comp}.} By definition, $f_\star(X,R)\smc f_\star(X,R')$ is the relation $(Y,S)$ obtained via the following pullback:
\[
\begin{tikzcd}[column sep=3em]
  & &
  S 
  \pullbackangle{-90}
  \ar{dl}[swap]{\pi_{S,R}}
  \ar{dr}{\pi_{S,R'}}
  \ar[shiftarr = {xshift=-40}]{ddll}[swap]{\outl_{S}}
  \ar[shiftarr = {xshift=40}]{ddrr}{\outr_{S}}
  &  & \\
& R \ar{dl}[description]{f\comp\outl_R} \ar{dr}[description]{f\comp \outr_R} & & R' \ar{dl}[description]{f\comp \outl_{R'}} \ar{dr}[description]{f\comp \outr_{R'}}  & \\
Y & & Y & & Y 
\end{tikzcd}
\]
Since $\outr_R,\outl_{R'}$ merge $\pi_{R\smc R',R}$ and $\pi_{R\smc R',R'}$, so do $f\comp \outr_R,f\comp \outl_{R'}$, hence the universal property of the pullback $S$ yields a unique $h_1\c R\smc R'\to S$ such that 
\[\pi_{R\smc R',R}=\pi_{S,R}\comp h_1\qqand \pi_{R\smc R',R'}=\pi_{S,R'}\comp h_1.\]
It follows that
\[ (\id,h_1)\c f_\star((X,R)\smc (X,R')) \to f_\star(X,R)\smc f_\star(X,R')  \]
is a $\Gra_Y(\C)$-morphism, as required.

\medskip\noindent\emph{Proof of \eqref{eq:f-star-respects-dag}.} By definition, $(f_\star(X,R))^\dag$ is the relation $(Y,S)$ obtained via the (strong epi, mono)-factorization $e$ and $\langle \outl_S,\outr_S \rangle$ of the morphism $\langle f\comp \outl_R, f\comp \outr_R\rangle= (f\times f)\comp \langle \outl_R,\outr_R\rangle$. We thus obtain a diagonal fill-in $d\c R^\dag\to S$ making the diagram below commute:
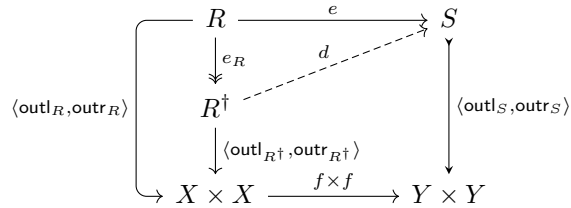
\[
\begin{tikzcd}[column sep=5em]
R \ar{r}{e} \ar[shiftarr = {xshift=-30}]{dd}[swap]{\langle \outl_R,\outr_R\rangle} \ar[two heads]{d}{e_R} & S \ar[>->]{dd}{\langle \outl_S, \outr_S \rangle} \\
R^\dag \ar[dashed]{ur}{d} \ar{d}{\langle \outl_{R^\dag},\outr_{R^\dag}\rangle} & \\
X\times X \ar{r}{f\times f} & Y\times Y 
\end{tikzcd}
\]
 It follows that
\[(\id,d)\c f_\star((X,R)^\dag)\to (f_\star(X,R))^\dag\]
is a $\Gra_Y(\C)$-morphism, as required.
\end{proof}

\begin{corollary}\label{eq:f-star-respects-rel-comp}
For every $f\c X\to Y$ in $\C$ and $(X,R),(X,R')\in \Rel_X(\C)$ we have
\begin{align}
 f_\star((X,R)\bullet (X,R')) & \leq f_\star(X,R)\bullet f_\star(X,R'). \label{eq:f-star-resp-comp-rel}
\end{align}
\end{corollary}

\begin{proof}
This follows from the computation
\begin{align*}
 f_\star((X,R)\bullet (X,R')) &= f_\star(((X,R)\smc (X,R'))^\dag) \\
&\leq (f_\star((X,R)\smc (X,R'))^\dag \\
&\leq (f_\star(X,R)\smc f_\star(X,R'))^\dag \\
&=f_\star(X,R)\bullet f_\star(X,R'). 
\end{align*}
The second step uses \eqref{eq:f-star-respects-dag} and the third one uses \eqref{eq:f-star-respects-comp} and functoriality of $(-)^\dag$.
\end{proof}

\section{Omitted proofs}\label{sec:omitted-proofs}

\section*{Proof of \Cref{prop:cool-vs-sound}}
The rules \eqref{eq:passive-1} for a passive operator are clearly
sound, as is the first rule in \eqref{eq:active-1} for an active
operator. To see that the second rule in \eqref{eq:active-1} is sound, suppose that $p_1,\ldots,p_n,\ol{p_j}\in \mS$ are programs such that $p_j\Downarrow \ol{p_j}$. This means that
\[ p_j=p_j^0\to p_j^1 \to \cdots \to p_j^k = \ol{p_j}\qquad \text{for some $k\geq 0$ and $p_j^0,\ldots,p_j^k\in \mS$}.  \] 
By applying the first of \eqref{eq:active-1} (the first ones $k$
times) 
we obtain
\begin{align*}
 \f(p_1,\ldots,p_j,\ldots,p_n)\to \cdots & \to \f(p_1,\ldots \ol{p_j}, \ldots,p_n)    \\ &\to t[p_1,\ldots, p_{j-1}, p_{j+1},\ldots, p_n, (\ol{p_j})_{p_1},\ldots, (\ol{p_j})_{p_{j-1}},(\ol{p_j})_{p_{j+1}},\ldots, (\ol{p_j})_{p_n}].  
\end{align*}
(The bracket $[\cdots]$ indicates how the variables of $t$ are substituted. Note that $t$ does not depend on $x_j$ and $y_j^{x_j}$.) It follows $\f(p_1,\ldots,p_n)\To t[p_1,\ldots, p_{j-1}, p_{j+1},\ldots, p_n, (\ol{p_j})_{p_1},\ldots, (\ol{p_j})_{p_{j-1}},(\ol{p_j})_{p_{j+1}},\ldots, (\ol{p_j})_{p_n}]$, as required. The soundness of the third rule in \eqref{eq:active-1} is shown analogously.

\section*{Proof of \Cref{prop:free-monad-lift}}

We first establish some auxiliary results.

\begin{rem}\label{rem:c-props}
By putting $D=1$ in \Cref{lem:pullback-stable}, we see that strong epimorphisms in $\C$ are product-stable: $e_0\times e_1$ is
  strongly epic whenever $e_0$ and $e_1$ are strongly epic.
\end{rem}

\begin{lemma}\label{lem:lift-dag-prod}
If $\Sigma$ preserves strong epimorphisms, we have the following isomorphism for every $(X,R),(Y,S)\in \Gra(\C)$:
\[ \ol{\Sigma}_\Rel((X,R)^\dag\times (Y,S)^\dag) \;\cong\; (\ol{\Sigma}_\Gra((X,R)\times (Y,S))^\dag. \]
\end{lemma}

\begin{proof}
This follows from the commutative diagram below and the uniqueness of (strong epi, mono)-factorizations. Note that $e_R\times e_S$ is strongly epic by \Cref{rem:c-props}.
\[
\begin{tikzcd}[column sep=13em, row sep=8em]
\Sigma (R\times S) \ar[two heads]{d}[swap]{e_{\Sigma(R\times S)}} \ar{drr}[swap]{\langle \Sigma(\outl_R\times \outl_S), \Sigma(\outr_R\times \outr_S)\rangle} \ar[two heads]{r}{\Sigma (e_{R}\times e_S)} & \Sigma (R^\dag\times S^\dag) \ar[two heads]{r}{e_{\Sigma(R^\dag\times S^\dag)}} \ar{dr}[description]{\langle \Sigma\outl_{R^\dag\times S^\dag}, \Sigma \outr_{R^\dag\times S^\dag}\rangle } & (\Sigma (R^\dag\times S^\dag))^\dag \ar{d}[description]{\langle \outl_{(\Sigma(R^\dag\times S^\dag))^\dag}, \outr_{(\Sigma(R^\dag\times S^\dag))^\dag}  \rangle} \\
(\Sigma(R\times S))^\dag \ar{rr}{\langle \outl_{(\Sigma (R\times S))^\dag}, \outr_{(\Sigma (R\times S)^\dag} \rangle} & & \Sigma (X\times Y) \times \Sigma (X\times Y)
\end{tikzcd}
\]
\end{proof}

\begin{lemma}\label{lem:ol-sigma-vs-dag}
If $\Sigma$ preserves strong epimorphisms, the following diagrams commute:
\[
\begin{tikzcd}
\Gra(\C) \ar{r}{\ol{\Sigma}_\Gra} \ar{d}[swap]{(-)^\dag} & \Gra(\C) \ar{d}{(-)^\dag} \\
\Rel(\C) \ar{r}{\ol{\Sigma}_\Rel} & \Rel(\C) 
\end{tikzcd}
\]
\end{lemma}

\begin{proof}
For each $(X,R)\in \Gra(\C)$ we have
\[
(\ol{\Sigma}_\Gra(X,R))^\dag = (\Sigma X,\Sigma R)^\dag \cong (\Sigma X, (\Sigma R^\dag))^\dag  = \ol{\Sigma}_\Rel((X,R)^\dag) 
\]
The isomorphism in the second step follows from the commutative diagram below and the uniqueness of (strong epi, mono)-factorizations. 
\[
\begin{tikzcd}[column sep=12em, row sep=5em]
\Sigma R \ar[two heads]{d}[swap]{e_{\Sigma R}} \ar{drr}[swap]{\langle \Sigma\outl_R, \Sigma\outr_R\rangle} \ar[two heads]{r}{\Sigma e_{R}} & \Sigma R^\dag \ar[two heads]{r}{e_{(\Sigma R^\dag)^\dag}} \ar{dr}[description]{\langle \Sigma\outl_{R^\dag}, \Sigma \outr_{R^\dag}\rangle } & (\Sigma R^\dag)^\dag \ar{d}[description]{\langle \outl_{(\Sigma R^\dag)^\dag}, \outr_{(\Sigma R^\dag)^\dag}  \rangle} \\
(\Sigma R)^\dag \ar{rr}{\langle \outl_{(\Sigma R)^\dag}, \outr_{(\Sigma R)^\dag} \rangle} & & \Sigma X \times \Sigma X
\end{tikzcd}
\]
\end{proof}
Since $(-)^\dag$ is a left adjoint, the above lemma and \cite[Thm.~2.14]{hj98} yield

\begin{corollary}\label{cor:free-algebra-adjunction}
If $\Sigma$ preserves strong epimorphisms, there is an adjunction
\[
\begin{tikzcd}
\Alg(\ol{\Sigma}_\Gra) \ar[phantom]{r}{\bot} \ar[bend left=1em]{r}{F} & \Alg(\ol{\Sigma}_\Rel) \ar[bend left=1em]{l}{U}
\end{tikzcd}
\]
where the right adjoint $U$ is given by
\[  
(\ol{\Sigma}_\Rel(A,R)\xto{a} (A,R)) \;\;\mapsto\;\; (\ol{\Sigma}_\Gra(A,R) \xto{(\id,e_{\Sigma R})} \ol{\Sigma}_\Rel(A,R)\xto{a} (A,R))\qqand h\mapsto h,   
\]
the left adjoint $F$ by
\[ (\ol{\Sigma}_\Gra(A,R)\xto{a} (A,R)) \;\;\mapsto\;\;  (\ol{\Sigma}_\Rel(A,R^\dag) \xto{a^\dag} (A,R^\dag)) \qqand h\mapsto h^\dag, \]
and the unit at $((A,R),a)\in \Alg(\ol{\Sigma}_\Gra)$ by
\[ (\id,e_{R})\c ((A,R),a)\to UF((A,R),a). \]
\end{corollary}
With this preparation, we are ready to prove the proposition.
\begin{proof}[Proof of \Cref{prop:free-monad-lift}]
Let $(\Sigmas,\eta,\mu)$ be the free monad generated by $\Sigma$, with 
associated natural transformation \[\iota\c \Sigma\Sigmas\to
  \Sigmas;\] that is, $\iota_X\c \Sigma\Sigmas X\to \Sigmas X$ is the
free $\Sigma$-algebra on $X$ with the universal morphism $\eta_X\c X
\to \Sigmas X$.
\begin{enumerate} 
\item\label{lem:free-monad-lift-1} We first consider the graph lifting $\ol{\Sigma}=\ol{\Sigma}_\Gra$. It suffices to show that a free $\ol{\Sigma}$-algebra on $(X,R)\in \Gra(\C)$ is given by 
\[ 
  \ol{\Sigma}\,\barSigmas(X,R)=(\Sigma\Sigmas X,\Sigma\Sigmas R)
  \xra{(\iota_X,\iota_R)} 
(\Sigmas X, \Sigmas R) = \barSigmas(X,R) 
\] 
with the universal morphism
\[ 
(X,R)
\xra{(\eta_X,\eta_R)}
 (\Sigmas X,\Sigmas R)=\barSigmas(X,R). 
\]
To prove this, we verify the required universal property. Thus suppose that we are given a $\ol{\Sigma}$-algebra $a\c \ol{\Sigma}(A,S)\to (A,S)$ and a $\Gra(\C)$-morphism $h\c (X,R)\to (A,S)$. The universal property of the free $\Sigma$-algebra $(\Sigmas X,\iota_X)$ yields a unique $\Sigma$-algebra morphism
\[ \ol{h}_0\c (\Sigmas X,\iota_X)\to (A,a_0)\qquad\text{such that}\qquad \ol{h}_0\circ \eta_X=h_0, \]
and similarly the universal property of $(\Sigmas R,\iota_R)$ yields a unique $\Sigma$-algebra morphism
\[ \ol{h}_1\c (\Sigmas R,\iota_R) \to (S,a_1)\qquad\text{such that}\qquad \ol{h}_1\circ \eta_R=h_1.\]
This implies that the $\ol{\Sigma}$-algebra morphism
\[\ol{h}=(\ol{h}_0,\ol{h}_1)\c (\Sigmas(X,R),(i_X,i_R))\to ((A,S),a)\]
satisfies $\ol{h}\circ (\eta_X,\eta_R) = h$, and it is clearly unique with that property. 
\item Now consider the relation lifting ${\widetilde{\Sigma}}=\ol{\Sigma}_\Rel$. It suffices to show that a free $\widetilde{\Sigma}$-algebra on $(X,R)\in \Rel(\C)$ is given by
\[  
\widetilde{\Sigma}\widetilde{\Sigmas}(X,R) \xto{(\iota_X,\iota_R)^\dag} \widetilde{\Sigmas}(X,R)
\]
with the universal morphism 
\[ (X,R) \xto{(\eta_X,\eta_R)^\dag} \widetilde{\Sigmas}(X,R). \]
But this is immediate from part \ref{lem:free-monad-lift-1} above and \Cref{cor:free-algebra-adjunction}.\qedhere
\end{enumerate}
\end{proof}

\section*{Proof of \Cref{lem:sim-props}}
We prove a more refined result:
\begin{lemma}
Suppose that the functor $\barF$ satisfies the following conditions for all $X\in \C$ and $(X,R),(X,S)\in \Rel_X(\C)$:
\begin{enumerate}
\item the relation $\barF(X,X)$ is reflexive;
\item $\barF(X,R)\bullet \barF(X,S)\leq \barF((X,R)\bullet (X,S))$. 
\end{enumerate}Then for every coalgebra $c\c C\to FC$:
\begin{enumerate}[label=(\alph*)]
\item\label{lem:sim-props-1} If $(C,R_i)$, $i\in I$, are simulations on $(C,c)$, then the coproduct $\bigvee_{i\in I} (C,R_i)$ in $\Rel_C(\C)$ is a simulation.
\item\label{lem:sim-props-2} The identity relation $(C,C)$ is a simulation on $(C,c)$.
\item\label{lem:sim-props-3} If $(C,R)$ and $(C,S)$ are simulations on $(C,c)$, then the composite $(C,R)\bullet (C,S)$ is a simulation. 
\item\label{lem:sim-props-4} The similarity relation on $(C,c)$ exists, and it is reflexive and transitive.
\end{enumerate}
\end{lemma}

\begin{proof}
\begin{enumerate}[label=(\alph*)]
\item Let $(C,R)=\bigvee_i (C,R_i)$ and denote by $\mathsf{in}_i\c (C,R_i)\to (C,R)$ the coproduct injection. We first compute
\[ c_\star(\coprod_i (C,R_i)) = \coprod_i c_\star(C,R_i) \leq \coprod_i \barF(C,R_i) \leq \barF(C,R), \]
where the first $\coprod$ refers to the coproduct in $\Gra_C(\C)$ and the other two to the coproduct in $\Gra_{FC}(\C)$. The first step follows from the definition of $\coprod$ and $c_\star$, the second step uses that $(C,R_i)$ is a simulation, and the third step is witnessed by the morphism $[\barF(\mathsf{in}_i)]_i$. It follows that
\[ c_\star(C,R) = c_\star((\coprod_i (C,R_i))^\dag) \leq (c_\star(\coprod_i (C,R_i)))^\dag \leq \barF(C,R)^\dag = \barF(C,R). \]
The first step uses that $(C,R)=\bigvee_i (C,R_i)$ and the definition of $\bigvee$, the second one follows from \eqref{eq:f-star-respects-dag}, the third one from the computation above, and the last one since $\barF(C,R)$ is a relation.
\item This follows from the computation

\[ c_\star(C,C) \leq (FC,FC) \leq \barF(C,C) \]
where the first step is witnessed by the $\Gra_{FC}(\C)$-morphism $(\id,c)\c c_\star(C,C)\to (FC,FC)$ and the second one uses that $\barF(C,C)$ is reflexive by condition \ref{eq:good-for-simulations-1}.
\item This follows from the computation
\[ c_\star((C,R)\bullet (C,S)) \leq c_\star(C,R)\bullet c_\star(C,S) \leq\barF(C,R)\bullet \barF(C,S) \leq \barF((C,R)\bullet (C,S)). \]
The first step follows from \eqref{eq:f-star-respects-rel-comp}, the second step uses that $c_\star(C,R)\leq \barF(C,R)$ and $c_\star(C,S)\leq \barF(C,S)$ because $(C,R)$ and $(C,S)$ are simulations and that $\bullet$ is functorial, and the third step uses condition \ref{eq:good-for-simulations-2}.
\item Since the category $\C$ is well-powered (\Cref{asm:cat}), the collection $\{\,(C,R_i) : i\in I\,\}$ of all simulations, taken up to isomorphism in $\Rel_C(\C)$, forms a small set. Thus the greatest simulation is given by the coproduct $(C,R)=\bigvee_i (C,R_i)$, see part \ref{lem:sim-props-1}. It is reflexive because $(C,C)$ is a simulation by part \ref{lem:sim-props-2} and thus $(C,C)\leq (C,R)$ because $(C,R)$ is the greatest simulation. It is transitive because $(C,R)\bullet (C,R)$ is a simulation by part \ref{lem:sim-props-3} and thus $(C,R)\bullet (C,R)\leq (C,R)$.\qedhere
\qedhere
\end{enumerate}
\end{proof}

\begin{remark}\label{rem:congruence}
By \eqref{eq:f-star-respects-dag} the congruence property in \Cref{def:cong} is equivalent to
\[ a_\star\ol{\Sigma}_\Gra(A,R)\leq (A,R), \]
that is, existence of a morphism $a_R\c \Sigma R\to R$ such that the diagram below commutes.
\[
\begin{tikzcd}
\Sigma A \ar{d}[swap]{a} & \Sigma R \ar{l}[swap]{\Sigma\outl_R} \ar{d}{a_R} \ar{r}{\Sigma\outr_R} & \Sigma A \ar{d}{a} \\
A  & \ar{l}[swap]{\outl_{R}} R \ar{r}{\outr_{R}} & A
\end{tikzcd}
\]
\end{remark}

\section*{Proof of \Cref{lem:howe-props}}

\begin{enumerate}
\item Suppose that $(X,R)$ is reflexive. Then $(X,\hatR)$ is reflexive
  because
  \[
    (X,X)\leq (X,R) \leq (X,\hatR).
  \]
  The first step uses that $(X,R)$ is reflexive and the second step is
  witnessed by $\alpha_{R,\xi}\comp \inl$, see
  \eqref{eq:howe-alg-struct}. Similarly, $(X,\hatR)$ is a congruence
  on $(X,\xi)$ because
  \[ \xi_\star\ol{\Sigma}(X,\hatR) = (\xi_\star\ol{\Sigma}(X,\hatR))
    \bullet (X,X) \leq (\xi_\star\ol{\Sigma}(X,\hatR)) \bullet (X,R) \leq
    (X,\hatR).
  \]
  The first step follows from the definition of relation composition, the
  second step uses that $(X,R)$ is reflexive and that relation
  composition is functorial, and the third step is witnessed by
  $\alpha_{R,\xi}\comp \inr$.
\item Suppose that $(X,R)$ is transitive. Putting $\vee=\vee_X$, we obtain
\begin{align*} 
(X,\hatR)\bullet (X,R) & \cong ((X,R)\vee(\xi_\star \ol{\Sigma}(X,\hatR))\bullet (X,R))\bullet (X,R) \\
&\cong (X,R)\bullet (X,R) \vee (\xi_\star \ol{\Sigma}(X,\hatR))\bullet (X,R)\bullet (X,R) \\
& \leq (X,R) \vee (\xi_\star \ol{\Sigma}(X,\hatR))\bullet (X,R) \\
& \cong (X,\hatR).
\end{align*}
The first and the last step are witnessed by the isomorphism \eqref{eq:howe-alg-struct}, the second one uses \Cref{lem:comp-dist-over-coproducts}, and the third one uses that $(X,R)\bullet (X,R)\leq (X,R)$ by transitivity of $(X,R)$ and that $\vee$ and $\bullet$ are monotone.
\end{enumerate}

\section*{Proof of \Cref{prop:howe}}
The proof of the proposition requires some notation:
\begin{notation}\label{not:b-rho}
\begin{enumerate}
\item For each $(X,R),(Y,S)\in \Rel(\C)$ we denote the relation $\barB((X,R),(Y,S))$ by
\[
\begin{tikzcd}[row sep=2em]
E_{R,S} \ar[bend right=2em]{d}[swap]{\outl_{R,S}}\ar[bend left=2em]{d}{\outr_{R,S}} \\
B(X,Y)
\end{tikzcd}\]
Thus the domain and codomain of the $\Rel(\C)$-morphism
\[ \barrho_{(X,R),(Y,S)}\c \ol{\Sigma}((X,R)\times \barB((X,R),(Y,S))) \to \barB((X,R),\ol{\Sigma}^\star((X,R)+(Y,S))) \]
are the relations
\[ \ol{\Sigma}((X,R)\times \barB((X,R),(Y,S))) = (\Sigma(X\times B(X,Y)), (\Sigma(R\times E_{R,S}))^\dag) \]
and 
\[ \barB((X,R),\ol{\Sigma}^\star((X,R)+(Y,S))) = (B(X,\Sigmas(X+Y)), E_{R,(\Sigmas(R\vee S))^\dag}).  \]
We put
\[ t_{R,S} = (\, \Sigma(R\times E_{R,S}) \xto{e_{\Sigma(R\times E_{R,S})}} (\Sigma(R\times E_{R,S}))^\dag \xto{(\barrho_{(X,R),(Y,S)})_1} E_{R,(\Sigmas(R\vee S))^\dag}  \,). \]  
Then we have the following commutative diagram, where $\out\in \{\outl,\outr\}$; the lower right-hand part commutes because $\barrho_{(X,R),(Y,S)}$ is a $\Rel(\C)$-morphism and $(\barrho_{(X,R),(Y,S)})_0=\rho_{X,Y}$.
\begin{equation}\label{eq:t-diagram}
\begin{tikzcd}[column sep=4em]
\Sigma(R\times E_{R,S}) \ar{d}[swap]{\Sigma(\out_R\times \out_{R,S})} \ar{r}{e_{\Sigma(R\times E_{R,S})}} \ar[shiftarr = {yshift=2em}]{rr}{t_{R,S}} & (\Sigma(R\times E_{R,S}))^\dag \ar{r}{(\barrho_{(X,R),(Y,S)})_1} \ar{dl}[description]{\out_{(\Sigma(R\times E_{R,S}))^\dag}} & E_{R,(\Sigmas(R\vee S))^\dag} \ar{d}{\out_{R,(\Sigmas(R\vee S))^\dag}} \\
 \Sigma(X\times B(X,Y)) \ar{rr}{\rho_{X,Y}} & & B(X,\Sigmas(X+Y)
\end{tikzcd}
\end{equation} 
\item  Since $(\mS,\hatR)$ is a congruence on $(\mS,\ini)$ by \Cref{lem:howe-props}, there exists a $\C$-morphism $\ini_\hatR$ such that diagram below commutes, cf.\ \Cref{rem:congruence}.
\begin{equation}\label{eq:ini-hatR}
\begin{tikzcd}
\Sigma(\mS) \ar{d}[swap]{\ini} & \Sigma \hatR  \ar{l}[swap]{\Sigma\outl_\hatR} \ar{d}{\ini_\hatR} \ar{r}{\Sigma\outr_\hatR} & \Sigma(\mS) \ar{d}{\ini} \\
\mS  & \ar{l}[swap]{\outl_{\hatR}} \hatR \ar{r}{\outr_{\hatR}} & \mS
\end{tikzcd} 
\end{equation}
\item Since $(\mS,R)$ is a weak simulation on $(\mS,\gamma)$, there exist $\C$-morphisms $\wt{\gamma}_R$ and $\dbtilde{\gamma}_R$ such that the following diagrams commute, cf.\ \Cref{def:weakaning-and-weak-sim}.
\begin{equation}\label{eq:gamma-R}
\begin{tikzcd}[row sep=2.3em]
\mS \ar{d}[swap]{{\gamma}} & R \ar{l}[swap]{\outl_R} \ar{d}{\wt{\gamma}_R} \ar{r}{\outr_R} & \mS \ar{d}{\wt{\gamma}} \\
B(\mS,\mS)  & \ar{l}[swap]{\outl_{\mS,R}} E_{\mS,R} \ar{r}{\outr_{\mS,R}} & B(\mS,\mS)
\end{tikzcd} 
\qquad
\begin{tikzcd}[row sep=2.3em]
\mS \ar{d}[swap]{\wt{\gamma}} & R  \ar{l}[swap]{\outl_R} \ar{d}{\dbtilde{\gamma}_R} \ar{r}{\outr_R} & \mS \ar{d}{\wt{\gamma}} \\
B(\mS,\mS)  & \ar{l}[swap]{\outl_{\mS,R}} E_{\mS,R} \ar{r}{\outr_{\mS,R}} & B(\mS,\mS)
\end{tikzcd} 
\end{equation}
\end{enumerate}
\end{notation}

\begin{proof}[Proof of \Cref{prop:howe}]
  Suppose that $(\mS,R)$ is a reflexive and transitive weak simulation on
 $(\mS,\gamma)$. Our task is to show that the Howe closure $(\mS,\hatR)$ 
  is a weak simulation on $(\mS,\gamma)$. To this end, form the
  relation $(\mS,P)\in \Rel_{\mS}(\C)$ via the pullback below; it is a relation because monomorphisms are stable under pullbacks.
\begin{equation}\label{diag:P}
\begin{tikzcd}
P \pullbackangle{-45} \ar{r}{p} \ar{d}[swap]{\langle \outl_P, \outr_P \rangle} & E_{\hatR,\hatR} \ar{d}{\langle \outl_{\hatR,\hatR},\outr_{\hatR,\hatR}\rangle} \\
\mS\times \mS \ar{r}{\langle \gamma, \wt{\gamma}\rangle} & B(\mS,\mS)\times B(\mS,\mS)
\end{tikzcd}
\end{equation}
The key to the proof is showing the existence of $\Rel_{\mS}(\C)$-morphisms
  \begin{equation}\label{eq:beta-l-r}
    \beta^\lft \c (\mS,R)\to (\mS,P)
    \qqand
    \beta^\rgt\c \iota_\star\ol{\Sigma}((\mS,\hatR)\times_{\mS} (\mS,P))\bullet (\mS,R) \to (\mS,P)
  \end{equation}
where $\times_{\mS}$ denotes the product in $\Rel_\mS(\C)$, see \Cref{sec:graph-cocomplete}. Once this is achieved, we can conclude the proof as follows. By copairing $\beta^\lft$ and $\beta^\rgt$ we obtain the $\Rel_{\mS}(\C)$-morphism
  \[
    \beta = [\beta^\lft,\beta^\rgt]\c 
    \ol{\Sigma}_{R,\iota} \big((\mS,\hatR) \times_{\mS}
    (\mS,P)\big)
    \to (\mS,P)
  \]
and thus primitive recursion \eqref{eq:primitive-recursion} yields the $\Rel_\mS(\C)$-morphism
\[
\primrec\beta\c \mu\ol{\Sigma}_{R,\ini}=(\mS,\hatR)\to (\mS,P).
\]
Choose a witness $r\c (\mS,\mS)\to (\mS,\hatR)$ that $(\mS,\hatR)$ is reflexive, see \Cref{lem:howe-props}. Then the following commutative diagram proves $(\mS,\hatR)$ to be a weak simulation on $(\mS,\gamma)$:
\[
  \begin{tikzcd}[row sep=3em, column sep=3em]
& \hatR \ar[bend right=2em]{dl}[swap]{\outl_\hatR} \ar[bend left=2em]{dr}{\outr_\hatR} \ar{d}{(\primrec\beta)_1} & \\
\mS \ar{d}[swap]{\gamma} & P  \ar{l}[swap]{\outl_{P}} \ar{d}{p} \ar{r}{\outr_{P}} & \mS \ar{d}{\wt{\gamma}} \\
B(\mS,\mS)  & \ar{l}[swap]{\outl_{\hatR,\hatR}} E_{\hatR,\hatR} \ar{d}[description]{\barB(r,(\mS,\hatR))_1} \ar{r}{\outr_{\hatR,\hatR}} & B(\mS,\mS) \\
& \ar[bend left=2em]{ul}{\outl_{\mS,\hatR}} \ar[bend right=2em]{ur}[swap]{\outr_{\mS,\hatR}} E_{\mS,\hatR} & 
\end{tikzcd}
\]
 It only remains to define the $\Rel_\mS(\C)$-morphisms $\beta^\lft$ and $\beta^\rgt$ in \eqref{eq:beta-l-r}. Their construction 
  imitates the arguments for the induction base and
  induction step, resp., in the proof of
  \Cref{thm:cool}.

  \medskip\noindent\textbf{Construction of
    $\beta^\lft\c (\mS,R)\to (\mS,P)$:}

\begin{enumerate}
%
\item We define the morphism $f\c \mS\to E_{\hatR,\hatR}$ via primitive recursion as follows:
\begin{equation}\label{diag:g1}
\begin{tikzcd}[column sep=3.7em]
\Sigma(\mS) \ar{rrrr}{\iota} \ar{d}[swap]{\Sigma\langle \id,f\rangle} & & & & \mS \ar{d}{f} \\
\Sigma(\mS\times E_{\hatR,\hatR}) \ar{r}{\Sigma(r_1\times \id)} & \Sigma(\hatR \times E_{\hatR,\hatR}) \ar{r}{t_{\hatR,\hatR}} & E_{\hatR,(\Sigmas(\hatR\vee\hatR))^\dag} \ar{r}{\barBs(\id,\ol{\Sigma}^\star \nabla)_1}  & E_{\hatR,(\Sigmas \hatR)^\dag} \ar{r}{\barBs(\id,\widehat{(\iota,\iota_\hatR)^\dag})_1 } & E_{\hatR,\hatR} 
\end{tikzcd}
\end{equation}
Here $t_{\hatR,\hatR}$ is defined in \Cref{not:b-rho}, and $\widehat{(\iota,\iota_\hatR)^\dag}\c
\ol{\Sigma}^\star(\mS,\hatR)\to (\mS,\hatR)$ is the
$\ol{\Sigma}^\star$-algebra corresponding to the
$\ol{\Sigma}$-algebra $(\iota,\iota_\hatR)^\dag\c
\ol{\Sigma}(\mS,\hatR)\to (\mS,\hatR)$, cf.\ \Cref{prop:free-monad-lift} and its proof.
\item The morphism $f\c \mS\to E_{\hatR,\hatR}$ satisfies
\begin{equation}\label{eq:gamma-vs-g1} \gamma = \outl_{\hatR,\hatR}\comp f \qqand \gamma = \outr_{\hatR,\hatR}\comp f. \end{equation}
By symmetry it suffices to prove the first equation, which follows from the observation that the morphism $\outl_{\hatR,\hatR}\comp f$ satisfies the commutative diagram \eqref{diag:gamma} defining $\gamma$:
\[  
\begin{tikzcd}
\Sigma(\mS) \ar{rrrr}{\ini} \ar{d}[swap, description]{\Sigma\langle\id,f\rangle} & & & & \mS \ar{d}{f}\\
\Sigma(\mS\times E_{\hatR,\hatR}) \ar{r}{\Sigma(r_1\times \id)} \ar{d}[swap, description]{\Sigma(\id\times \outl_{\hatR,\hatR})} & \Sigma(\hatR\times E_{\hatR,\hatR}) \ar[phantom]{dr}[description]{\eqref{eq:t-diagram}} \ar{dl}[near start]{\Sigma(\outl_{\hatR}\times \outl_{\hatR,\hatR})} \ar{r}{t_{\hatR,\hatR}} & E_{\hatR,(\Sigmas(\hatR\vee\hatR))^\dag} \ar{d}{\outl_{\hatR,(\Sigmas(\hatR\vee\hatR))^\dag}} \ar{r}{\barBs(\id,\ol{\Sigma}^\star\nabla)_1} & E_{\hatR,(\Sigmas\hatR)^\dag} \ar{d}{\outl_{\hatR,(\Sigmas\hatR)^\dag}} \ar{r}{\barBs(\id,\widehat{(\ini,\ini_\hatR)^\dag})_1} & E_{\hatR,\hatR} \ar{d}{\outl_{\hatR,\hatR}} \\
\Sigma(\mS\times B(\mS,\mS)) \ar{rr}{\rho_{\mS,\mS}} & & B(\mS,\Sigmas(\mS+\mS)) \ar{r}{B(\mS,\Sigmas\nabla)} & B(\mS,\Sigmas(\mS)) \ar{r}{B(\mS,\hat\ini)} & B(\mS,\mS) \\
\end{tikzcd}
\] 
\item Now consider the composite graph
  $\barB((\mS,\hatR),(\mS,\hatR))\smc
  \barB((\mS,\mS),(\mS,R))$.
  The universal property of the corresponding pullback
  $E_{\hatR,\hatR}\smc E_{\mS,R}$ yields a unique morphism $g$
  making the diagram below commute, where $\pi$ and $\pi'$ are the
  left and right projection of the pullback:
  \begin{equation}\label{diag:g2}
    \begin{tikzcd}
      \mS \ar[swap]{d}{f}
      &
      R
      \ar{l}[swap]{\outl_R}
      \ar[bend left=2em]{dr}{\wt{\gamma}_R}
      \ar{d}{g}
      &
      \\
      E_{\hatR,\hatR}
      &
      \ar{l}[swap]{\pi} \ar{r}{\pi'} E_{\hatR,\hatR}\smc E_{\mS,R}
      &
      E_{\mS,R}
    \end{tikzcd}
  \end{equation}
  Indeed, by definition $E_{\hatR,\hatR}\smc E_{\mS,R}$ is the pullback
  of $\outr_{\hatR,\hatR}$ and $\outl_{\mS,R}$, and the commutative
  diagram below shows that these two morphisms merge
  $f\comp \outl_R$ and $\wt{\gamma}_R$.
  \[
    \begin{tikzcd}[row sep=3em]
      R
      \ar[phantom]{dr}{\eqref{eq:gamma-R}}
      \ar{d}[swap]{\wt{\gamma}_R}
      \ar{r}{\outl_R}
      &
      \mS
      \ar{drr}{\gamma}
      \ar{rr}{f}
      & { }&
      E_{\hatR,\hatR}
      \ar{d}{\outr_{\hatR,\hatR}}
      &
      \\
      E_{\mS,R} 
	  \ar[phantom]{urrr}[description, pos=.75, xshift=2em]{\eqref{eq:gamma-vs-g1}}
      \ar{rrr}[swap]{\outl_{\mS,R}}
      & { } &{ } &
      B(\mS,\mS)
    \end{tikzcd}
  \]
\item\label{prop:howe:4} Next, we observe that there exists a $\Gra_{B(\mS,\mS)}(\C)$-morphism
\begin{equation}\label{eq:morphism-k}
 k\c \barB\big((\mS,\hatR),(\mS,\hatR)\big)\smc \barB\big((\mS,\mS),(\mS,R)\big) \to  \barB\big((\mS,\hatR),\, (\mS,\hatR) \big).
\end{equation}
This is shown by the following computation:
  \begin{align*}
    ~& \barB\big((\mS,\hatR),(\mS,\hatR)\big)\smc \barB\big((\mS,\mS),(\mS,R)\big) \\
    \leq~& \barB\big((\mS,\hatR),(\mS,\hatR)\big)\bullet \barB\big((\mS,\mS),(\mS,R)\big)
    \\
    \leq~&
    \barB\big((\mS,\hatR),\, (\mS,\hatR)\bullet (\mS,R)\big)
    \\
    \leq~&
    \barB\big((\mS,\hatR),\, (\mS,\hatR) \big).
  \end{align*}
The first step holds by definition of relation composition $\bullet$; the second step follows from \ref{eq:good-for-simulations-bifunctor-3}; the third step uses that $(\mS,\hatR)$ is
  weakly transitive by \Cref{lem:howe-props}.
\item\label{prop:howe:5} Finally, we have the following commutative diagram:
\[  
\begin{tikzcd}[row sep=3em]
\mS \ar{d}[swap]{\gamma} \ar{dr}{f} & & R \ar{ll}[swap]{\outl_R} \ar{rr}{\outr_R} \ar{d}{g}  \ar{dr}{\wt{\gamma}_R} & & \mS \ar{d}{\wt{\gamma}} \\
B(\mS,\mS) \ar[phantom]{urr}[pos=.15, xshift=-1.5em]{\eqref{eq:gamma-vs-g1}} & E_{\hatR,\hatR} \ar[phantom]{dr}{\eqref{eq:morphism-k}} \ar[phantom]{ur}{\eqref{diag:g2}} \ar{l}[swap]{\outl_{\hatR,\hatR}} & E_{\hatR,\hatR}\smc E_{\mS,R} \ar[phantom]{urr}[pos=.15, xshift=-1.5em]{\eqref{eq:gamma-vs-g1}}  \ar{l}[swap]{\pi} \ar{d}{k_1} \ar{r}{\pi'} & E_{\mS,R} \ar[phantom]{dl}{\eqref{eq:morphism-k}} \ar[phantom]{ur}{\eqref{eq:gamma-R}} \ar{r}{\outr_{\mS,R}} & B(\mS,\mS) \\
& & E_{\hatR,\hatR} \ar[bend left=1em]{ull}{\outl_{\hatR,\hatR}} \ar[bend right=1em]{urr}[swap]{\outr_{\hatR,\hatR}} & & 
\end{tikzcd}
\]
The universal property of the pullback \eqref{diag:P} yields a unique $\beta^\lft_1\c R\to P$ such that $\outl_R=\outl_P\comp \beta^\lft_1$ and $\outr_R=\outr_P\comp \beta^\lft_1$ and $k_1\comp g = p\comp \beta^\lft_1$. By the first two equations, we thus obtain the $\Rel_{\mS}(\C)$-morphism
\[ \beta^\lft=(\id,\beta^\lft_1)\c (\mS,R)\to (\mS,P). \]
\end{enumerate}

\medskip\noindent\textbf{Construction of $\beta^\rgt\c \iota_\star\ol{\Sigma}((\mS,\hatR)\times_{\mS} (\mS,P))\bullet (\mS,R) \to (\mS,P)$:}
\begin{enumerate}
\item Recall that the product $(\mS,\hatR)\times_{\mS} (\mS,P)$ is the relation $(\mS,Q)$ where $Q$ is the pullback
\begin{equation}\label{diag:Q}
  \begin{tikzcd}[column sep=5em]
    Q
    \pullbackangle{-45}
    \ar{d}[swap]{q_{\hatR}} \ar{r}{q_{P}}
    &
    P \ar{d}{\langle \outl_{P},\, \outr_{P}\rangle}
    \\  
    \hatR \ar{r}{\langle \outl_\hatR,\, \outr_\hatR \rangle}
    & \mS \times \mS 
  \end{tikzcd}
\end{equation}
and
\begin{equation}\label{eq:outq}
  \outl_Q=\outl_\hatR\comp q_\hatR,\qqand \outr_Q=\outr_\hatR\comp
  q_\hatR.
\end{equation}

\item\label{prop:howe:2:2}
Consider the $\C$-morphism $h\c \Sigma Q\to E_{\hatR,\hatR}$ defined as the composite
\[ 
\begin{tikzcd}[column sep=2.1em]
h\;=\;(\,\Sigma Q  \ar{r}[yshift=.3em]{\Sigma\langle q_\hatR,q_{P}\rangle} & \Sigma(\hatR\times {P})  \ar{r}[yshift=.3em]{\Sigma(\id\times p)} & \Sigma(\hatR\times E_{\hatR,\hatR}) \ar{r}{t_{\hatR,\hatR}} & E_{\hatR,(\Sigmas(\hatR\vee\hatR))^\dag} \ar{r}[yshift=.3em]{\barBs(\id,\ol{\Sigma}^\star\nabla)_1} & E_{\hatR,(\Sigmas\hatR)^\dag}  \ar{r}[yshift=.3em]{\barBs(\id,\widehat{(\ini,\ini_\hatR)^\dag})_1} & E_{\hatR,\hatR}\,).
\end{tikzcd}
\]
We claim that the diagram below commutes laxly, where $\outl_{\Sigma Q}=\ini\comp \Sigma \outl_Q$ and $\outr_{\Sigma Q}=\ini\comp \Sigma \outr_Q$ are the projections of the relation $\ini_\star\ol{\Sigma}(\mS,Q)$:
\begin{equation}\label{eq:h1-diag}
  \begin{tikzcd}
\mS \ar[phantom]{dr}[description, pos=.4]{\deq{-45}} \ar{d}[swap]{\gamma} & \Sigma Q \ar[phantom]{dr}[description, pos=.4]{\dleq{45}} \ar{l}[swap]{\outl_{\Sigma Q}} \ar{d}{h} \ar{r}{\outr_{\Sigma Q}} & \mS \ar{d}{\wt{\gamma}} \\
B(\mS,\mS)  & \ar{l}[swap]{\outl_{\hatR,\hatR}} E_{\hatR,\hatR} \ar{r}{\outr_{\hatR,\hatR}} & B(\mS,\mS)
\end{tikzcd}
\end{equation}
This is proven by the two diagrams below. In the part $(\ast)$ of the second diagram, we make use of our assumption that $(\mS,\ini,\wt{\gamma})$ is a lax $\rho$-bialgebra.
\[
\begin{tikzcd}[column sep=1.5em, row sep=3em]
\Sigma Q  \ar[shiftarr = {yshift=20}]{rrrrr}{h} \ar{r}{\Sigma\langle q_\hatR,q_{P}\rangle} \ar{dd}[swap]{\Sigma\outl_Q} \ar{dr}[description, pos=.7, xshift=-1em]{\Sigma\langle \outl_Q,\outl_Q\rangle} & \Sigma(\hatR\times {P}) \ar[phantom]{dr}[description]{\eqref{diag:P}}  \ar{r}{\Sigma(\id\times p)} \ar{d}[description, pos=.3]{\Sigma(\outl_\hatR\times \outl_{P})} & \Sigma(\hatR\times E_{\hatR,\hatR}) \ar[phantom]{dr}[description]{\eqref{eq:t-diagram}} \ar{r}{t_{\hatR,\hatR}} \ar{d}[description]{\Sigma(\outl_{\hatR}\times \outl_{\hatR,\hatR})} & E_{\hatR,(\Sigmas(\hatR\vee\hatR))^\dag} \ar{r}{\barBs(\id,\ol{\Sigma}^\star\nabla)_1} \ar{d}[description]{\outl_{\hatR,(\Sigmas(\hatR\vee\hatR))^\dag}} & E_{\hatR,(\Sigmas\hatR)^\dag}  \ar{r}{\barBs(\id,\widehat{(\ini,\ini_\hatR)^\dag})_1} \ar{d}[description]{\outl_{\hatR,(\Sigmas\hatR)^\dag}} & E_{\hatR,\hatR} \ar{dd}{\outl_{\hatR,\hatR}} \\
& \Sigma(\mS\times \mS) \ar{r}{\Sigma(\id\times \gamma)} & \Sigma(\mS\times B(\mS,\mS)) \ar{r}{\rho_{\mS,\mS}} & B(\mS,\Sigmas(\mS+\mS))  \ar{r}[yshift=.5em]{B(\id,\Sigmas\nabla)} & B(\mS,\Sigmas(\mS)) \ar{dr}{B(\id,\hat\ini)} & \\
\Sigma(\mS) \ar{ur}{\Sigma\langle\id,\id\rangle} \ar{urr}[swap]{\Sigma\langle \id,\gamma\rangle} \ar{rr}{\ini} & & \mS \ar[phantom]{ur}[description]{\text{\eqref{diag:gamma}}} \ar{rrr}{\gamma} & & & B(\mS,\mS) 
\end{tikzcd}
\]
\[
\begin{tikzcd}[scale cd=.98, column sep=1.5em, row sep=3em]
\Sigma Q  \ar[shiftarr = {yshift=20}]{rrrrr}{h} \ar{r}{\Sigma\langle q_\hatR,q_{P}\rangle} \ar{dd}[swap]{\Sigma\outr_Q} \ar{dr}[description, pos=.7, xshift=-1em]{\Sigma\langle \outr_Q,\outr_Q\rangle} & \Sigma(\hatR\times {P}) \ar[phantom]{dr}[description]{\eqref{diag:P}} \ar{r}{\Sigma(\id\times p)} \ar{d}[description, pos=.3]{\Sigma(\outr_\hatR\times \outr_{P})} & \Sigma(\hatR\times E_{\hatR,\hatR}) \ar[phantom]{dr}[description]{\eqref{eq:t-diagram}} \ar{r}{t_{\hatR,\hatR}} \ar{d}[description]{\Sigma(\outr_{\hatR}\times \outr_{\hatR,\hatR})} & E_{\hatR,(\Sigmas(\hatR\vee\hatR))^\dag} \ar{r}{\barBs(\id,\ol{\Sigma}^\star\nabla)_1} \ar{d}[description]{\outr_{\hatR,(\Sigmas(\hatR\vee\hatR))^\dag}} & E_{\hatR,(\Sigmas\hatR)^\dag}  \ar{r}{\barBs(\id,\widehat{(\ini,\ini_\hatR)^\dag})_1} \ar{d}[description]{\outr_{\hatR,(\Sigmas\hatR)^\dag}} & E_{\hatR,\hatR} \ar{dd}{\outr_{\hatR,\hatR}} \\
& \Sigma(\mS\times \mS) \ar{r}{\Sigma(\id\times \wt{\gamma})} & \Sigma(\mS\times B(\mS,\mS)) \ar{r}{\rho_{\mS,\mS}} & B(\mS,\Sigmas(\mS+\mS))  \ar{r}[yshift=.5em]{B(\id,\Sigmas\nabla)} & B(\mS,\Sigmas(\mS)) \ar{dr}{B(\id,\hat\ini)} & \\
\Sigma(\mS) \ar{ur}{\Sigma\langle\id,\id\rangle} \ar{urr}[swap]{\Sigma\langle \id,\wt{\gamma}\rangle} \ar{rr}{\ini} & & \mS \ar[phantom]{ur}[description]{\dgeq{45}~(\ast)} \ar{rrr}{\wt{\gamma}} & & & B(\mS,\mS) 
\end{tikzcd}
\]
From \eqref{eq:h1-diag} and \ref{eq:good-for-simulations-bifunctor-1}, it follows that there exists a $\C$-morphism $l\c \Sigma Q\to E_{\hatR,\hatR}$ such that the following diagram commutes:
\begin{equation}\label{eq:h2-diag}
  \begin{tikzcd}
\mS \ar{d}[swap]{\gamma} & \Sigma Q \ar{l}[swap]{\outl_{\Sigma Q}} \ar{d}{l} \ar{r}{\outr_{\Sigma Q}} & \mS \ar{d}{\wt{\gamma}} \\
B(\mS,\mS)  & \ar{l}[swap]{\outl_{\hatR,\hatR}} E_{\hatR,\hatR} \ar{r}{\outr_{\hatR,\hatR}} & B(\mS,\mS)
\end{tikzcd}
\end{equation}

\item Again consider the composite graph
  $\barB((\mS,\hatR),(\mS,\hatR))\smc
  \barB((\mS,\mS),(\mS,R))$.
  The universal property of the corresponding pullback
  $E_{\hatR,\hatR}\smc E_{\mS,R}$ yields a unique morphism $m$
  making the diagram below commute, where the upper row refers to the composite graph $\ini_\star\ol{\Sigma}(\mS,Q)\smc (\mS,R)$.
\begin{equation}\label{eq:l-diag}
\begin{tikzcd}
\Sigma Q \ar{d}[swap]{l} & \Sigma Q\smc R \ar{l}[swap]{\pi_{\Sigma Q\smc R,\Sigma Q}} \ar{d}{m} \ar{r}{\pi_{\Sigma Q\smc R, R}} & R \ar{d}{\dbtilde{\gamma}_R} \\
E_{\hatR,\hatR} & E_{\hatR,\hatR} \smc E_{\mS,R} \ar{l}[swap]{\pi} \ar{r}{\pi'} & E_{\mS,R} 
\end{tikzcd}
\end{equation}
  Indeed, by definition $E_{\hatR,\hatR}\smc E_{\mS,R}$ is the pullback
  of $\outr_{\hatR,\hatR}$ and $\outl_{\mS,R}$, and the
  diagram below shows that these two morphisms merge
  $l\comp \pi_{\Sigma Q\smc R, \Sigma Q}$ and $\dbtilde{\gamma}_R\comp \pi_{\Sigma Q\smc R,R}$.
\[
\begin{tikzcd}
\Sigma Q \smc R \ar{r}{\pi_{\Sigma Q\smc R,\Sigma Q}} \ar{dd}[swap]{\pi_{\Sigma Q\smc R, R}} & \Sigma Q \ar{r}{l} \ar{d}{\outr_{\Sigma Q}} \ar[phantom]{drd}[description]{\eqref{eq:h2-diag}} & E_{\hatR,\hatR} \ar{dd}{\outr_{\hatR,\hatR}} \\
& \mS \ar[phantom]{d}[description]{\eqref{eq:gamma-R}} \ar{dr}{\wt{\gamma}} & \\
R \ar{ur}{\outl_R} \ar{r}{\dbtilde{\gamma}_R}  & E_{\mS,R} \ar{r}{\outl_{\mS,R}} & B(\mS,\mS) 
\end{tikzcd}
\]
\item Finally, we have the commutative diagram below:
\[
\begin{tikzcd}
\mS \ar[phantom]{dr}[description]{\eqref{eq:h1-diag}} \ar{d}[swap]{\gamma} & \Sigma Q \ar[phantom]{dr}[description]{\text{\eqref{eq:l-diag}}} \ar{l}[swap]{\outl_{\Sigma Q}} \ar{d}{l} & \Sigma Q\smc R  \ar[phantom]{dr}[description]{\text{\eqref{eq:l-diag}}} \ar{l}[swap]{\pi_{\Sigma Q\smc R,\Sigma Q}} \ar{r}{\pi_{\Sigma Q\smc R, R}} \ar{d}{m} & R \ar[phantom]{dr}[description]{{\eqref{eq:gamma-R}}} \ar{r}{\outr_R} \ar{d}{\dbtilde{\gamma}_R} & \mS \ar{d}{\wt{\gamma}} \\
B(\mS,\mS) & E_{\hatR,\hatR} \ar[phantom]{dr}[description]{\eqref{eq:morphism-k}} \ar{l}[swap]{\outl_{\hatR,\hatR}} & E_{\hatR,\hatR}\smc E_{\mS,R} \ar{d}{k_1} \ar{l}[swap]{\pi} \ar{r}{\pi'} & E_{\mS,R} \ar{r}{\outr_{\mS,\hatR}} & B(\mS,\mS) \\
& & E_{\hatR,\hatR} \ar[phantom]{ur}[description]{\eqref{eq:morphism-k}} \ar[bend left=1em]{ull}{\outl_{\hatR,\hatR}} \ar[bend right=1em]{urr}[swap]{\outr_{\hatR,\hatR}} & & 
\end{tikzcd}
\]
The universal property of the pullback \eqref{diag:P} yields a unique $\gamma^\rgt_1\c \Sigma Q\smc R\to P$ such that $\outl_{\Sigma Q}\comp \pi_{\Sigma Q\smc R,\Sigma Q}=\outl_P\comp \gamma^\rgt_1$ and $\outr_R\comp \pi_{\Sigma Q\smc R,R}=\outr_P\comp \gamma^\rgt_1$ and $k_1\comp m = p\comp \gamma^\rgt_1$. Thus
\[ \gamma^\rgt=(\id,\gamma^\rgt_1)\c \iota_\star\ol{\Sigma}((\mS,\hatR)\times_{\mS} (\mS,P))\smc (\mS,R) \to (\mS,P) \]
is a $\Gra_{\mS}(\C)$-morphism, which yields the $\Rel_\mS(\C)$-morphism
\[ \beta^\rgt =(\gamma^\rgt)^\dag\c \iota_\star\ol{\Sigma}((\mS,\hatR)\times_{\mS} (\mS,P))\bullet (\mS,R) \to (\mS,P).  \]
This concludes the proof.\qedhere
\end{enumerate}
\end{proof}

\section*{Proof details for \Cref{S:HO-specs}}
We will prove that the given data satisfy the \Cref{asm-sim}:
\begin{enumerate}
\item\label{lem:assumptions-satisfied-ho-1} The functor $\Sigma$ preserves strong epimorphisms.
\item\label{lem:assumptions-satisfied-ho-2} The functor $\barB$ is good for simulations.
\item\label{lem:assumptions-satisfied-ho-3} The higher-order GSOS law $\barrho$ admits a relation lifting.
\end{enumerate}

\begin{proof}
\begin{enumerate}
\item Every set functor preserves epimorphisms (= surjections) since the latter are precisely the right-invertible maps.
\item By definition, we have $\barB=\ol{P}\comp \barB_0$ where $\barPow$ and $\barB_0$ are the relation liftings of \Cref{rem:weak-vs-strong}. Condition  \ref{eq:good-for-simulations-bifunctor-1} holds by \Cref{ex:good-for-simulations}. For
 \ref{eq:good-for-simulations-bifunctor-2} we first observe that for each $X\in \Set$ the relation $\barPow(X,X)$ is given by $(\Pow X,\seq)$, hence it is reflexive. The following computation then proves $\barB((X,X),(Y,Y))$ to be reflexive:
\begin{align*}
(B(X,Y)),B(X,Y)) &= (\Pow B_0(X,Y),\Pow B_0(X,Y)) \\
&\leq \barPow(B_0(X,Y), B_0(X,Y)) \\ 
&\cong \barPow\,\barB_0((X,X),(Y,Y)) \\
&= \barB((X,X),(Y,Y)).
\end{align*}
In the second step we use the above observation about $\barPow$. The third step follows from \Cref{prop:olB-resp-identities}, using the fact that $\barB_0$ is the canonical lifting of $B_0$. For \ref{eq:good-for-simulations-bifunctor-3} note that
\begin{equation}\label{eq:pow-pres-comp} 
\barPow(X,R)\bullet \barPow(X,S) \leq \barPow((X,R)\bullet (X,S)) \qquad \text{for every $(X,R),(X,S)\in \Rel_X(\C)$},
\end{equation}
which is immediate from the definition of the Egli-Milner relation. Thus
\begin{align*}
\barB((X,R),(Y,S))\bullet \barB((X,X),(Y,S'))  &= \barPow(\barB_0((X,R),(Y,S))) \bullet \barPow(\barB_0((X,X),(Y,S'))) \\ 
&\leq \barPow(\barB_0((X,R),(Y,S))\bullet \barB_0((X,X),(Y,S'))) \\
&\leq \barPow(\barB_0((X,R),(Y,S)\bullet(Y,S'))) \\
&= \barB((X,R),(Y,S)\bullet(Y,S')).
\end{align*}
The second step uses \eqref{eq:pow-pres-comp}. The third step follows from \Cref{eq:olB-resp-rel-comp} and the fact that $\barB_0$ is the canonical relation lifting of $B_0$, see \Cref{sec:bifunctor-lift}. 
\item By definition of $\rho$ as the composite \eqref{eq:rho0-to-rho}, it suffices to show that for all relations $(X,R)$ and $(Y,S)$ the maps
\begin{align*}
\st_{X,Y} & \c X\times \Pow Y \to \Pow(X\times Y), \\
\delta_{X} & \c \Sigma\Pow X  \to \Pow\Sigma X, \\
\rho_{X,Y}^0 &\c \Sigma(X\times B_0(X,Y)) \to B_0(X,\Sigmas(X+Y)), 
\end{align*}
are relation-preserving with respect to the relations induced by the liftings $\ol{\Sigma}$, $\barPow$, and $\barB_0$. For the first two maps this is clear by definition of the Egli-Milner relation, and for the third one it follows from \Cref{cons:rho-lift-rel}, using again that $\barB_0$ is the canonical lifting of $B_0$ \qedhere
\end{enumerate}
\end{proof}

\section{Canonical Liftings of Bifunctors}
\label{sec:bifunctor-lift}
In this section we demonstrate how to construct canonical graph and relation liftings of mixed-variance bifunctors on $\C$.
\begin{definition}Let $B\c \C^\opp\times \C \to\C$ be a bifunctor.
\begin{enumerate}
\item A \emph{graph lifting} of $B$ is a bifunctor $\barB$ on $\Gra(\C)$ making the first diagram commute.
\item  A \emph{relation lifting} of $B$ is a bifunctor $\barB$ on $\Rel(\C)$ making the second diagram commute.
\end{enumerate}
\[
\begin{tikzcd}
 \Gra(\C)^\opp\times \Gra(\C)  
 \ar{d}[swap]{\under{-}^\opp\times \under{-}} \ar{r}{\barBs} & \Gra(\C) \ar{d}{\under{-}}  \\
\C^\opp \times \C \ar{r}{B}  & \C 
\end{tikzcd}
\quad\quad 
\begin{tikzcd}
 \Rel(\C)^\opp\times \Rel(\C)  
 \ar{d}[swap]{\under{-}^\opp\times \under{-}} \ar{r}{\barBs} & \Rel(\C) \ar{d}{\under{-}}  \\
\C^\opp \times \C \ar{r}{B}  & \C 
\end{tikzcd}
\]
\end{definition}
 Recall the relation lifting of $B_0(X,Y)=Y+Y^X$ in \Cref{rem:weak-vs-strong}\ref{rem:weak-vs-strong-1}. Its construction can be generalized to the present categorical setting:
\begin{construction}\label{cons:lifting-bifunctor}
For every bifunctor $B\c \C^\opp\times \C\to \C$ we construct a \emph{canonical graph lifting}
\[\barB=\barB_\Gra\c \Gra(\C)^\opp\times \Gra(\C) \to \Gra(\C). \]
\begin{enumerate}
\item Given $(X,R),(Y,S)\in \Gra(\C)$, we define $\barB((X,R),(Y,S))\in \Gra(\C)$ to be the graph
\begin{equation}\label{eq:olB}
\begin{tikzcd}[row sep=2em]
T_{R,S} \ar[bend right=2em]{d}[swap]{\outl_{R,S}}\ar[bend left=2em]{d}{\outr_{R,S}} \\
B(X,Y)
\end{tikzcd}
\end{equation}
where $T_{R,S}$ is obtained by the pullback below and $\outl_{R,S}$ and $\outr_{R,S}$ are given by
\begin{equation}\label{eq:out}
  \outl_{R,S}=\cev{p}_{R,S}\circ p_{R,S}
  \qqand
  \outr_{R,S}=\vec{q}_{R,S}\circ q_{R,S}.
\end{equation}
\begin{equation}\label{eq:bartrs}
\begin{tikzcd}[column sep=5em, row sep=3em]
  &
  {T}_{R,S}
  \pullbackangle{-90}
  \ar[bend right=4em]{ddl}[swap]{{p}_{R,S}}
  \ar[bend left=4em]{ddr}{{q}_{R,S}}
  &
  \\
  &
  B(R,S)  
  \ar[shift right=.5em]{d}[swap]{B(\id,\outl_S)} \ar[shift
  left=.5em]{d}{B(\id,\outr_S)}
  & \\
  \cev{T}_{R,S}
  \pullbackangle{0}
  \ar[bend right=2em]{dr}[swap]{\cev{p}_{R,S}}  \ar[bend left=2em]{ur}{\cev{q}_{R,S}}
  &
  B(R,Y)
  &  \vec{T}_{R,S}
  \pullbackangle{180}
  \ar[bend right=2em]{ul}[swap]{\vec{p}_{R,S}}
  \ar[bend left=2em]{dl}{\vec{q}_{R,S}}
  \\
& B(X,Y) \ar[shift left=.5em]{u}{B(\outl_R,\id)}  \ar[shift right=.5em]{u}[swap]{B(\outr_R,\id)} & 
\end{tikzcd}
\end{equation}
In addition, we put
\begin{equation}\label{eq:b}
  b_{R,S} \;:=\;\cev{q}_{R,S}\circ p_{R,S} = \vec{p}_{R,S}\circ
  q_{R,S}.
\end{equation}

\item\label{cons:lifting-bifunctor:2} Given $\Gra(\C)$-morphisms $h\c (X',R')\to (X,R)$ and $k\c (Y,S)\to (Y',S')$, we define
  \[
    \barB(h,k)\c \barB((X,R),(Y,S))\to \barB((X',R'),(Y',S')),
  \]
  which is%
  \smnote{@Henning: to say it in Larry's words: you're using `i.e.'
    way to often! We should all try and cut down on it. (I believe
    Pitts' book does not have a single `i.e.' and uses `that is' sparingly.)}
  a pair of morphisms $\barB(h,k)_0$ and $\barB(h,k)_1$ making the following diagram commute:
\begin{equation}\label{diag:Bbar10}
  \begin{tikzcd}[row sep=2em, column sep = 7em] 
    T_{R,S}
    \ar[bend right=2em]{d}[swap]{\outl_{R,S}}
    \ar[bend left=2em]{d}{\outr_{R,S}}
    \ar{r}{\barBs(h,k)_1}
    &
    T_{R',S'}
    \ar[bend right=2em]{d}[swap]{\outl_{R',S'}}
    \ar[bend left=2em]{d}{\outr_{R',S'}}
    \\
    B(X,Y)
    \ar{r}{\barBs(h,k)_0}
    &
    B(X',Y')
\end{tikzcd}
\end{equation}
We put
\[
  \barB(h,k)_0
  = \big(B(X,Y)\xra{B(h_0,k_0)} B(X',Y') \big).
\]
For the definition of $\barB(h,k)_1$, we first use the universal
property of the pullback $\cev{T}_{R',S'}$ to obtain a unique morphism
$\cev{d}_{h,k}$ making the outside and the upper rectangle of the
diagram below commute. Note that the central and the lower rectangle
commute because $h$ and $k$ are $\Gra(\C)$-morphisms and the left and
right parts commute by \eqref{eq:bartrs}
\begin{equation}\label{diag:cevd}
\begin{tikzcd}[column sep=4em]
  &
  {T}_{R,S}
  \ar[shiftarr = {xshift=-50}]{ddd}[swap]{\outl_{R,S}}
  \ar[dashed]{r}{\cev{d}_{h,k}}
  \ar{d}[swap]{b_{R,S}}
  &
  \cev{T}_{R',S'}
  \ar{d}{\cev{q}_{R',S'}}
  \ar[shiftarr = {xshift=50}]{ddd}{\cev{p}_{R',S'}}
  \\
  { }
  &
  B(R,S)
  \descto[near start]{l}{\eqref{eq:bartrs}}
  \descto{rd}{\eqref{diag:hom}}
  \ar{r}{B(h_1,k_1)}
  \ar{d}[swap]{B(\id,\outl_S)}
  &
  B(R',S')
  \descto[near start]{r}{\eqref{eq:bartrs}}
  \ar{d}{B(\id,\outl_{S'})}
  &
  { }
  \\
  &
  B(R,Y) \ar{r}{B(h_1,k_0)}
  \descto{rd}{\eqref{diag:hom}}
  &
  B(R',Y')
  \\
  &
  B(X,Y) \ar{u}{B(\outl_R,\id)} \ar{r}{B(h_0,k_0)}
  &
  B(X',Y') \ar{u}[swap]{B(\outl_{R'},\id)} 
\end{tikzcd} 
\end{equation}
Analogously, we obtain a unique morphism $\vec{d}_{h,k}\c {T}_{R,S}\to \vec{T}_{R',S'}$ such that
\begin{equation}\label{eq:vecd}
  \vec{q}_{R',S'}\circ \vec{d}_{h,k}
  =
  B(h_0,k_0)\circ \outr_{R,S}
  \qqand
  \vec{p}_{R',S'}\circ \vec{d}_{h,k}
  =
  B(h_1,k_1)\circ b_{R,S}.
\end{equation}
In particular, we have
\[ \cev{q}_{R',S'}\circ \cev{d}_{h,k} = B(h_1,k_1)\circ b_{R,S} = \vec{p}_{R',S'} \circ \vec{d}_{h,k}, \]
so the universal property of the pullback $T_{R',S'}$ yields a unique morphism $\barB(h,k)_1$ making both triangles in the diagram below commute:
\begin{equation}\label{diag:Bbar}
  \begin{tikzcd}
    & {T}_{R,S} \ar[bend right=2em]{dl}[swap]{\cev{d}_{h,k}} \ar[bend left=2em]{dr}{\vec{d}_{h,k}} \ar[dashed]{d}{\barBs(h,k)_1}  & \\
    \cev{T}_{R',S'} &  \ar{l}[swap]{p_{R',S'}} \ar{r}{{q}_{R',S'}} T_{R',S'} & \vec{T}_{R',S'} 
  \end{tikzcd}
\end{equation}
\end{enumerate}
\end{construction}
\begin{lemma}\label{lem:dbar-welldef}
The assignment $\barB$ is a functor, and it forms a graph lifting of $B$.
\end{lemma} 
\begin{rem}
In the proof and later on, we will repeatedly make use of the
following diagram which commutes for every $h\c (X,R')\to (X,R)$ and $k\c (Y,S)\to (Y,S')$:
\begin{equation}\label{eq:brs-diag}
\begin{tikzcd}
T_{R,S} \ar{dr}{\cev{d}_{h,k}} \ar{rr}{\ol B(h,k)_1}
\ar{dd}[swap]{b_{R,S}}
&
\descto{d}{\eqref{diag:Bbar}}
&
T_{R',S'}
\ar{dd}{b_{R',S'}}
\ar{dl}[swap, inner sep =0]{p_{R',S'}}
\\
&
\cev{T}_{R',S'}
\descto{r}{\eqref{eq:b}}
\descto[near start]{dl}{\eqref{diag:cevd}}
\ar{dr}[inner sep =0]{\cev{q}_{R',S'}}
&
{ }
\\
B(R,S) \ar{rr}{B(h_1,k_1)} & & B(R',S')
\end{tikzcd}
\end{equation}
\end{rem} 

\begin{proof}
\begin{enumerate}
\item\label{lem:dbar-welldef:1} We show that $\barB(h,k)$ defined in
  \Cref{cons:lifting-bifunctor}\ref{cons:lifting-bifunctor:2} is a
  $\Gra(\C)$-morphism. In fact, the commutative diagram below shows
  that $\barB(h,k)$ is compatible with left projections; the argument
  for right projections is symmetric.
\begin{equation}\label{diag:Bbar-gra}
\begin{tikzcd}
  T_{R,S}
  \ar{rr}{\barB(h,k)_1}
  \ar{dd}[swap]{\outl_{R,S}}
  \ar{dr}{\cev{d}_{h,k}}
  &
  \descto{d}{\eqref{diag:Bbar}}
  &
  T_{R',S'}
  \ar{dd}{\outl_{R',S'}}
  \ar{dl}[swap,inner sep = 0]{p_{R',S'}}
  \\
  &
  \cev{T}_{R,S}
  \descto{r}{\eqref{eq:out}}
  \descto[near start]{ld}{\eqref{diag:cevd}}
  \ar{rd}[inner sep = 0]{\cev{p}_{R',S'}}
  &
  { } 
  \\
  B(X,Y) \ar{rr}{\barBs(h,k)_0 = B(h_0,k_0)} & & B(X',Y')
\end{tikzcd}
\end{equation}

\item We show that $\barB$ preserves identity morphisms, that is
  \[
    \barB(\id_{(X,R)},\id_{(Y,S)})=\id_{\ol{B}((X,R),(X,R))}
  \]
  for all graphs $(X,R)$ and $(Y,S)$. In the following we omit the
  subscripts of $\id$. By the uniqueness part of the definition of
  $\barB(\id,\id)_1$, it suffices to prove that the following diagram
  commutes:
  \[
    \begin{tikzcd}
      & {T}_{R,S} \ar[bend right=2em]{dl}[swap]{\cev{d}_{\id,\id}} \ar[bend left=2em]{dr}{\vec{d}_{\id,\id}} \ar[dashed]{d}{\id}  & \\
      \cev{T}_{R,S} &  \ar{l}[swap]{p_{R,S}} \ar{r}{{q}_{R,S}} T_{R,S} & \vec{T}_{R,S} 
    \end{tikzcd}
\]
The left-hand triangle commutes because it commutes when postcomposed
with the pullback projections $\cev{q}_{R,S}$ and $\cev{p}_{R,S}$, as
shown by the two commutative diagrams below. The argument for the
right-hand triangle is analogous.
\[
  \begin{tikzcd}[column sep=4em]
    T_{R,S}
    \descto{rd}{\eqref{diag:cevd}}
    \ar{r}{\cev{d}_{\id,\id}}
    \ar[shiftarr = {xshift = -30}]{dd}[swap]{\id}
    \ar{d}[swap]{b_{R,S}}
    & 
    \cev{T}_{R,S}
    \ar[shiftarr = {xshift = 30}]{dd}{\id}
    \ar{d}[swap]{\cev{q}_{R,S}}
    \\
    \llap{\footnotesize\commu\ }B(R,S)
    \descto[pos = .4]{rd}{\eqref{eq:b}}
    \ar{r}{B(\id,\id)=\id}
    &
    B(R,S)\rlap{\ \footnotesize\commu}
    \\
    T_{R,S}
    \ar{r}{p_{R,S}}
    \ar{u}{b_{R,S}}
    &
    \cev{T}_{R,S}
    \ar{u}[swap]{\cev{q}_{R,S}} 
\end{tikzcd}
\qquad
\begin{tikzcd}[column sep=4em]
  T_{R,S}
  \descto{rd}{\eqref{diag:cevd}}
  \ar{r}{\cev{d}_{\id,\id}}
  \ar[shiftarr = {xshift = -30}]{dd}[swap]{\id}
  \ar{d}[swap]{\outl_{R,S}}
  &
  \cev{T}_{R,S}
  \ar{d}{\cev{p}_{R,S}}
  \ar[shiftarr = {xshift = 30}]{dd}{\id}
  \\
  \llap{\footnotesize\commu\ }
  B(X,Y)
  \descto[pos = .4]{rd}{\eqref{eq:out}}
  \ar{r}{B(\id,\id)=\id}
  &
  B(X,Y)
  \rlap{\ \footnotesize\commu}
  \\
  T_{R,S}
  \ar{u}{\outl_{R,S}}
  \ar{r}{p_{R,S}}
  &
  \cev{T}_{R,S}
  \ar{u}[swap]{\cev{p}_{R,S}} 
\end{tikzcd}
\]
(Above and subsequently, when we indicate the reason why a part of a
diagram commutes unless it is easy to see without further reference.)
\item We show that $\barB$ preserves composition, i.e.\
\[ \barB(h\circ h',k'\circ k)=\barB(h',k')\circ \barB(h,k) \]  
for all graph morphisms $h\c (X',R')\to (X,R)$, $h'\c (X'',R'')\to (X',R')$, $k\c (Y,S)\to (Y',S')$ and $k'\c (Y',S')\to (Y'',S'')$. By the uniqueness part of the definition of $\barB(h\circ h',k'\circ k)_1$, it suffices to prove that the following diagram commutes:
\[
\begin{tikzcd}[column sep = 40]
& {T}_{R,S} \ar[bend right=3em]{ddl}[swap]{\cev{d}_{h\circ h', k'\circ k}} \ar[bend left=4em]{ddr}{\vec{d}_{h\circ h',k'\circ k}} \ar{d}{\barBs(h,k)_1}  & \\
& T_{R',S'} \ar{d}{\barBs(h',k')_1} & \\
\cev{T}_{R'',S''} &  \ar{l}[swap]{p_{R'',S''}} \ar{r}{{q}_{R'',S''}} T_{R'',S''} & \vec{T}_{R'',S''} 
\end{tikzcd}
\]
The left-hand part commutes because it commutes when postcomposed with the pullback projections $\cev{q}_{R'',S''}$ and $\cev{p}_{R'',S''}$, as shown by the two commutative diagrams below; in the second diagram, we use that $\barB(h,k)$ is a graph morphism, see~\Cref{lem:dbar-welldef:1}. The argument for the right-hand part is analogous.
\[
\begin{tikzcd}[column sep=4em]
  T_{R,S} \descto{ddr}{\eqref{eq:brs-diag}} \ar{rrr}{\id}
  \ar{dd}[swap]{\barBs(h,k)_1}
  & & &
  T_{R,S} \ar{dl}{b_{R,S}} \ar{dddd}{\cev{d}_{h\circ h', k'\circ k}}
  \\
  & &
  B(R,S)
  \ar{dl}[swap]{B(h_1,k_1)}
  \ar{dd}{B(h_1\circ h_1', k_1'\circ  k_1)}
  & \\
  T_{R',S'} \ar{r}{b_{R',S'}} \ar{dd}[swap]{\barBs(h',k')_1}
  &
  B(R',S') \ar{dr}[swap]{B(h_1',k_1')}
  \descto{r}{\commu}
  &
  { }
  \descto[pos = .8]{r}{\eqref{diag:cevd}}
  &
  { }
  \\
  & &
  B(R'',S'')
  \descto{d}{\eqref{eq:b}}
  &
  \\
  T_{R'',S''}
  \descto{uur}{\eqref{eq:brs-diag}}
  \ar{rrr}{p_{R'',S''}} \ar{urr}{b_{R'',S''}}
  & & { } &
  \cev{T}_{R'',S''} \ar{ul}[swap]{\cev{q}_{R'',S''}} 
\end{tikzcd}
\]
\[
\begin{tikzcd}[column sep=4em]
  T_{R,S} \ar{rrr}{\id} \ar{dd}[swap]{\barBs(h,k)_1}
  & & &
  T_{R,S} \ar{dl}{\outl_{R,S}} \ar{dddd}{\cev{d}_{h\circ h', k'\circ k}}
  \\
  & &
  B(X,Y) \ar{dl}[swap]{B(h_0,k_0)} \ar{dd}{B(h_0\circ h_0', k_0'\circ k_0)}
  \\
  T_{R',S'} \ar{r}{\outl_{R',S'}} \ar{dd}[swap]{\barBs(h',k')_1}
  &
  B(X',Y') \ar{dr}[swap]{B(h_0',k_0')}
    \descto{r}{\commu}
  \descto{luu}{\Cref{lem:dbar-welldef:1}}
  \descto{ldd}{\Cref{lem:dbar-welldef:1}}
  &
  { }
  \descto[pos = .8]{r}{\eqref{diag:cevd}}
  &
  { }
  \\
  & &
  B(X'',Y'')
  \descto{d}{\eqref{eq:out}}
  \\
  T_{R'',S''} \ar{rrr}{p_{R'',S''}} \ar{urr}{\outl_{R'',S''}}
  & & { } &
  \cev{T}_{R'',S''} \ar{ul}[swap]{\cev{p}_{R'',S''}} 
\end{tikzcd}
\]
\item The functor $\barB$ is a lifting of $B$, i.e.~the first diagram in
  \Cref{def:bifunctor-lift} commutes. Indeed,
\[
      \under{\barB((X,R),(Y,S))}
      =B(X,Y)=B(\under{(X,R)},\under{(Y,S)})
\]
and
\[
      \under{\barB(h,k)} = \barB(h,k)_0 = B(h_0,k_0) =
      B(\under{h},\under{k}).\tag*{\qedhere}
\] 
\end{enumerate}
\end{proof}

\begin{example}[$\C=\Set$]\label{ex:graph-lifting-bifunctor} Let us spell out the construction of the lifted bifunctor
\[\barB=\barB_\Gra\c \Gra(\Set)^\opp \times \Gra(\Set)\to \Gra(\Set)\] 
for $B\c \Set^\opp\times \Set\to \Set$. Recall that in $\Set$, the pullback of $f_i\c A_i \to B$ ($i=1,2$) is given by \[P= \{\, (a_1,a_2) \in A_1\times A_2: f_1(a_1)=f_2(a_2)  \,\}\] with the projections 
\[p_i\c P\to A_i, \qquad (a_1,a_2) \mapsto a_i.\] Thus, we arrive at the following concrete description of $\barB$:
\begin{enumerate}
\item The pullback ${T}_{R,S}$ in \eqref{eq:bartrs} is the set of all triples $(f,z,g)$ such that $f,g\in B(X,Y)$, $z\in B(R,S)$, 
\[ B(R,\outl_S)(z)=B(\outl_R,Y)(f) \qqand B(R,\outr_S)(z)=B(\outr_R,Y)(g), \]
and $\outl_{R,S}$ and $\outr_{R,S}$ are the projections $(f,z,g)\mapsto f$ and $(f,z,g)\mapsto g$, respectively. 
\item  The components of the morphism $\barB(h,k)$ are given by
\[ \barB(h,k)_0 = {B}(h_0,k_0)\c B(X,Y)\to B(X',Y') \]
and   
\[ \barB(h,k)_1\c T_{R,S}\to T_{R',S'},\qquad (f,z,g) \mapsto (B(h_0,k_0)(f), B(h_1,k_1)(z), B(h_0,k_0)(g) )  \] 
\end{enumerate}
\end{example}

\begin{construction}\label{cons:bifunctor-lift-rel}
For every bifunctor $B\c \C^\opp\times \C\to\C$ the \emph{canonical relation lifting} $\barB_\Rel$ is defined via the following commutative diagram:
\[ 
\begin{tikzcd}
\Rel(\C)^\opp\times\Rel(\C) \ar[hook]{d}  \ar{r}{\barB_\Rel} & \Rel(\C)  \\
\Gra(\C)^\opp\times\Gra(\C) \ar{r}{\barB_\Gra}  & \Gra(\C) \ar{u}[swap]{(-)^\dag}
\end{tikzcd}
 \]
\end{construction}

\begin{example}
The canonical relation lifting of $B_0(X,Y)=Y+Y^X$ on $\Set$ is the lifting $\barB_0$ described in \Cref{rem:weak-vs-strong}\ref{rem:weak-vs-strong-1}.
\end{example}
The following two propositions show that the canonical liftings satisfy the properties \ref{eq:good-for-simulations-bifunctor-2} and \ref{eq:good-for-simulations-bifunctor-3} of \Cref{def:good-for-simulations-bifunctor}:
\begin{proposition}\label{prop:olB-resp-identities}
For every $B\c \C^\opp\times \C\to \C$ and $X,Y\in \C$,
\begin{eqnarray*}
 (B(X,Y),B(X,Y)) & \cong & \barB_\Gra((X,X),(Y,Y)),\\
 (B(X,Y),B(X,Y)) & \cong & \barB_\Rel((X,X),(Y,Y)).
\end{eqnarray*}
\end{proposition}

\begin{proof}
Immediate from the definitions.
\end{proof}

\begin{proposition}\label{prop:resp-comp}
If $B\c \C^\opp\times \C\to \C$ weakly preserves pullbacks in the second component, then
\begin{eqnarray}
 \barB_\Gra((X,R),(Y,S))\smc\barB_\Gra((X,R'),(Y,S')) &\leq& \barB_\Gra((X,R)\smc
 (X,R'),(Y,S)\smc (Y,S')), \label{eq:olB-resp-graph-comp} \\
 \barB_\Rel((X,R),(Y,S))\bullet\barB_\Rel((X,X),(Y,S')) &\leq& \barB_\Rel((X,R),(Y,S)\bullet (Y,S')), \label{eq:olB-resp-rel-comp}
\end{eqnarray}
where $(X,R),(X,R'),(Y,S),(Y,S')\in \Gra(\C)$ in \eqref{eq:olB-resp-graph-comp} and $(X,R),(Y,S),(Y,S')\in \Rel(\C)$ in \eqref{eq:olB-resp-rel-comp}.
\end{proposition}

\hunote{This strongly suggests that the following more general result is true: Every bifunctor $B\c \C^\opp\times \C\to \C$ that weakly preserves pullbacks in the second component has a canonical extension to a lax bifunctor on the bicategory of spans over $\C$.}

\begin{proof}
\begin{enumerate}
\item Consider the first diagram in \Cref{fig:resp-comp-diag:1}, where the morphisms $\pi,\pi'$ are the projections of the
  pullback $T_{R,S}\smc T_{R',S'}$. Since the endofunctor
  $B(RR',-)\c \C\to\C$ weakly preserves pullbacks by assumption, the
  part $(\ast)$ is a weak pullback square. Note that the outside and
  all remaining parts commute by definition. Thus, there exists a
  morphism $f_{RR',SS'}$ making the two upper parts of the diagram
  commute.
\begin{figure}
  \[
\begin{tikzcd}[row sep=4em, column sep=2em, ampersand replacement=\&]
  \& \& \&
  T_{R,S}T_{R',S'}
  \ar[dashed]{d}{f_{RR',SS'}}
  \ar[%
  rounded corners,
  to path = {--(\tikztotarget |- \tikztostart)\tikztonodes -- (\tikztotarget)}
  ]{llldd}[swap]{\pi}
  \ar[%
  rounded corners,
  to path = {--(\tikztotarget |- \tikztostart)\tikztonodes -- (\tikztotarget)}
  ]{rrrdd}{\pi'}
  \&
  \\
  \& \& \& B(RR',SS') \ar[phantom]{dd}[description]{(\ast)}
  \ar{dl}[swap]{B(RR',\pi_{SS',S})} \ar{dr}{B(RR',\pi_{SS',S'})} \& \&
  \&
  \\
  T_{R,S}
  \ar[%
  rounded corners,
  to path = {--(\tikztostart |- \tikztotarget) -- (\tikztotarget)\tikztonodes}
  ]{rrrddd}{\outr_{R,S}}
  \ar{r}{b_{R,S}}
  \&
  B(R,S)
  \descto[pos = .4]{ld}{\eqref{eq:bartrs}}
  \ar[bend right=1em]{ddr}[swap]{B(R,\outr_S)}
  \ar{r}[yshift=.5em]{B(\pi_{RR',R},S)}
  \&
  B(RR',S)
  \ar{dr}[swap]{B(RR',\outr_S)}
  \& \&
  B(RR',S')
  \ar{dl}{B(RR',\outl_{S'})}
  \&
  B(R',S')
  \descto[pos = .4]{rd}{\eqref{eq:bartrs}}
  \ar[bend left=1em]{ddl}{B(R',\outl_{S'})}  \ar{l}[swap,
  yshift=.5em]{B(\pi_{RR',R'},S')}
  \&
  T_{R',S'}
  \ar[%
  rounded corners,
  to path = {--(\tikztostart |- \tikztotarget) -- (\tikztotarget)\tikztonodes}
  ]{lllddd}[swap]{\outl_{R',S'}}
  \ar{l}[swap]{b_{R',S'}}
  \\
  { }\& \& \&
  B(RR',Y)
\descto{dd}{\eqref{eq:composite-graph}}
\& \& \& { } \& { }
\\
\& \& B(R,Y) \ar{ur}[description]{B(\pi_{RR',R},Y)} \& \&  B(R',Y)
\ar{ul}[description]{B(\pi_{RR',R'},Y)} \& \& \\
\& \& \& B(X,Y) \ar{ul}[description]{B(\outr_R,Y)} \ar{ur}[description]{B(\outl_{R'},Y)} \& \& \& 
\end{tikzcd}
\]
\caption{First diagram for the proof of \Cref{prop:resp-comp}}\label{fig:resp-comp-diag:1}
\end{figure}

\begin{figure}[p!]
  \[
\begin{tikzcd}[row sep=4em, column sep=3.5em, ampersand replacement=\&]
  \& \& \& \&
  B(RR',SS')
  \ar{dl}[description]{B(RR',\pi_{SS',S})}
  \ar{dd}{B(RR',\outl_{SS'})}
  \\
  \& \& { }\&
  B(RR',S)
  \descto{r}{\eqref{eq:composite-graph}}
  \ar{dr}[description]{B(RR',\outl_S)} \& { }
  \\
  \cev{T}_{RR',SS'}
  \ar[%
  rounded corners,
  to path = {--(\tikztostart |- \tikztotarget) -- (\tikztotarget)\tikztonodes}
  ]{uurrrr}{\cev{q}_{RR',SS'}}
  \ar[%
  rounded corners,
  to path = {--(\tikztostart |- \tikztotarget) -- (\tikztotarget)\tikztonodes}
  ]{ddrrrr}{\cev{p}_{RR',SS'}}
  \&
  T_{R,S}T_{R',S'}
  \ar[bend left=1.5em]{uurrr}{f_{RR',SS'}}
  \ar[dashed]{l}[swap]{\cev{g}_{RR',SS'}} \ar{r}{\pi}
  \&
  T_{R,S}
  \descto{u}{\Crefname{figure}{Fig.}{Figs.}(\Cref{fig:resp-comp-diag:1})}
  \ar[bend right=2em]{ddrr}[swap]{\outl_{R,S}}  \ar{r}{b_{R,S}}   \&
  B(R,S)
  \ar{u}{B(\pi_{RR',R},S)} \ar{d}[swap]{B(R,\outl_S)}
  \descto[pos= .2]{ld}{\eqref{diag:cevd}}
  \&
  B(RR',Y)
  \\
  \& \&
  { }
  \&
  B(R,Y)
  \descto{r}{\eqref{eq:composite-graph}}
  \ar{ur}[description]{B(\pi_{RR',R},Y)}
  \& { }
  \\
\& \& \& \& B(X,Y) \ar{ul}[description]{B(\outl_R,Y)} \ar{uu}[swap]{B(\outl_{RR'},Y)}
\end{tikzcd}
\]
\caption{Second diagram for the proof of \Cref{prop:resp-comp}}\label{fig:resp-comp-diag:2}
\end{figure}

\item Next, consider the diagram in \Cref{fig:resp-comp-diag:2}. Note
  that the outside and all remaining parts commute by definition. By
  the universal property of the pullback $\cev{T}_{RR',SS'}$, there
  exists a unique morphism $\cev{g}_{RR',SS'}$ making the top left and
  bottom left parts of the diagram commute.  Symmetrically, we obtain
  a unique $\vec{g}_{RR',SS'}\c T_{R,S}T_{R',S'}\to \vec{T}_{RR',SS'}$
  such that
  \[
    \vec{p}_{RR,SS'} \circ \vec{g}_{RR',SS'}
    =
    f_{RR',SS'}\qqand \vec{q}_{RR',SS'}\circ \vec{g}_{RR',SS'}
    =
    \outr_{R',S'} \circ \pi'.
  \]
  
\item The universal property of the pullback $T_{RR,SS'}$ now yields a unique $g_{RR',SS'}$ making both triangles in the diagram below commute: 
\begin{equation}\label{diag:g}
\begin{tikzcd}
&  T_{R,S}T_{R',S'} \ar[bend right=2em]{dl}[swap]{\cev{g}_{RR',SS'}} \ar[bend left=2em]{dr}{\vec{g}_{RR',SS'}} \ar[dashed]{d}{g_{RR',SS'}}  & \\
\cev{T}_{RR',SS'} &  \ar{l}[swap]{p_{RR',SS'}} \ar{r}{{q}_{RR',SS'}} T_{RR',SS'} & \vec{T}_{RR',SS'} 
\end{tikzcd}
\end{equation}
\item To prove \eqref{eq:olB-resp-graph-comp}, we show that
  \[
    \barB_\Gra((X,R),(Y,S))\smc\barB_\Gra((X,R'),(Y,S'))
    \xra{(\id_{B(X,Y)},g_{RR',SS'})}
    \barB_\Gra((X,R)\smc (X,R'),(Y,S)\smc (Y,S'))
  \]
  is a $\Gra(\C)$-morphism. To this end, consider the commutative
  diagram below, all whose parts commute by definition:
\[ 
\begin{tikzcd}
  T_{R,S}T_{R',S'}
  \ar{dr}{\cev{g}_{RR',SS'}}
  \ar{rr}{g_{RR',SS'}}
  \ar{d}[swap]{\pi}
  & { } &
  T_{RR',SS'} \ar{dd}{\outl_{RR',SS'}} \ar{dl}[swap]{p_{RR',SS'}}
  \\
  T_{R,S}  \ar{d}[swap]{\outl_{R,S}}
  &
  \cev{T}_{RR',SS'}
  \descto[pos=.2]{u}{\eqref{diag:g}}
  \descto{r}{\eqref{eq:out}}
  \descto{l}{\Crefname{figure}{Fig.}{Figs.}(\Cref{fig:resp-comp-diag:2})}
  \ar{dr}{\cev{p}_{RR',SS'}}
  &
  { }
  \\
  B(X,Y) \ar{rr}{\id}
  & &
  B(X,Y) 
\end{tikzcd}
\]
The morphism $\outl_{R,S}\comp \pi$ is the left projection of the graph
$\barB_\Gra((X,R),(Y,S))\smc\barB_\Gra((X,R'),(Y,S'))$, so the above
commutative diagram shows that $(\id,g_{RR',SS'})$ is compatible with
left projections. The proof that $(\id,g_{RR',SS'})$ is
compatible with right projections is symmetric.
\item Finally, \eqref{eq:olB-resp-rel-comp} follows from the computation
\begin{align*}
 \barB_\Rel((X,R),(Y,S))\bullet\barB_\Rel((X,X),(Y,S')) & = \barB_\Gra((X,R),(Y,S))^\dag\bullet \barB_\Gra((X,X),(Y,S'))^\dag   \\
& = (\barB_\Gra((X,R),(Y,S))\smc \barB_\Gra((X,X),(Y,S')))^\dag  \\
& \leq \barB_\Gra((X,R)\smc
 (X,X),(Y,S)\smc (Y,S'))^\dag \\
& \cong \barB_\Gra((X,R),(Y,S)\smc (Y,S'))^\dag \\
& \leq \barB_\Gra((X,R),(Y,S)\bullet (Y,S'))^\dag \\
& =
 \barB_\Rel((X,R),(Y,S)\bullet (Y,S')).
\end{align*}
The first step uses the definition of $\barB_\Rel$; the second step follows from \Cref{lem:reflection-comp}; the third step uses \eqref{eq:olB-resp-graph-comp}; the fourth step uses that $(X,R)\smc (X,X) \cong (X,R)$ by definition of graph composition; the fifth step uses that $(Y,S)\smc (Y,S')\leq (Y,S)\bullet (Y,S')$ by definition of relation composition; the last step uses the definition of $\barB_\Rel$.
\qedhere
\end{enumerate}
\end{proof}

\section{Canonical Liftings of Higher-Order GSOS Laws}\label{sec:can-lift-gsos-laws}
Next, we show how to lift higher-order abstract GSOS laws from $\C$ to $\Gra(\C)$ and $\Rel(\C)$. 

\begin{rem} The notion of relation lifting of a higher-order GSOS law is given in \Cref{def:rho-lift}. \emph{Graph liftings} are defined analogously: replace $\Rel$ by $\Gra$ and $\vee$ by $+$. Note that unlike relation liftings, graphs liftings are usually not unique. 
\end{rem}

\begin{construction}\label{cons:rho-lift}
Let $\Sigma\c \C\to \C$ and $B\c \C^\opp \times \C \to \C$ be functors with their canonical graph liftings 
\[\ol{\Sigma}=\ol{\Sigma}_\Gra\c \Gra(\C)\to \Gra(\C) \qqand \barB=\barB_\Gra\c \Gra(\C)^\opp\times \Gra(\C)\to \Gra(\C) \]
given by \Cref{cons:lifting-endofunctor} and \Cref{cons:lifting-bifunctor}, and let 
\[ \rho_{X,Y}\c \Sigma(X\times B(X,Y))\to B(X,\Sigmas(X+Y)) \qquad ((X,p_X)\in \Pt/\C,Y\in \C)\]
be a $\Pt$-pointed higher-order GSOS law
of $\Sigma$ over $B$. The \emph{canonical graph lifting} of $\rho$ is the $(\Pt,\Pt)$-pointed higher-order GSOS law $\barrho=\barrho^\Gra$ of $\ol{\Sigma}$ over $\barB$ whose components
\[
\begin{tikzcd} \ol{\Sigma}((X,R)\times \barB((X,R),(Y,S))) \ar{d}{\barrho_{(X,R),(Y,S))}}\\ \barB((X,R),\ol{\Sigma}^\star((X,R)+(Y,S)))
\end{tikzcd}
\qquad\qquad (((X,R),(p_{(X,R)})\in (V,V)/\Gra(\C),\,(Y,S)\in \Gra(\C))
\]
are defined as follows. First, put
\[ (\barrho_{(X,R),(Y,S)})_0 \c \big(\Sigma(X\times B(X,Y))\xra{\rho_{X,Y}} B(X,\Sigmas(X+Y))\big). \]
Second, we define 
\[ (\barrho_{(X,R),(Y,S)})_1\c \Sigma(R\times T_{R,S})\to T_{R,\Sigmas(R+S)} \]
in two steps via the universal properties of the pullbacks occurring in the construction of $\barB$: 
\begin{enumerate}
\item Consider the diagram below, where we regard the objects $X$ and $R$ as $V$-pointed by the morphisms $p_X=(p_{(X,R)})_0\c V\to X$ and $p_R=(p_{(X,R)})_1\c V\to R$.
  \begin{equation}\label{diag:fig4}
    \hspace*{-5pt}
    \mbox{
    \begin{tikzcd}[cramped,row sep = 30, column sep=4em, ampersand
      replacement=\&,baseline = (U.north)]
      { } \& { } \& { } \&
      B(R,\Sigmas(R+S))
      \ar[shiftarr = {xshift = 50}]{dd}[rotate = 90, pos = .9, yshift=-5]{B(R,\Sigmas(\outl_R+\outl_S))}
      \ar{d}[swap]{B(R,\Sigmas(R+\outl_S))}
      \ar[phantom]{d}[xshift=20]{\text{\footnotesize\commu}}
      \\
      \&
      \Sigma(R\times B(R,S))
      \ar[%
      rounded corners,
      to path = {--(\tikztostart |- \tikztotarget) -- (\tikztotarget)\tikztonodes}
      ]{urr}{\rho_{R,S}}
      \ar{r}[yshift=.3em]{\Sigma(R\times B(R,\outl_S))}
      \descto{rd}{($*$)}
      \&
      \Sigma(R\times B(R,Y))
      \ar{r}{\rho_{R,Y}}
      \descto[pos = .4]{u}{\text{(nat.~$\rho_{R,-}$)}}
      \&
      B(R,\Sigmas(R+Y))
      \ar{d}[description]{B(R,\Sigmas(\outl_R+Y))}
      \\
      \cev{T}_{R,\Sigmas(R+S)}
      \ar[%
      rounded corners,
      to path={--([yshift=20]\tikztostart |- \tikztotarget.center) --
        ([yshift=20]\tikztotarget.center)\tikztonodes --
        (\tikztotarget)}
      ]{uurrr}{\cev{q}_{R,\Sigmas(R+S)}}
      \ar[%
      rounded corners,
      to path={--([yshift=-20]\tikztostart |- \tikztotarget.center) --
        ([yshift=-20]\tikztotarget.center)\tikztonodes --
        (\tikztotarget)}
      ]{drrr}{\cev{p}_{R,\Sigmas(R+S)}}
      \&
      \Sigma(R\times T_{R,S})
      \ar[dashed]{l}[swap]{(\cev{\rho}_{(X,R),(Y,S)})_1}
      \ar{u}{\Sigma(R\times b_{R,S})}
      \ar{r}{\Sigma(R\times \outl_{R,S})}
      \ar{dr}[swap]{\Sigma(\outl_R\times \outl_{R,S})}
      \& 
      \Sigma(R\times B(X,Y))
      \descto[pos=.2]{ld}{\commu}
      \ar{u}[description]{\Sigma(R\times B(\outl_R,Y))}
      \ar{d}[description]{\Sigma(\outl_R\times B(X,Y))}
      \descto[yshift=15]{r}{\text{(dinat.~$\rho_{-,Y}$)}}
      \&
      |[alias = A]|
      B(R,\Sigmas(X+Y))
      \\
      \& { } \&
      \Sigma(X\times B(X,Y))
      \ar{r}{\rho_{X,Y}}
      \&
      |[alias = B]|
      B(X,\Sigmas(X+Y))
      \arrow[u, "{B(\outl_R,\Sigmas(X+Y))}" description, ""{name=U}]{}
    \end{tikzcd}
    \hspace*{-30pt}
    }
  \end{equation}
  %
  %

  Its outside commutes due to~\eqref{eq:bartrs} and using $\outl_{R+S}= \outl_R + \outl_S$ (see \Cref{sec:graph-cocomplete}), and for the part marked ($*$) we remove $\Sigma$ and consider the
  product components separately: the left-hand one is the identity on
  $R$, and for the right-hand one we have the commutative diagram
  below:
  \[
    \begin{tikzcd}
      B(R,S)
      \ar{rr}{B(R,\outl_S)}
      &&
      B(R,Y)
      \ar{dd}{B(\outl_R, Y)}
      \\
      \descto{r}{\eqref{eq:b}}
      &
      \cev{T}_{R,S}
      \descto{ru}{\eqref{eq:bartrs}}
      \descto[pos=.4]{d}{\eqref{eq:out}}
      \ar{lu}[swap]{\cev{q}_{R,S}}
      \ar{rd}{\cev{p}_{R,S}}
      \\
      T_{R,S}
      \ar{rr}{\outl_{R,S}}
      \ar{ru}{p_{R,S}}
      \ar{uu}{b_{R,S}}
      & { } &
      B(X,Y)
    \end{tikzcd}
  \]
  The universal property of the pullback
  $\cev{T}_{R,\Sigmas(R+S)}$ now yields a unique morphism
  $(\cev{\rho}_{(X,R),(Y,S)})_1$ making the top and bottom part of the
  diagram ab commute. Analogously, we obtain a unique morphism
  $(\vec{\rho}_{R,\Sigmas(R+S)})_1\c \Sigma(R\times T_{R,S})\to
  \vec{T}_{R,\Sigmas(R+S)}$ such that
  \begin{align*}
    \vec{p}_{R,\Sigmas(R+S)}\circ (\vec{\rho}_{R,\Sigmas(R+S)})_1
    &=
    \rho_{R,S}\circ \Sigma(R\times b_{R,S}), \text{and}
    \\
    \vec{q}_{R,\Sigmas(R+S)} \circ (\vec{\rho}_{R,\Sigmas(R+S)})_1
    &=
    \rho_{X,Y}\circ \Sigma(\outr_R\times \outr_{R,S}).
  \end{align*}
  
\item We take $(\barrho_{(X,R),(Y,S)})_1$ to be the unique morphism
  making both triangles in the diagram below commute, using the
  universal property of the pullback $T_{R,\Sigmas(R+S)}$:
  \begin{equation}\label{diag:barrho}
    \begin{tikzcd}[row sep=3em]
      &
      \Sigma(R\times T_{R,S})
      \ar[bend right=2em]{dl}[swap]{(\cev{\rho}_{(X,R),(Y,S)})_1}
      \ar[bend left=2em]{dr}{(\vec{\rho}_{(X,R),(Y,S)})_1}
      \ar[dashed]{d}[description]{(\barrho_{(X,R),(Y,S)})_1}
      &
      \\
      \cev{T}_{R,\Sigmas(R+S)}
      &
      \ar{l}[swap]{p_{R,\Sigmas(R+S)}}
      \ar{r}{q_{R,\Sigmas(R+S)}} T_{R,\Sigmas(R+S)}
      &
      \vec{T}_{R,\Sigmas(R+S)} 
    \end{tikzcd}
  \end{equation}
\end{enumerate}
\end{construction}
\begin{proposition}\label{prop:rho-lift-cat}
  The family $\barrho^\Gra$ is a $(V,V)$-pointed higher-order GSOS law of $\ol{\Sigma}_\Gra$ over $\barB_\Gra$.
\end{proposition}
\begin{remark}
  We shall use in the proof below that the following diagram commutes:
  \begin{equation}\label{diag:II}
    \begin{tikzcd}[column sep = 50]
      \Sigma(R \times T_{R,S})
      \ar{r}{(\ol\rho_{(X,R),(Y,S)})_1}
      \ar{dd}[swap]{\Sigma(R\times b_{R,S})}
      \ar{rd}[swap]{(\cev\rho_{(X,R),(Y,S)})_1}
      &
      T_{R,\Sigmas(R+S)}
      \ar{d}{p_{R,\Sigmas(R+S)}}
      \ar[shiftarr = {xshift=50}]{dd}{b_{R,\Sigmas(R+S)}}
      \descto[pos=.2]{ld}{\eqref{diag:barrho}}
      \\
      { }
      &
      \cev{T}_{R,\Sigmas(R+S)}
      \ar{d}{\cev{q}_{R,\Sigmas(R+S)}}
      \descto[pos=.4]{ld}{\eqref{diag:fig4}}
      \descto[pos=.1]{r}{\eqref{eq:bartrs}}
      &
      { }
      \\
      \Sigma(R\times B(R,S))
      \ar{r}{\rho_{R,S}}
      &
      B(R,\Sigmas(R+S))
    \end{tikzcd}
  \end{equation}
\end{remark}
\begin{proof}
\begin{enumerate}
\item We show that $\barrho_{(X,R),(Y,S)}$ is a $\Gra(\C)$-morphism
  for every $(X,R),p_{(X,R)})\in (V,V)/\Gra(\C)$ and $(Y,S)\in \Gra(\C)$. In fact, that
  $\barrho_{(X,R),(Y,S)}$ is compatible with left projections is
  shown by the diagram below, all of whose parts commute by
  definition. The proof that $\barrho_{(X,R),(Y,S)}$ is compatible
  with right projections is symmetric.
  \begin{equation}\label{diag:barrho-gra}
    \begin{tikzcd}[row sep=3em]
      \Sigma(R\times T_{R,S})
      \ar{dd}[swap]{\Sigma(\outl_R\times \outl_{R,S})}
      \ar{rr}{(\barrho_{(X,R),(Y,S)})_1}
      \ar{dr}[description]{(\cev{\rho}_{(X,R),(Y,S)})_1}
      &
      { }
      &
      T_{R,\Sigmas(R+S)} \ar{dd}{\outl_{R,\Sigmas(R+S)}}
      \ar{dl}[description]{p_{R,\Sigmas(R+S)}}
      \\
      { }
      &
      \cev{T}_{R,\Sigmas(R+S)}
      \ar{dr}[description]{\cev{p}_{R,\Sigmas(R+S)}}
      \descto{u}{\eqref{diag:barrho}}
      \descto{r}{\eqref{eq:out}}
      \descto{l}{\eqref{diag:fig4}}
      & { }
      \\
      \Sigma(X\times B(X,Y))
      \ar{rr}{(\barrho_{(X,R),(Y,S)})_0\,=\,\rho_{X,Y}}
      & &
      B(X,\Sigmas(X+Y)) 
\end{tikzcd}
\end{equation}
\item To prove naturality, we need to show that for every $(X,R),p_{(X,R)})\in (V,V)/\Gra(\C)$ and every $\Gra(\C)$-morphism $k\c (Y,S)\to (Y',S')$ the diagram below commutes:
\[
\begin{tikzcd}[column sep=5em, row sep=3em]
\ol{\Sigma}((X,R)\times\barB((X,R),(Y,S))) \ar{r}{\barrho_{(X,R),(Y,S)}} \ar{d}[swap]{\Sigma((X,R)\times B((X,R),k))} & \barB((X,R),\ol{\Sigma}^\star((X,R)+(Y,S))) \ar{d}{\barBs((X,R),\ol{\Sigma}^\star((X,R)+k))} \\
\ol{\Sigma}((X,R)\times\barB((X,R),(Y',S'))) \ar{r}{\barrho_{(X,R),(Y',S')}}  & \barB((X,R),\ol{\Sigma}^\star((X,R)+(Y',S'))) 
\end{tikzcd}
\]  
Commutativity in the $(-)_0$-component is clear because $(\barrho_{(X,R),(Y,S)})_0=\rho_{X,Y}$ and $(\barrho_{(X,R),(Y',S')})_0=\rho_{X,Y'}$ by definition and $\rho_{X,-}$ is natural. Commutativity in the $(-)_1$-component amounts to showing that the following rectangle commutes:
\begin{equation}\label{eq:nat-bar-rho}
\begin{tikzcd}[column sep=5em, row sep=3em]
\Sigma(R\times T_{R,S}) \ar{r}{(\barrho_{(X,R),(Y,S)})_1} \ar{d}[swap]{\Sigma(R\times \barBs((X,R),k)_1)}  & T_{R,\Sigmas(R+S)}  \ar{d}{ \barBs((X,R),\ol{\Sigma}^\star((X,R)+k))_1 } \\
\Sigma(R\times T_{R,S'}) \ar{r}{(\barrho_{(X,R),(Y',S')})_1} & T_{R,\Sigmas(R+S')} 
\end{tikzcd}
\end{equation}
Indeed, the  diagrams below show that \eqref{eq:nat-bar-rho} commutes when postcomposed with 
\[ \outl_{R,\Sigmas(R+S')} = \cev{p}_{R,\Sigmas(R+S')}\circ
  p_{R,\Sigmas(R+S')} \qqand
  b_{R,\Sigmas(R+S')}=\cev{q}_{R,\Sigmas(R+S')}\circ
  p_{R,\Sigmas(R+S')}.\]

\[
\begin{tikzcd}[row sep=2em, column sep=4em, ampersand replacement=\&]
  \Sigma(R\times T_{R,S})
  \ar{rrr}{(\barrho_{(X,R),(Y,S)})_1}
  \ar{dr}[description]{\Sigma(\outl_R\times \outl_{R,S})}
  \ar{ddd}[description]{\Sigma(R\times \barBs((X,R),k)_1)}
  \& \& { } \&
  T_{R,\Sigmas(R+S)}
  \ar{dl}[description]{\outl_{R,\Sigmas(R+S)}}
  \ar{ddd}[description, pos=.6]{
    \barBs((X,R),\ol{\Sigma}^\star((X,R)+k))_1 }
  \\
  \&
  \Sigma(X\times B(X,Y))
  \descto{ru}{\eqref{diag:barrho-gra}}
  \descto{rd}{\text{(nat. $\rho_{X,-}$)}}
  \descto{ld}{\eqref{diag:Bbar-gra}}
  \ar{r}{\rho_{X,Y}}
  \ar{d}[swap]{\Sigma(R\times B(R,k_0))}
  \&
  B(X,\Sigmas(X+Y)) \ar{d}[description]{B(R,\Sigmas(R+k_0))}
  \descto{rd}{\eqref{diag:Bbar-gra}}
  \&
  \\
  { }
  \&
  \Sigma(X\times B(X,Y')) \ar{r}{\rho_{X,Y'}}
  \descto{rd}{\eqref{diag:barrho-gra}}
  \&
  B(X,\Sigmas(X+Y'))
  \& { }
  \\
  \Sigma(R\times T_{R,S'})
  \ar{ur}[description]{\Sigma(\outl_R\times \outl_{R,S'})}
  \ar{rrr}{(\barrho_{(X,R),(Y',S')})_1}
  \& \& { }\&
  T_{R,\Sigmas(R+S')}  \ar{ul}[description]{\outl_{R,\Sigmas(R+S')}}
\end{tikzcd}
\]
\[
\begin{tikzcd}[row sep=2em, column sep=4em, ampersand replacement=\&]
  \Sigma(R\times T_{R,S})
  \ar{rrr}{(\barrho_{(X,R),(Y,S)})_1}
  \ar{dr}[description]{\Sigma(R\times b_{R,S})}
  \ar{ddd}[description]{\Sigma(R\times \barBs((X,R),k)_1)}
  \& \& { } \&
  T_{R,\Sigmas(R+S)}
  \ar{dl}[description]{b_{R,\Sigmas(R+S)}}
  \ar{ddd}[description, pos=.6]{
    \barBs((X,R),\ol{\Sigma}^\star((X,R)+k))_1 }
  \\
  \&
  \Sigma(R\times B(R,S))
  \descto{rd}{(\text{nat. $\rho_{R,-}$})}
  \descto{ru}{\eqref{diag:II}}
  \descto{ld}{\eqref{eq:brs-diag}}
  \ar{r}{\rho_{R,S}} \ar{d}[swap]{\Sigma(R\times B(R,k_1))}
  \&
  B(R,\Sigmas(R+S))
  \ar{d}[description]{B(R,\Sigmas(R+k_1))}
  \descto{rd}{\eqref{eq:brs-diag}}
  \& 
  \\
  { }
  \&
  \Sigma(R\times B(R,S'))
  \descto{rd}{\eqref{diag:II}}
  \ar{r}{\rho_{R,S'}}
  \&
  B(R,\Sigmas(R+S'))
  \&
  { }
  \\
  \Sigma(R\times T_{R,S'})
  \ar{ur}[description]{\Sigma(R\times b_{R,S'})}
  \ar{rrr}{(\barrho_{(X,R),(Y',S')})_1}
  \& \& { } \&
  T_{R,\Sigmas(R+S')} \ar{ul}[description]{b_{R,\Sigmas(R+S')}}
\end{tikzcd}
\]

 Since $\cev{q}_{R,\Sigmas(R+S')}$ and
$\cev{p}_{R,\Sigmas(R+S')}$ are the projections of the pullback
$\cev{T}_{R,\Sigmas(R+S')}$ and thus jointly monic, it follows that
\eqref{eq:nat-bar-rho} commutes when postcomposed with
$p_{R,\Sigmas(R+S')}$. A symmetric argument shows that
\eqref{eq:nat-bar-rho} commutes when postcomposed with
$q_{R,\Sigmas(R+S')}$. Finally, since $p_{R,\Sigmas(R+S')}$ and
$q_{R,\Sigmas(R+S')}$ are the projections of the pullback
$T_{R,\Sigmas(R+S')}$ and thus jointly monic, we conclude
that~\eqref{eq:nat-bar-rho} commutes.
\item To prove dinaturality, we need to show that for every $(Y,S)\in \Gra(\C)$ and every $(V,V)/\Gra(\C)$-morphism $h\c (X',R')\to (X,R)$ the diagram below commutes:
\[
\begin{tikzcd}[column sep=5em]
 \ol{\Sigma}((X,R)\times \barB((X,R),(Y,S))) \ar{r}{\rho_{(X,R),(Y,S)}} & \barB((X,R),\ol{\Sigma}^\star((X,R)+(Y,S))) \ar{d}{\barBs(h,\ol{\Sigma}^\star((X,R)+(Y,S)))}  \\
\ol{\Sigma}((X',R')\times \barB((X,R),(Y,S))) \ar{u}{\ol{\Sigma}(h\times \barBs((X,R),(Y,S)))} \ar{d}[swap]{\ol{\Sigma}((X',R')\times \barBs(h,(Y,S))) } & \barB((X',R'),\ol{\Sigma}^\star((X,R)+(Y,S)))  \\
 \ol{\Sigma}((X',R')\times \barB((X',R'),(Y,S))) \ar{r}{\rho_{(X',R'),(Y,S)}} & \barB((X',R'),\ol{\Sigma}^\star((X',R')+(Y,S))) \ar{u}[swap]{\barBs((X',R'),\ol{\Sigma}^\star(h+(Y,S))) }
\end{tikzcd}
\]
Commutativity in the $(-)_0$-component is clear because $(\barrho_{(X,R),(Y,S)})_0=\rho_{X,Y}$ and $(\barrho_{(X',R'),(Y,S)})_0=\rho_{X',Y}$ by definition and $\rho_{-,Y}$ is dinatural. Commutativity in the $(-)_1$-component amounts to showing that the following diagram commutes:
\begin{equation}\label{eq:dinat-bar-rho}
\begin{tikzcd}[column sep=5em]
 \Sigma(R\times T_{R,S}) \ar{r}{(\rho_{(X,R),(Y,S)})_1} & T_{R,\Sigmas(R+S)} \ar{d}{\barBs(h,\ol{\Sigma}^\star((X,R)+(Y,S)))_1}  \\
 \Sigma(R'\times T_{R,S}) \ar{u}{\Sigma(h_1\times T_{R,S})} \ar{d}[swap]{\Sigma(R'\times \barBs(h,(Y,S))_1)} & T_{R',\Sigmas(R+S)} \\
 \Sigma(R'\times T_{R',S}) \ar{r}{(\rho_{(X',R'),(Y,S)})_1} & T_{R',\Sigmas(R'+S)} \ar{u}[swap]{\barBs((X',R'),\ol{\Sigma}^\star(h+(Y,S)))_1 }
\end{tikzcd}
\end{equation}
The argument is similar to the one for naturality: The two diagrams below show that~\eqref{eq:dinat-bar-rho} commutes when postcomposed
with
\[\outl_{R',\Sigmas(R+S)}=\cev{p}_{R',\Sigmas(R+S)}\circ
  p_{R',\Sigmas(R+S)} \qqand b_{R',\Sigmas(R+S)} =
  \cev{q}_{R',\Sigmas(R+S)}\circ p_{R',\Sigmas(R+S)}.\] 
Thus it
commutes when postcomposed with $p_{R',\Sigmas(R'+S)}$ (and
analogously for $q_{R',\Sigmas(R'+S)}$), and so it commutes.\qedhere
\[
\begin{tikzcd}[row sep=2em, column sep=3.75em, ampersand replacement=\&]
  \Sigma(R\times T_{R,S})
  \ar{dr}[description]{\Sigma(\outl_R\times\outl_{R,S})}
  \ar{rrr}{(\ol\rho_{(X,R),(Y,S)})_1}
  \& \& { }\&
  T_{R,\Sigmas(R+S)}
  \ar{dl}[description]{\outl_{R,\Sigmas(R+S)}}
  \ar{dd}[description]{\barBs(h,\ol{\Sigma}^\star((X,R)+(Y,S)))_1}
  \\
  { }
  \&
  \Sigma(X\times B(X,Y))
  \descto{l}{(\text{morph. $h$})}
  \ar{r}{\rho_{X,Y}}
  \descto{dr}{(\text{dinat. $\rho_{-,Y}$})}
  \descto{ru}{\eqref{diag:barrho-gra}}
  \&
  B(X,\Sigmas(X+Y))
  \descto[pos=.3]{r}{\eqref{eq:brs-diag}}
  \ar{d}[description]{B(h_0,\Sigmas(X+Y))}
  \&
  { }
  \\
  \Sigma(R'\times T_{R,S})
  \descto{dr}{\eqref{eq:brs-diag}}
  \ar{r}[yshift=.5em]{\Sigma(\outl_{R'}\times \outl_{R,S})}
  \ar{uu}[description]{\Sigma(h_1\times T_{R,S})}
  \ar{dd}[description]{\Sigma(R'\times \barBs(h,(Y,S))_1)}
  \&
  \Sigma(X'\times B(X,Y))
  \ar{u}[description]{\Sigma(h_0\times B(X,Y))}
  \ar{d}[description]{\Sigma(X'\times B(h_0,Y))}
  \&
  B(X',\Sigmas(X+Y))
  \&
  T_{R',\Sigmas(R+S)} \ar{l}[swap]{\outl_{R',\Sigmas(R+S)}}
  \\
  \&
  \Sigma(X'\times B(X',Y))
  \descto{rd}{\eqref{diag:barrho-gra}}
  \ar{r}{\rho_{X',Y}}
  \&
  B(X',\Sigmas(X'+Y))
  \descto{ru}{\eqref{eq:brs-diag}}
  \ar{u}[description]{B(X',\Sigmas(h_0+Y))}
  \&
  \&
  { }
  \\
  \Sigma(R'\times T_{R',S})
  \ar{ur}[description]{\Sigma(\outl_{R'}\times \outl_{R',S})}
  \ar{rrr}{(\rho_{(X',R'),(Y,S)})_1}
  \& \& { }\&
  T_{R',\Sigmas(R'+S)}
  \ar{uu}[description]{\barBs((X',R'),\ol{\Sigma}^\star(h+(Y,S)))_1 }
  \ar{ul}[description]{\outl_{R',\Sigmas(R'+S)}}
\end{tikzcd}
\]
\[
\begin{tikzcd}[row sep=2em, column sep=4em, ampersand replacement=\&]
  \Sigma(R\times T_{R,S})
  \ar{dr}[description]{\Sigma(R\times b_{R,S})}
  \ar{rrr}{(\rho_{(X,R),(Y,S)})_1}
  \& \& { }\&
  T_{R,\Sigmas(R+S)}
  \ar{dl}[description]{b_{R,\Sigmas(R+S)}}
  \ar{dd}[description]{\barBs(h,\ol{\Sigma}^\star((X,R)+(Y,S)))_1}
  \\
  \&
  \Sigma(R\times B(R,S))
  \descto{ru}{\eqref{diag:II}}
  \ar{r}{\rho_{R,S}}
  \descto{dr}{(\text{dinat. $\rho_{-,S}$})}
  \&
  B(R,\Sigmas(R+S))
  \descto[pos=.3]{r}{\eqref{eq:brs-diag}}
  \ar{d}[description]{B(h_1,\Sigmas(R+S))}
  \&
  { }
  \\
  \Sigma(R'\times T_{R,S})
  \descto{dr}{\eqref{eq:brs-diag}}
  \ar{r}[yshift=.5em]{\Sigma(R'\times b_{R,S})}
  \ar{uu}[description]{\Sigma(h_1\times T_{R,S})}
  \ar{dd}[description]{\Sigma(R'\times \barBs(h,(Y,S))_1)}
  \&
  \Sigma(R'\times B(R,S))
  \ar{u}[description]{\Sigma(h_1\times B(R,S))}
  \ar{d}[description]{\Sigma(R'\times B(h_1,Y))}
  \&
  B(R',\Sigmas(R+S))
  \&
  T_{R',\Sigmas(R+S)}
  \ar{l}[swap]{b_{R',\Sigmas(R+S)}}
  \\
  \&
  \Sigma(R'\times B(R',S)) \ar{r}{\rho_{R',S}}
  \descto{dr}{\eqref{diag:II}}
  \&
  B(R',\Sigmas(R'+S))
  \descto{ru}{\eqref{eq:brs-diag}}
  \ar{u}[description]{B(R',\Sigmas(h_1+Y))}
  \&
  \\
  \Sigma(R'\times T_{R',S})
  \ar{ur}[description]{\Sigma(R'\times b_{R',S})}
  \ar{rrr}{(\rho_{(X',R'),(Y,S)})_1}
  \& \& { } \&
  T_{R',\Sigmas(R'+S)}
  \ar{uu}[description]{\barBs((X',R'),\ol{\Sigma}^\star(h+(Y,S)))_1 }
  \ar{ul}[description]{b_{R',\Sigmas(R'+S)}}
\end{tikzcd}
\]
\end{enumerate}
\end{proof}
Relation liftings of higher-order GSOS laws can be derived from their canonical graph liftings.
\begin{construction}\label{cons:rho-lift-rel}
Let $\Sigma\c \C\to \C$ and $B\c \C^\opp \times \C \to \C$ be functors with their canonical relation liftings
\[\ol{\Sigma}_\Rel\c \Rel(\C)\to \Rel(\C) \qqand \barB_\Rel\c \Rel(\C)^\opp\times \Rel(\C)\to \Rel(\C), \]
and suppose that $\Sigma$ preserves strong epimorphisms. Every $\Pt$-pointed higher-order GSOS law
\[ \rho_{X,Y}\c \Sigma(X\times B(X,Y))\to B(X,\Sigmas(X+Y)) \qquad ((X,p_X)\in \Pt/\C,Y\in \C)\]
of $\Sigma$ over $B$ has a (necessarily unique) lifting to a $(V,V)$-pointed higher-order GSOS law $\barrho^\Rel$ of $\ol{\Sigma}_\Rel$ over $\barB_\Rel$. Its component at $((X,R),(p_{(X,R)})\in (V,V)/\Rel(\C)$ and $(Y,S)\in \Rel(\C))$ is given by the composite
\[
\begin{tikzcd} 
\ol{\Sigma}_\Rel((X,R)\times \barB_\Rel((X,R),(Y,S))) \ar[equals]{d}\\
\ol{\Sigma}_\Rel((X,R)\times \barB_\Gra((X,R),(Y,S))^\dag) \ar{d}{\cong}\\
(\ol{\Sigma}_\Gra((X,R)\times \barB_\Gra((X,R),(Y,S))))^\dag \ar{d}{(\barrho^\Gra_{(X,R),(Y,S)})^\dag}\\
\barB_\Gra((X,R),\ol{\Sigma}_\Gra^\star((X,R)+(Y,S)))^\dag 
\ar{d}{\barBs_\Gra(\id, \ol{\Sigma}_\Gra^\star e)^\dag}\\
\barB_\Gra((X,R),\ol{\Sigma}_\Gra^\star((X,R)\vee (Y,S)))^\dag 
\ar{d}{\barBs_\Gra(\id, h)^\dag}\\
\barB_\Gra((X,R),\ol{\Sigma}_\Rel^\star((X,R)\vee (Y,S)))^\dag \ar[equals]{d} \\
\barB_\Rel((X,R),\ol{\Sigma}_\Rel^\star((X,R)\vee (Y,S))) \\
\end{tikzcd}
\]
Here the isomorphism in the second step follows from \Cref{lem:lift-dag-prod}, and
\[e\c (X,R)+(Y,S)\epito ((X,R)+(X,S))^\dag= (X,R)\vee (Y,S)\]
and 
\[ 
h\c \ol{\Sigma}_\Gra^\star((X,R)\vee (X,S))\epito \ol{\Sigma}_\Rel^\star((X,R)\vee (X,S))
\]
are the reflections. For the latter recall that the free $\ol{\Sigma}_\Rel$-algebra on a relation $(X,T)$ is given by applying $(-)^\dag$ to the free $\ol{\Sigma}_\Gra$-algebra on $(X,T)$, see \Cref{cor:free-algebra-adjunction}.
\end{construction}

\section{The $\lambda$-Calculus}\label{sec:lambda}We give a more detailed account of the $\lambda$-calculus in the higher-order abstract GSOS framework. Recall from \Cref{sec:lambda-sketch} that we work with the functors
\[  
\Sigma X = V + \delta X + X\times X \qqand B_0(X,Y)=\llangle X,Y\rrangle \times (Y+Y^X+1)
\]
on the presheaf category $\vcat$, where $\fset$ is the category of finite cardinals and functions and the presheaves $V$, $\delta X$ and $\llangle X,Y\rrangle$ are given by
\[ V(n)=n,\qquad \delta X(n) = X(n+1),\qquad \llangle X,Y\rrangle(n)=\vcat(X^n,Y).  \]
 The initial algebra for $\Sigma$ is the presheaf $\Lambda$ of $\lambda$-terms modulo $\alpha$-equivalence~\cite{DBLP:conf/lics/FiorePT99}. To introduce the higher-order GSOS law for the $\lambda$-calculus, we need some notation.

\begin{notation}
\begin{enumerate}
\item Given $X\in \vcat$ we sometimes write $X_n$ for $X(n)$. For a presheaf morphism (i.e.\ a natural transformation) $f\c X\to Y$, we drop subscripts of components and write $f$ for $f_n\c X_n\to Y_n$.
\item We let $\ev\c Y^X\times X\to X$ denote the evaluation morphism of the exponential object $Y^X$. Given $n\in \fset$, $f\in Y^X(n)$ and $e\in X(n)$ we write $f(e)$ for $\ev(f,e)$.
\item We define the maps $n \xrightarrow{\oname{old}_{n}} n + 1 \xleftarrow{\oname{new}_{n}} 1$ by $\oname{old}_{n}(i)=i$ and $\oname{new}_{n}(0)=n$.
\item For a presheaf $X\in\vcat$ we define $\oname{up}_{X,n}
    = \big(\,X(n) \xra{X(\oname{old}_n)} X(n + 1)\,\big)$.
\item Given a pointed presheaf $(X,\var)\in \V/\vcat$ and a presheaf $Y\in \vcat$ we define
$\rho_{1} \c \delta\llangle X,Y \rrangle \to \llangle X,\delta Y \rrangle$
to be the map sending a natural transformation $f\colon X^{n+1}\to Y$ to the natural transformation $\rho_1(f)\c X^n\to \delta Y$ given by
\[
  \vec{u}
  \in X(m)^{n} \quad \mapsto \quad f_{m+1}(\oname{up}_{X,m}(\vec{u}),\var_{m+1}(\oname{new}_m)) \in Y(m+1).
\]
\item Similarly, for a pointed presheaf $(X,\var)\in \V/\vcat$ and a presheaf $Y\in \vcat$ we define the map $\rho_{2} \c \delta\llangle X,Y \rrangle \to Y^{X}$ by 
\[ \rho_2(f)(e) = f_{n}(\var_n(0),\ldots, \var_n(n-1),e)\qquad\text{for $f\c X^{n+1}\to Y$ and $e\in X(n)$.} \]
\item We write
  $\lambda.(-)\c \delta\Sigmas\to \Sigmas$ and $\circ\c \Sigmas\times \Sigmas\to \Sigmas$
  for the natural
  transformations whose components come from the $\Sigma$-algebra structure
  on free $\Sigma$-algebras; here $\circ$ denotes application. In the following we will consider free
  algebras of the form $\Sigmas(X+Y)$. For simplicity, we usually keep inclusion maps implicit: Given $t_1,t_2\in X(n)$ and
  $t_1'\in Y(n)$ we write $t_1\app t_2$
  for $[\eta\comp \inl(t_1)]\circ [\eta\comp \inl(t_2)]$, and similarly
  $t_1\app t_1'$ for $[\eta\comp\inl(t_1)]\circ [\eta\comp \inr(t_1')]$
  etc., where $\inl$ and $\inr$ are the coproduct injections and $\eta\c \Id\to\Sigmas$ is the unit of the free monad $\Sigmas$.

\item Finally, define \(\pi\colon V\to \llangle X,\Sigmas(X+Y)\rrangle \) to be the map sending
   $v\in V(n)=n$ to the the natural transformation
  $\pi(v)(n)\colon X^n\to \Sigmas(X+Y)$ given by the $v$-th projection
  $X^n\to X$ followed by $\eta\comp \inl$.
\end{enumerate} 
\end{notation}
With these preparations at hand, we are now ready to phrase the small-step operational semantics of the
call-by-name $\lambda$-calculus in terms of a
$V$-pointed higher-order GSOS law of the syntax endofunctor $\Sigma X = V+\delta X + X \product X$ over the behaviour bifunctor
$B_0(X,Y)=\llangle X,Y\rrangle \times (Y+Y^X+1)$. A law of this type is given by a family of presheaf maps
\[
\begin{tikzcd}
  V+\delta(X\times \llangle X,Y\rrangle \product (Y + Y^{X} + 1) ) + (X\times \llangle X,Y \rrangle \product (Y + Y^{X} + 1))^2 \ar{d}{\rho_{X,Y}^0}  \\
 \llangle X,\Sigmas(X+Y)  \rrangle \product (\Sigmas(X+Y) + (\Sigmas(X+Y))^{X} + 1)
\end{tikzcd}
\]
dinatural in $(X,\var_X)\in V/\vcat$ and natural in $Y\in \vcat$. We let $\rho_{X,Y,n}$ denote the component of $\rho_{X,Y}$ at $n\in \fset$.

\begin{definition}[$V$-pointed higher-order GSOS law for the call-by-name
  $\lambda$-calculus]
  \label{def:lamgsos}
  \begin{align*}
    \rho^{0}_{X,Y} \c \quad & \Sigma(X \times B_0(X,Y)) & \to \quad & B_0(X, \Sigma^\star (X+Y)) \\
  \rho^{0}_{X,Y,n}(tr) = \quad & \texttt{case}~tr~\texttt{of} \\
  & v \in V(n) & \mapsto \quad & \pi(v),* & \\
  & \mathsf{\lambda}.(t, f,\_) & \mapsto
    \quad & \llangle X, \lambda.(-)\comp \eta \comp \inr \rrangle (\rho_{1}(f)),
   (\eta \comp \inr)^{X}(\rho_{2}(f)) & \\
  & (t_{1}, g, t_1') \app (t_{2}, h,\_) & \mapsto
    \quad & \lambda \vec{u}. ( g_{m}(\vec{u}) \app h_{m}(\vec{u}) ), t_1' \app t_{2} & \\
  & (t_{1}, g, k) \app (t_{2}, h,\_) & \mapsto
    \quad & \lambda \vec{u}. (g_{m}(\vec{u}) \app h_{m}(\vec{u})),\eta \comp \inr \comp k(t_{2}) & \\
  & (t_{1}, g, *) \app (t_{2}, h,\_) & \mapsto
    \quad & \lambda \vec{u}. (g_{m}(\vec{u}) \app h_{m}(\vec{u})),* &
\end{align*}
where $t\in \delta X(n)$, $f\in \delta\llangle X,Y\rrangle(n)$, $g,h\in \llangle X,Y\rrangle(n)$, $\vec{u} \in
X(m)^{n}$ for $m \in \mathbb{N}$, $k\in Y^X(n)$, $t_1, t_2\in X(n)$
and $t_1'\in Y(n)$ (we have omitted the brackets around the pairs on
the right).
\end{definition}

\begin{rem}\label{rem:tau}
By the above definition the first component 
\[ \fst\comp \rho^0_{X,Y}\c V+\delta(X\times \llangle X,Y\rrangle \product (Y + Y^{X} + 1) ) + (X\times \llangle X,Y \rrangle \product (Y + Y^{X} + 1))^2  \to 
 \llangle X,\Sigmas(X+Y)  \rrangle \]
of $\rho_{X,Y}^0$ only depends on $f$ in the clause for abstraction, and only on $g$ and $h$ in the clauses for application. Therefore $\fst\comp \rho_{X,Y}^0$ can be expressed as a composite
\[
\begin{tikzcd}
  V+\delta(X\times \llangle X,Y\rrangle \product (Y + Y^{X} + 1) ) + (X\times \llangle X,Y \rrangle \product (Y + Y^{X} + 1))^2 \ar{d}{\id+\delta(p) + p^2}  \\
  V+\delta \llangle X,Y\rrangle + \llangle X,Y \rrangle^2 \ar{d}{\tau_{X,Y}^0}  \\
 \llangle X,\Sigmas(X+Y)  \rrangle 
\end{tikzcd}
\] 
for suitable $\tau_{X,Y}^0$, where $p$ is the middle product projection.
\end{rem}
The higher-order GSOS law $\rho^0$ correctly captures the operational semantics of the call-by-name $\lambda$-calculus:

\begin{proposition}[\cite{gmstu23}]\label{prop:lambda-law}
The operational model
\begin{equation}\label{eq:operational-model-lambda-app}
 \gamma^0=\langle \gamma_1^0,\gamma_2^0\rangle \c \Lambda\to \llangle \Lambda,\Lambda\rrangle \times (\Lambda+\Lambda^\Lambda +1) 
\end{equation}
of the higher-order GSOS law $\rho^0$ satisfies the following for $n\in \fset$ and $t\in \Lambda(n)$:
\begin{enumerate}
\item\label{lem:lambda-law-1} $\gamma_1^0(t)(\vec{u}) = t[u_0,\ldots,u_{n-1}/0,\ldots,n-1]$ for  all $m\in\fset$ and $\vec{u}\in \Lambda(m)^n$.
\item\label{lem:lambda-law-2} If $t\to t'$, then $\gamma_2^0(t)=t'\in \Lambda(n)$.
\item\label{lem:lambda-law-3} If $t=\lambda x.t'$, then $\gamma_2^0(t)\in \Lambda^\Lambda(n)$ and $\gamma_2(t)(e)=t'[e]$ for all $e\in \Lambda(n)$.
\item\label{lem:lambda-law-4} Otherwise (that is, if $t$ is stuck), one has $\gamma_2^0(t)=\ast$.
\end{enumerate}
\end{proposition}

\begin{remark}\label{rem:rho0-to-rho-lambda}
In order to deal with weak similarity, we need to work with a nondeterministic version $B$ of the bifunctor $B_0(X,Y)=\llangle X,Y\rrangle \times (Y+Y^X+1)$. The nondeterminism is introduced via the pointwise powerset functor $\Pow_\star\c \vcat\to\vcat$ given by $X\mapsto \Pow\comp X$, and we put
\[ B(X,Y)=\llangle X,Y\rrangle \times \Pow_\star(Y+Y^X). \]
Note that we dropped the ``$+1$''; the reason is our intended notion of weak similarity, viz.\ the open extension of applicative similarity, which does not detect whether a term is stuck. 
We extend the law $\rho^0$ of \Cref{def:lamgsos} to a higher-order GSOS law $\rho$ of $\Sigma$ over $B$ as follows. Given $(X,\var_X)\in V/\vcat$ and $Y\in \vcat$, the first component $\fst\comp \rho_{X,Y}$ is the following composite, where $p$ is the middle product projection and $\tau_{X,Y}^0$ has been introduced in \Cref{rem:tau}:
\[
\begin{tikzcd}
\Sigma(X\times \llangle X,Y\rrangle \times \Pow_\star(Y+Y^X)) \ar[equals]{d} \\
V + \delta (X\times \llangle X,Y\rrangle \times \Pow_\star(Y+Y^X)) + (X\times \llangle  X,Y\rrangle\times \Pow(Y+Y^X))^2 \ar{d}{\id+\delta(p)+p^2} \\
V + \delta \llangle X,Y\rrangle + \llangle  X,Y\rrangle^2 \ar{d}{\tau_{X,Y}^0} \\
\llangle X,\Sigmas(X+Y)\rrangle
\end{tikzcd}
\]
For the second component $\snd\comp \rho_{X,Y}$ we need a number of auxiliary natural transformations involving the powerset functor:
\begin{align*}
\st_{A,B} &\c A\times \Pow(B)\to \Pow(A\times B), && (a,S)\mapsto \{ (a,b) : b\in S \};\\
\delta_A &\c \Pow A\times \Pow A \to \Pow(A\times A), && (S,T)\mapsto \{ (s,t) : s\in S, t\in T \};\\
\phi_{A} &\c \Pow(A+1)\to \Pow(A), && S\mapsto S\smin \{\ast\};\\
\varepsilon_A & \c \Pow(A)\to \Pow(A+1), && \emptyset\mapsto 1,\, S\mapsto S~(S\neq \emptyset);\\
\eta_A &\c A\to \Pow A, && a\mapsto \{a\}; \\
\can_{A,B,C}&\c \Pow A + \Pow B + \Pow Z \to \Pow(A+B+C) && S\mapsto S.
\end{align*}

The map $\snd\comp \rho_{X,Y}$ is the defined to be composite given as follows for $n\in\fset$. We drop subscripts of $\st,\delta,\phi,\epsilon,\eta$.
\begin{equation}\label{eq:rho0-to-rho-lambda}
\begin{tikzcd}[scale cd=.95]
\Sigma(X\times \llangle X,Y\rrangle \times \Pow_\star(Y+Y^X))(n) \ar[equals]{d} \\
V(n)+ X(n+1) \times \llangle X,Y\rrangle(n+1) \times \Pow(Y(n+1)+Y^X(n+1)) + (X(n)\times \llangle X,Y\rrangle(n)\times \Pow(Y(n)+Y^X(n)))^2 \ar{d}{\id+\id\times\id\times \varepsilon + (\id\times\id\times\varepsilon)^2}\\
V(n)+ X(n+1) \times \llangle X,Y\rrangle(n+1) \times \Pow(Y(n+1)+Y^X(n+1)+1) + (X(n)\times \llangle X,Y\rrangle(n)\times \Pow(Y(n)+Y^X(n)+1))^2 \ar{d}{\id + \st + \st^2}\\
V(n)+ \Pow(X(n+1)\times \llangle X,Y\rrangle(n+1)\times (Y(n+1)+Y^X(n+1)+1)) + (\Pow(X(n)\times \llangle X,Y\rrangle (n) \times (Y(n)+Y^X(n)+1)))^2 \ar[equals]{d}\\
V(n) + \Pow(X(n+1)\times B_0(X,Y)(n+1) ) + (\Pow(X(n)\times B_0(X,Y)(n) ))^2 \ar{d}{\eta + \id + \delta}\\
\Pow V(n)+ \Pow(X(n+1)\times B_0(X,Y)(n+1) ) + \Pow((X(n)\times B_0(X,Y)(n) )^2) \ar{d}{\can}\\
\Pow(V(n)+ X(n+1)\times B_0(X,Y)(n+1) + (X(n)\times B_0(X,Y)(n))^2) \ar[equals]{d} \\
\Pow(\Sigma(X\times B_0(X,Y))(n)) \ar{d}{\Pow \rho^0_{X,Y,n}} \\
\Pow(B_0(X,\Sigmas(X+Y))(n)) \ar[equals]{d}\\
\Pow(\llangle X,\Sigmas(X\times Y) \rrangle(n) \times (\Sigmas(X+Y)(n)+(\Sigmas (X+Y))^X(n)+1)) \ar{d}{\Pow\snd} \\
\Pow((\Sigmas(X+Y)(n)+(\Sigmas (X+Y))^X(n)+1)) \ar{d}{\phi} \\
\Pow((\Sigmas(X+Y)(n)+(\Sigmas (X+Y))^X(n)))
\end{tikzcd}
\end{equation}
\end{remark}

\begin{lemma}\label{lem:operational-model-lambda-extended}
The operational model of the higher-order GSOS law $\rho$ is given by
\[ \gamma=\langle \gamma_1,\gamma_2\rangle\c \Lambda\to \llangle \Lambda,\Lambda \rrangle \times \Pow_\star(\Lambda+\Lambda^\Lambda) \]
where $\gamma_1=\gamma_1^0$ and $\gamma_2$ is the following composite at $n\in \fset$:
\[ \gamma_2 = (\, \Lambda(n)\xto{\gamma_2^0} \Lambda(n)+\Lambda^\Lambda(n)+1 \xto{\eta} \Pow(\Lambda(n)+\Lambda^\Lambda(n)+1) \xto{\phi} \Pow(\Lambda(n)+\Lambda^\Lambda(n)) \,). \]
\end{lemma}

\begin{proof}
One only needs to show that the map $\langle \gamma_1^0,\,\phi\comp \eta\comp \gamma_2^0\rangle$ satisfies the diagram \eqref{diag:gamma} defining $\gamma$. This follows via a lengthy routine verification from the definition of $\rho$, using elementary properties of the involved maps $\st,\delta,\can,\epsilon,\phi$.  
\end{proof}

We now instantiate the data of \Cref{asm-sim} to
\begin{enumerate}
\item the functor  $\Sigma X=V+\delta X+X\times X$;
\item the functor $B(X,Y)=\llangle X,Y\rrangle \times \Pow_\star(Y+Y^X)$ of \Cref{rem:rho0-to-rho-lambda}, preordered by equality in the first component and inclusion in the second one. Its relation lifting is $\barB=F\times (\barPow_\star \comp G)$, where $F$ and $G$ are the canonical relation liftings of the bifunctors $(X,Y)\mapsto \llangle X,Y\rrangle$ and $(X,Y)\mapsto (Y+Y^X)$ and $\barPow_\star$ is the lifting of $\Pow_\star$ given by
\[ \ol{\Pow_\star}(X,R) = (\Pow_\star(X),S_R), \]
where $S_R(n)\seq \Pow(X(n)) \times \Pow(X(n))$ is the (one-sided) Egli-Milner relation induced by $R(n)\seq X(n)\times X(n)$, cf. \Cref{rem:weak-vs-strong};
\item the higher-order GSOS law $\rho$ of $\ol{\Sigma}$ over $\barB$ as described in \Cref{rem:rho0-to-rho-lambda}.
\end{enumerate}
Let us verify that this data satisfies the required properties:
\begin{lemma}
\begin{enumerate}
\item The functor $\Sigma$ preserves strong epimorphisms.
\item The functor $\barB$ is good for simulations.
\item The higher-order GSOS law $\rho$ admits a relation lifting.
\end{enumerate}
\end{lemma}

\begin{proof}
\begin{enumerate}
\item Since strong epimorphisms (i.e.~componentwise surjective natural transformations) are stable under coproducts, it suffices to show that the functors $\delta$ and $\Id\times \Id$ preserve strong epimorphisms. For the functor $\delta$ this follows from the fact that it is a left adjoint~\cite{DBLP:conf/lics/FiorePT99}. For $\Id\times \Id$ use that strong epimorphisms are stable under products, see \Cref{rem:c-props}.
\item This is shown as in the proof for \Cref{S:HO-specs} verifying the
  \Cref{asm-sim} \Cref{lem:assumptions-satisfied-ho-2}, using that all the structure involved (including relation composition $\bullet$ in $\vcat$) is just formed componentwise in $\Set$. 
\item Given a $(V,V)$-pointed relation $(X,R),p_{(X,R)})\in (V,V)/\Rel(\vcat)$ and a relation $(Y,S)\in \Rel(\vcat)$ we need to show that the map $\rho_{X,Y}$ is relation-preserving with respect to the relations on the domain $\Sigma(X\times \llangle X,Y\rrangle \times \Pow_\star(Y+Y^X))(n)$ and codomain $\llangle X, \Sigmas(X+Y)\rrangle \times \Pow(\Sigmas(X+Y)(n)+\Sigmas(X+Y)^X(n))$ obtained by applying the relation liftings $\ol{\Sigma}$, $\barPow_\star$, $F$, $G$ of the functors $\Sigma$, $\Pow_\star$, $(X,Y)\mapsto \llangle X,Y\rrangle$, $(X,Y)\mapsto Y+Y^X$. There are several cases; we follow the notation of \Cref{def:lamgsos}.
\begin{enumerate} 
\item Suppose that $\lambda.(t,f,A)$ and $\lambda.(t',f',A')\in X(n+1)\times \llangle X,Y\rrangle(n+1)\times \Pow(Y(n+1)+Y^X(n+1))$ are related. Then $\rho_{X,Y}$ sends $\lambda.(t,f,A)$ to the pair $(\llangle X, \lambda.(-)\comp \eta \comp \inr \rrangle (\rho_{1}(f)),
   (\eta \comp \inr)^{X}(\rho_{2}(f)))$ and $\lambda.(t',f',A')$ to the pair $(\llangle X, \lambda.(-)\comp \eta \comp \inr \rrangle (\rho_{1}(f')),
   (\eta \comp \inr)^{X}(\rho_{2}(f')))$. These pairs are related because $f$ and $f'$ are related in $Y^X(n+1)\seq Y(n+1)+Y^X(n+1)$ and $\rho_1$, $\rho_2$ are relation-preserving, which is easy to see by their definition.
\item Suppose that $(t_1,g,A_1)\app (t_2,h,A_2)$ and $(t_1',g',A_1')\app (t_2',h',A_2')$ are related in $(X(n)\times \llangle X,Y\rrangle(n)\times \Pow(Y(n)+Y^X(n)))^2$. Then $\fst\comp \rho_{X,Y}$ sends the two pairs to $\lambda \vec{u}. ( g_{m}(\vec{u}) \app h_{m}(\vec{u})$ and $\lambda \vec{u}. ( g'_{m}(\vec{u}) \app h'_{m}(\vec{u})$, respectively, and these are related in $\llangle X,\Sigmas(X+Y)\rrangle(n)$ because the $\Sigma$-algebra structure on $\Sigmas(X+Y)$ is relation-preserving by \Cref{prop:free-monad-lift}. For $\snd\comp \rho_{X,Y}$ we consider two subcases:
\begin{enumerate}
\item If $A_1=\emptyset$, then $\rho_{X,Y}$ sends $(t_1,g,A_1)\app (t_2,h,A_2)$ to $\emptyset\in \Pow(\Sigmas(X+Y)(n)+(\Sigmas(X+Y))^X(n))$, which is related to every element of $\Pow(\Sigmas(X+Y)(n)+(\Sigmas(X+Y))^X(n))$ by definition of the Egli-Milner relation.
\item If $A_1\neq \emptyset$, then $\rho_{X,Y}$ sends $(t_1,g,A_1)\app (t_2,h,A_2)$ to 
$\{ s \app t_2 : s\in A_1\cap Y(n) \} \cup \{ \eta\comp \inr\comp k(t_2) : k\in A_1\cap Y^X(n) \}$ and $(t_1',g',A_1')\app (t_2',h',A_2')$ to 
$\{ s' \app t_2' : s'\in A_1'\cap Y(n) \} \cup \{ \eta\comp \inr\comp k'(t_2') : k'\in A_1'\cap Y^X(n) \}$, and these two sets are clearly related by the Egli-Milner relation because $A_1,A_1'$ are related, the $\Sigma$-algebra structure on $\Sigmas(X+Y)$ is relation-preserving, and by definition of the relation lifting of $(X,Y)\mapsto Y+Y^X$.\qedhere  
\end{enumerate}
\end{enumerate}
\end{enumerate}

\end{proof}
Next we describe the weakening $\wt{\gamma}$ of the operational model \eqref{eq:operational-model-lambda-app}.
\begin{definition}\label{def:weak-operational-model-lambda}
The \emph{weak operational model} is the $B(\Lambda,-)$-coalgebra
\[ \wt{\gamma}=\langle \wt{\gamma}_1, \wt{\gamma}_2\rangle \c \Lambda\to \llangle \Lambda,\Lambda\rrangle \times \Pow_\star(\Lambda+\Lambda^\Lambda +1)   \]
given for $t\in \Lambda(n)$ by
\[
\wt{\gamma}_1(t)=\gamma_1(t) \qqand \wt{\gamma}_2(t) = \;\{ \ol{t}\in \Lambda(n) : t\To \ol{t} \}\;\cup \; \{ f\in \Lambda^\Lambda(n) : \exists \ol{t}.\,t\To \ol{t} \wedge \gamma^0_2(\ol{t})=f  \}.
\]
Here $\To$ is the reflexive transitive hull of the reduction relation $\to$.\end{definition}

\begin{lemma}\label{lem:weakening-lambda}
The coalgebra $\wt\gamma$ is a weakening of $\gamma$, cf.\ \Cref{def:weakaning-and-weak-sim}.
\end{lemma}

\begin{proof}
For each relation $(\Lambda,R)$ we need to prove that existence of a morphism $\delta$ making \eqref{eq:weak-sim-diag-lambda-1} commute is equivalent to existence of a morphism $\epsilon$ making \eqref{eq:weak-sim-diag-lambda-2} commute. Here we denote the relation $\ol{B}((\Lambda,\Lambda), (\Lambda, R))$ by $(B(\Lambda,\Lambda), E_{\Lambda,R})$.
\begin{equation}\label{eq:weak-sim-diag-lambda-1}
\begin{tikzcd}[column sep = 35]
\Lambda \ar{d}[swap]{\gamma} & R \ar{l}[swap]{\outl_R} \ar{d}{\delta} \ar{r}{\outr_R} & \Lambda \ar{d}{\wt{\gamma}} \\
\llangle \Lambda,\Lambda\rrangle \times \Pow_\star (\Lambda+\Lambda^\Lambda) & \ar{l}[swap]{\outl_{\Lambda,R}} E_{\Lambda,R} \ar{r}{\outr_{E_{\Lambda,R}}} & \llangle \Lambda,\Lambda\rrangle \times \Pow_\star (\Lambda+\Lambda^\Lambda)
\end{tikzcd}
\end{equation}
\begin{equation}\label{eq:weak-sim-diag-lambda-2}
\begin{tikzcd}[column sep = 35]
\Lambda \ar{d}[swap]{\wt\gamma} & R \ar{l}[swap]{\outl_R} \ar{d}{\varepsilon} \ar{r}{\outr_R} & \Lambda \ar{d}{\wt{\gamma}} \\
\llangle \Lambda,\Lambda\rrangle \times \Pow_\star (\Lambda+\Lambda^\Lambda) & \ar{l}[swap]{\outl_{\Lambda,R}} E_{\Lambda,R} \ar{r}{\outr_{E_{\Lambda,R}}} & \llangle \Lambda,\Lambda\rrangle\times \Pow_\star (\Lambda+\Lambda^\Lambda)
\end{tikzcd}
\end{equation}
By \Cref{prop:lambda-law} the existence of $\delta$ in \eqref{eq:weak-sim-diag-lambda-1} is equivalent to the following properties for every $n\in \fset$ and $R_n(t_1,t_2)$:
\begin{enumerate}
\item $R_m(t_1[\vec{u}], t_2[\vec{u}])$ for all $m\in \fset$ and $\vec{u}\in \Lambda(m)^n$;
\item $t_1\to t_1' \implies \exists t_2'.~ t_2\To t_2' \wedge R_n(t_1', t_2')$;
\item $t_1=\lambda x.t_1' \implies \exists t_2'.t_2\To\lambda x.t_2' \wedge \forall e\in \Lambda(n). R_n(t_1'[e/x], t_2'[e/x])$.
\end{enumerate}
Similarly, the existence of $\varepsilon$ in \eqref{eq:weak-sim-diag-lambda-2} is equivalent to the following properties for every $n\in \fset$ and $R_n(t_1,t_2)$:
\begin{enumerate}[label=(\arabic*')]
\item $R_m(t_1[\vec{u}], t_2[\vec{u}])$ for all $m\in \fset$ and $\vec{u}\in \Lambda(m)^n$;
\item $t_1\To t_1' \implies \exists t_2'.~ t_2\To t_2' \wedge R_n(t_1', t_2')$;
\item $t_1\To\lambda x.t_1' \implies \exists t_2'.t_2\To\lambda x.t_2' \wedge \forall e\in \Lambda(n). R_n(t_1'[e/x], t_2'[e/x])$.
\end{enumerate}
The conditions (1)--(3) are clearly equivalent to (1')--(3').
\end{proof}

Recall from \Cref{def:app-sim} the notion of applicative similarity and its open extension.

\begin{proposition}
Weak similarity on the operational model \eqref{eq:operational-model-lambda-app} coincides with the open extension of applicative similarity:
\[ \lesssim \;=\; \lesssim^\ap. \]
\end{proposition}

\begin{proof}
By \Cref{lem:weakening-lambda} and its proof, weak similarity is the greatest relation $\lesssim\,\seq \Lambda\times \Lambda$ such that for every $n\in \fset$ and $t_1 \lesssim_n t_2$:
\begin{enumerate}
\item $t_1[\vec{u}] \lesssim_m t_2[\vec{u}]$ for all $m\in \fset$ and $\vec{u}\in \Lambda(m)^n$;
\item $t_1\to t_1' \implies \exists t_2'.~ t_2\To t_2' \wedge t_1' \lesssim_n t_2'$;
\item $t_1=\lambda x.t_1' \implies \exists t_2'.t_2\To\lambda x.t_2' \wedge \forall e\in \Lambda(n).  t_1'[e/x] \lesssim_n t_2'[e/x]$.
\end{enumerate}

\medskip\noindent
\emph{Proof of $\lesssim \,\seq\, \lesssim^\ap$.}  Note first that $\lesssim_0\,\seq\, \Lambda(0)\times \Lambda(0)$ is an applicative simulation by the above conditions (2) and (3) for $n=0$. Therefore $\lesssim_0\,\seq\, \lesssim^\ap_0$ because $\lesssim^\ap_0$ is the greatest applicative simulation. Moreover, for $n>0$ and $t_1\lesssim_n t_2$, we have
\[ t_1[\vec{u}] \lesssim_0 t_2[\vec{u}]\quad \text{for every $\vec{u}\in \Lambda(0)^n$}\]
by condition (1), whence 
\[ t_1[\vec{u}] \lesssim^\ap_0 t_2[\vec{u}]\quad \text{for every $\vec{u}\in \Lambda(0)^n$}\]
because $\lesssim_0\,\seq\, \lesssim^\ap_0$, and so $t_1\lesssim^\ap_n t_2$. This proves $\lesssim_n \,\seq\, \lesssim^\ap_n$ for $n>0$ and thus $\lesssim\,\seq\, \lesssim^\ap$ overall.

\medskip\noindent
\emph{Proof of $\lesssim^\ap \,\seq\, \lesssim$.} Since $\lesssim$ is the greatest weak simulation, it suffices to show that $\lesssim^\ap$ is a weak simulation. Thus suppose that $n\in \fset$ and $t_1\lesssim^\ap_n t_2$; we need to verify the above conditions (1)--(3) with $\lesssim_n$ replaced by $\lesssim^\ap_n$.

 Let us first consider the case $n=0$, i.e.~$t_1\lesssim^\ap_0 t_2$.
\begin{enumerate}
\item Since $t_1$ and $t_2$ are closed terms, this condition simply states that $t_1\lesssim^\ap_m t_2$ for every $m>0$. This holds by definition of $\lesssim^\ap_m$ because $t_1[\vec{u}]=t_1 \lesssim^\ap_0 t_2=t_2[\vec{u}]$ for every $\vec{u}\in \Lambda(0)^m$. 
\item holds because $\lesssim^\ap_0$ is closed under reduction: $t_1\to t_1'$ and $t_1\lesssim^\ap_0 t_2$ implies $t_1'\lesssim^\ap_0 t_2$. Thus we can take $t_2'=t_2$.
\item holds by definition of $\lesssim^\ap_0$.
\end{enumerate} 
Now suppose that $t_1\lesssim_n^\ap t_2$ for some $n>0$: 
\begin{enumerate}
\item Let $\vec{u}=(u_0,\ldots,u_{n-1})\in \Lambda(m)^n$. If $m=0$ we have $t_1[\vec{u}]\lesssim^\ap_0 t_2[\vec{u}]$ by definition of $\lesssim^\ap_n$. If $m>0$ and   $\vec{v}\in \Lambda(0)^m$ we have
\[ t_1[\vec{u}][\vec{v}] = t_1[u_0[\vec{v}],\ldots, u_{n-1}[\vec{v}]] \lesssim^\ap_0 t_2[u_0[\vec{v}],\ldots, u_{n-1}[\vec{v}]] = t_2[\vec{u}][\vec{v}], \]
whence $t_1[\vec{u}]\lesssim^\ap_m t_2[\vec{u}]$.   
\item Suppose that $t_1\to t_1'$. Then $t_1[\vec{u}]\to t_1'[\vec{u}]$ for every $\vec{u}\in \Lambda(0)^n$ because reductions respect substitution, and moreover $t_1[\vec{u}]\lesssim^\ap_0 t_2[\vec{u}]$ because $t_1\lesssim^\ap_n t_2$. It follows that $t_1'[\vec{u}] \lesssim^\ap_0 t_2[\vec{u}]$ because $\lesssim^\ap_0$ is closed under reduction, whence $t_1'\lesssim^\ap_n t_2$. 
\item Suppose that $t_1=\lambda x.t_1'$. To show that $t_2\To \lambda x. t_2'$ for some $t_2'$, suppose the contrary. There are two cases:

\medskip\noindent\underline{Case 1:} The term $t_2$ diverges, that is, its reduction sequence $t_2\to t_2'\to t_2'' \to \cdots$ is infinite.

\medskip\noindent Choose an arbitrary $\vec{u}\in \Lambda(0)^n$. Then $t_1[\vec{u}]$ is a $\lambda$-abstraction, while $t_2[\vec{u}]$ diverges. It follows that $t_1[\vec{u}]\not\lesssim^\ap_0 t_2[\vec{u}]$, in contradiction to $t_1\lesssim^\ap_n t_2$.

\medskip\noindent\underline{Case 2:} The term $t_2$ reduces in finitely many steps to $y \app s_1 \app \cdots \app s_m$ for some variable $y\in n$ and terms $s_1,\ldots, s_m\in \Lambda(n)$.

\medskip\noindent Choose an arbitrary $\vec{u}\in \Lambda(0)^n$ such that the term $u_y$ diverges, e.g. $u_y= (\lambda x. x\app x) \app (\lambda x. x\app x)$. Then $t_1[\vec{u}]$ is a $\lambda$-abstraction while
$t_2[\vec{u}]$ diverges, again contradicting $t_1\lesssim^\ap_n t_2$.

\medskip\noindent Thus $t_2\To \lambda x.t_2'$ for $x=n$ and $t_2'\in \Lambda(n+1)$. Moreover, for every  $e\in \Lambda(n)$ and $\vec{u}\in \Lambda(0)^n$ we have
\[ t_1'[e/x][\vec{u}] = t_1'[\vec{u},e[\vec{u}]] = t_1'[\vec{u},x][e[\vec{u}]/x] \lesssim^\ap_0 t_2'[\vec{u},x][e[\vec{u}]/x] = t_2'[\vec{u},e[\vec{u}]] = t_2'[e/x][\vec{u}]   \] 
using that $t_1[\vec{u}]\lesssim^\ap_0 t_2[\vec{u}]$ by definition of $\lesssim^\ap_n$. This proves $t_1'[e/x] \lesssim^\ap_n t_2'[e/x]$.\qedhere
\end{enumerate} 
\end{proof}
Finally, let us verify that the condition of \Cref{thm:compositionality-cat} is satisfied:

\begin{proposition}
The triple $(\Lambda,\ini,\wt\gamma)$ forms a lax $\rho$-bialgebra.
\end{proposition}

\begin{proof}
We need to prove lax commutativity of the following diagram:
\[
\begin{tikzcd}[row sep=3em, scale cd=.85]
\Sigma(\Lambda) \ar{r}{\ini} \ar{d}[swap]{\Sigma\langle \id, \wt{\gamma}\rangle} & \Lambda \ar{r}{\wt{\gamma}} & \llangle \Lambda,\Lambda\rrangle \times \Pow_\star(\Lambda+\Lambda^\Lambda) \\
\Sigma(\Lambda\times \llangle \Lambda,\Lambda\rrangle \times \Pow_\star(\Lambda+\Lambda^\Lambda)) \ar{r}{\rho_{\Lambda,\Lambda}} & \llangle \Lambda, \Sigmas(\Lambda+\Lambda)\rrangle \times \Pow_\star(\Sigmas(\Lambda+\Lambda) + (\Sigmas(\Lambda+\Lambda))^\Lambda) \ar[phantom]{u}[description]{\dgeq{-90}} \ar{r}[yshift=.6em]{\llangle \Lambda, \Sigmas \nabla\rrangle \times \Pow(\Sigmas\nabla+(\Sigmas\nabla)^\Lambda)} & \llangle \Lambda, \Sigmas(\Lambda)\rrangle \times \Pow_\star(\Sigmas(\Lambda) + (\Sigmas(\Lambda))^\Lambda) \ar{u}[swap]{\llangle \Lambda,\hat\ini\rrangle\times \Pow(\hat\ini+\hat\ini^\Lambda)}  
\end{tikzcd}
\]
In the first component, the diagram strictly commutes by definition of $\fst\comp \rho_{X,Y}$. In the second component, by definition of $\snd\comp \rho_{X,Y}$, lax commutativity amounts to the assertion that the weak versions
\[
    \begin{array}{l@{\qquad}l@{\qquad}l}
      \inference[\texttt{w-app1}]{s\To s'} {s \app t \To s' \app t}
      &
        \inference[\texttt{w-app2}]{s\To \lambda x.s'}{s \app t \To s'[t/x]}
    \end{array}
\]
of the rules $\texttt{app1}$ and $\texttt{app2}$ in \eqref{eq:lambda-cbn} are sound. This is clearly the case, since $\texttt{w-app1}$ and $\texttt{w-app2}$ amount to repeated application of $\texttt{app1}$ and $\texttt{app2}$. 
\end{proof}
We thus obtain \Cref{thm:app-sim-cong} as an instance of \Cref{thm:compositionality-cat}.

\end{document}
